\documentclass[11pt]{article}

\usepackage{macros}

\bluehyperref
\usemedgeometry
\usepackage[numbers]{natbib}
\usepackage{algorithm}
\usepackage{algorithmic}
\usepackage{upgreek}
\usepackage{overpic}
\usepackage{pdfpages}

\usepackage{booktabs}
\usepackage{longtable}
\usepackage{lscape}

\usepackage{arydshln}
\usepackage{xr}

\newcommand{\df}{\textup{df}}
\newcommand{\Cov}{\textup{Cov}}
\providecommand{\vc}{\mathsf{VC}}
\newcommand{\VC}{\vc}

\newcommand{\noise}{\xi}

\usepackage{multirow}
\usepackage{tikz,forest}
\usetikzlibrary{arrows.meta}

\forestset{
    .style={
        for tree={
            base=bottom,
            child anchor=north,
            align=center,
            s sep+=1cm,
    straight edge/.style={
        edge path={\noexpand\path[\forestoption{edge},thick,-{Latex}] 
        (!u.parent anchor) -- (.child anchor);}
    },
    if n children={0}
        {tier=word, draw, thick, rectangle}
        {draw, diamond, thick, aspect=2},
    if n=1{%
        edge path={\noexpand\path[\forestoption{edge},thick,-{Latex}] 
        (!u.parent anchor) -| (.child anchor) node[pos=.2, above] {Y};}
        }{
        edge path={\noexpand\path[\forestoption{edge},thick,-{Latex}] 
        (!u.parent anchor) -| (.child anchor) node[pos=.2, above] {N};}
        }
        }
    }
}

\usepackage{authblk}

\usepackage[textsize=footnotesize]{todonotes}

\usepackage{booktabs}

\newcommand{\mcRmcX}{{\mc R_{\mc X}}}
\newcommand{\Ropt}{R^\star}
\newcommand{\mcRopt}{\mc \Ropt}
\newcommand{\hmcR}{\what{\mc{R}}}

\newcommand{\SURE}{\textnormal{SURE}}
\newcommand{\Biggmid}{\,\Bigg\vert\,}

\newcommand{\muopt}{\mu_\star}
\newcommand{\deffective}{d_{\textup{snr}}}
\newcommand{\Rest}{\what{R}}
\newcommand{\RXs}{\mc{R}_\mc{X}}
\newcommand{\ropt}{{r_\star}}

\newcommand{\hop}{h_{\textup{op}}}
\newcommand{\ave}{\textrm{ave}}

\newcommand{\AAMODE}{1}

\title{The Lifecycle of a Statistical Model}
\ifnum\AAMODE=1
    \title{The Lifecycle of a Statistical Model: \\ Model Failure Detection, Identification, and Refitting}
\fi

\author[1, 2]{Alnur Ali}
\author[1]{Maxime Cauchois}
\author[1, 2]{John C. Duchi}
\affil[1]{Department of Statistics, Stanford University}
\affil[2]{Department of Electrical Engineering, Stanford University}
\affil[ ]{\texttt{\{alnurali, maxcauch, jduchi\}@stanford.edu}}
\ifnum\AAMODE=1
    \date{November 2021}
\else
    \date{\today}
\fi

\begin{document}
\maketitle

\begin{abstract}
The statistical machine learning community has demonstrated considerable resourcefulness over the years in developing highly expressive tools for estimation, prediction, and inference.  The bedrock assumptions underlying these developments are that the data comes from a fixed population and displays little heterogeneity.  But reality is significantly more complex: statistical models now routinely fail when released into real-world systems and scientific applications, where such assumptions rarely hold.  Consequently, we pursue a different path in this paper vis-a-vis the well-worn trail of developing new methodology for estimation and prediction.  In this paper, we develop tools and theory for detecting and identifying regions of the covariate space (subpopulations) where model performance has begun to degrade, and study intervening to fix these failures through refitting.  We present empirical results with three real-world data sets---including a time series involving forecasting the incidence of COVID-19---showing that our methodology generates interpretable results, is useful for tracking model performance, and can boost model performance through refitting.  We complement these empirical results with theory proving that our methodology is minimax optimal for recovering anomalous subpopulations as well as refitting to improve accuracy in a structured normal means setting.
\end{abstract}


\section{Introduction}
\label{sec:intro}

The standard view of statistical modeling is simplistic: we fit a statistical model to the training data and evaluate its performance on test data resembling the training data \citep{HastieTiFr09,BuhlmannGe11,HastieTiWa15,Efron16,Wainwright19}.  Questionable assumptions lurk: the underlying model is correct, samples are i.i.d., labels are unambiguous, the fit model is immutable, and the population is constant.  Yet, despite its simplicity, the standard viewpoint is prevalent at all points on the spectrum from cutting-edge research to introductory teaching in statistical machine learning.  To be sure, the standard viewpoint has borne fruit: the machine learning and statistics communities have displayed extraordinary resourcefulness and creativity in developing highly expressive and flexible methodologies for estimation, prediction, and inference over the years.

Yet reality is more complex.  Practitioners now routinely release (deploy) statistical models into applications---search engines, autonomous vehicles, quantitative finance, epidemic tracking and forecasting systems, and personalized healthcare applications---where a number of new challenges arise, for example (unexpected) changes to the underlying data-generating distribution, ambiguous supervision, and situations where practitioners must intervene to fix deployed models that no longer demonstrate good performance.  Indeed, recent work \citep{TaoriDaShCaReSc19,HendrycksBaMuKaWaDoDeZhPaGuSoStGi21,HendrycksZhBaStSo21} demonstrates that standard machine learning models consistently suffer significant drops in accuracy when the test-time conditions do \textit{not} resemble the training conditions---and, moreover, even when they \textit{do}.  Importantly, the drops in accuracy persist \textit{even after} we employ various training strategies (ostensibly) encouraging good performance across changes to the data-generating distribution.

Given these challenges, we adopt a perspective in this paper that departs from the conventional viewpoint in statistical machine learning: our baseline assumption is that a deployed statistical model \textit{will} inevitably fail in the real-world.  Consequently, instead of developing a statistical model in the current paper under the assumption that the data comes from a single population, we consider the fuller \textit{lifecycle of a statistical object}.  We propose a framework for this more holistic view, delineating methodology for detecting and identifying model failures and intervening to fix them through retraining.  In our view, the literature is notably silent on such issues, forcing practitioners to develop a patchwork of bespoke and unprincipled solutions to address the challenges arising post-model deployment.  We argue that the community's focus on accuracy comes at the expense of more holistic consideration of the end-to-end lifecycle of a statistical object: model fitting, deployment, monitoring, and refitting.

To ground our discussion, we consider a supervised learning problem with covariates $X \in \mc X$ and responses $Y \in \mc Y$.  We assume access to a statistical model outputting scores $s(X, Y)$ that reflect error, i.e., $s(X, Y) < s(X', Y')$ indicates the model suffers larger error on $(X', Y')$ than on $(X, Y)$.  As an example, a standard scoring function with an estimate $\hat \mu : \mc X \to \mc Y$ of the regression function $\E(Y \mid X)$ is just the absolute residual $s(X,Y) = |Y - \hat \mu(X)|$.

In this paper, we consider the following ``one-step lookahead'' setting.  For a distribution $F$ on $\mc X \times \mc Y$ and an epoch $t=0,1,2,\ldots$, we observe a set of $m$ points $\{(X_i^{t}, Y_i^{t})\}_{i=1}^m \simiid F$ at epoch $t$ that we call the calibration set.  Test data $\{(X_i^{t+1}, Y_i^{t+1})\}_{i=1}^n$ arrives at the next epoch $t+1$, drawn independently from either $F$ or another distribution $G$ on $\mc X \times \mc Y$.  Finally, let $\mcRmcX$ denote a (potentially infinite) family of subsets of the feature space $\mc X$.  The collection of subpopulations $\mcRmcX$ may be any collection of ``nice'' subsets, e.g., one with low VC-dimension, such $d$-dimensional balls, or it may encapsulate prior knowledge \citep{ChenWuWeRaRe19}.

Our goals in this paper are to (i) detect regions $R_{\mc X} \in \mcRmcX$ with poor model performance (if they exist) at epoch $t+1$, and (ii) identify (recover) the subpopulations showing degraded model performance, by using the calibration set and the scoring function.  As a third goal, we seek to (iii) identify those subpopulations that can boost model accuracy on test data arriving at epoch $t+2$ by refitting the model.  After we review related work and give the requisite background, we make these goals precise in Sections \ref{sec:goals-detection}, \ref{sec:goals-identification}, and \ref{sec:goals-retraining}, before detailing our proposals.

\section{Background and approach}
\label{sec:bkgd}

Here we review some of the work most relevant to our approach,
giving background on conformal and predictive inference,
then highlight the methodology we develop briefly, devoting
full sections to each of the three main problems we consider:
detection of model degradation, identification of regions where
the model degrades, and model refitting.

\subsection{Related work}


Though the bulk of the work in statistics and machine learning focuses on
the pre-deployment phases of the lifecycle of a statistical object---model
fitting and inference---a growing line of work in statistics considers
tracking the outcome of a stochastic process broadly, and provides
inferential guarantees that are valid uniformly over time.  For example,
\citet{Balsubramani14a}, \citet{JohariPeWa15,JohariKoPeWa17}, and
\citet{HowardRaMcSe20,HowardRaMcSe21} use martingale theory to develop
confidence sequences (equivalently, sequential tests) that provide coverage
valid at any (stopping) time, assuming the process tails behave suitably.
These works are clearly useful in situations where the data comes from a
single population, but we argue that they are less relevant to the
post-deployment phases of the lifecycle of a statistical object, as they do
not treat the subtleties that arise when identifying anomalous
subpopulations that are responsible for model failures; in contrast, these
are major foci in the current paper.  Moreover, on a technical level, we
seek to make minimal distributional assumptions in this paper, preferring
instead to view the deployed model as a black box, which is the perspective
that practitioners must frequently take.

Conformal inference \citep{PapadopoulosPrVoGa02,VovkGaSh05,ShaferVo08,BalasubramanianHoVo14}---a useful tool for constructing predictions sets that are valid so long as the data is merely exchangeable---forms the starting point of our approach for identifying anomalous subpopulations, as conformal inference generates p-values in the event that the data is in fact exchangeable.  In particular, the recent work of \citet{CauchoisGuAlDu22} is especially relevant to our current paper, as this work provides extensions to the standard fully supervised conformal inference methodology when \textit{weak} (i.e., partial) supervision is available, which we leverage in the sequel.  Strongly supervised labels are generally unavailable in real-world predictive systems, so accommodating weak supervision is an important goal.


Finally, the long line of work on detection (see, e.g., \citep{NeillMo03,DonohoJi04,Arias-CastroDoHu05,Arias-CastroCaHeZe08,Walther10,Addario-BerryBrDeLu10,Neill12} for some recent examples), which seeks to identify anomalies in spatial data, is conceptually similar to the tack we take in the current paper, as we seek to detect and identify regions (of the covariate space) with anomalous model performance.  However, here we build off of the (important) task of detection, considering both identification and model refitting as well.

\subsection{Conformal inference and leveraging weak supervision}

As it forms the basis for our proposals to come, we review (split) conformal
inference~\cite{VovkGaSh05}.  Let us assume a calibration set $\{(X_i^0,
Y_i^0)\}_{i=1}^m \simiid F$, an independent test point $(X_{m+1}^0,
Y_{m+1}^0)$, and a scoring function $s(X,Y)$.  The usual goal in conformal
inference is to produce a prediction set $\what{C}_m : \mc X \rightrightarrows
\mc Y$ based on the $m$ calibration points satisfying, for some fixed
miscoverage level $\alpha \in (0,1)$, the marginal coverage guarantee
$\P(Y_{m+1}^0 \in \what{C}_m(X_{m+1}^0)) \geq 1-\alpha$, no matter the
underlying distribution $F$.  By exchangeability, the normalized rank
$\pi_j^0$ of the $j$th calibration point's score,
\begin{equation}
  \pi_{j}^0 \defeq \frac{1}{m+1} \sum_{i=1}^m \indic{
    s(X_i^0, Y_i^0) \leq s(X_{j}^0, Y_{j}^0)}
  + \frac{1}{m+1}, \quad j=1,\ldots,m+1,
  \label{eq:p-value}
\end{equation}
follows a uniform distribution on $\{1/(m+1),\ldots,1\}$ so long as $(X_{m+1}^0,
Y_{m+1}^0) \sim F$ and we break ties at random.  Therefore, writing
$\textrm{Quantile}(\beta; W_1,\ldots,W_{m})$ for the $\beta$-quantile of the
points $W_1,\ldots,W_{m}$ and letting
$\what{q}_m(\alpha) = $, we immediately~\cite{VovkGaSh05} have
\begin{equation*}
  \P \left( \pi_{m+1}^0 \leq \mbox{Quantile}((1 + 1/m)(1 - \alpha);
  \pi_1^0,\ldots,\pi_m^0 )\right)
  \geq 1-\alpha.
\end{equation*}
Setting $S_i = s(X_i^0, Y_i^0)$ and $\what{q}_m = \mbox{Quantile}((1 +
\frac{1}{m})(1 - \alpha); \{S_i\}_{i=1}^m)$, one may invert this normalized
rank to obtain the prediction set $\what{C}_m(x) \defeq \{y \mid s(x, y) \le
\what{q}_m\}$, which then satisfies $\P(Y_{m + 1}^0 \in
\what{C}_m(X_{m+1}^0)) \ge 1 - \alpha$ as
desired~\citep{VovkGaSh05,LeiGSRiTiWa18,RomanoPaCa19}. It is
immediate to convert the discrete uniform random variables $\pi_j^0$,
$j=1,\ldots,m+1$, to continuous uniform random variables through
randomization~\cite[e.g.][Ch.~7, Prop.~3.2]{Shorack00},
which we do without mention in the sequel.

Key to our approach is that conformal inference is really a test for
exchangeability, more precisely, that $\pi_{m+1}^0$ is a p-value for testing
whether the test point $(X_{m+1}^0, Y_{m+1}^0) \sim F$.  Recall that we seek
to detect and identify subpopulations where model performance is unusually
poor.  Then letting $\pi_j$, $j=1,\ldots,n$, denote the normalized rank of
the $j$th \emph{test} point score among the calibration set scores, the
natural approach, which we pursue, is to leverage the
conformal p-values $\pi_j$, $j=1,\ldots,n$, to check whether $(X_{m+1}^0,
Y_{m+1}^0) \sim F$: we expect test points that do not have this property to
demonstrate irregular model performance.

\subsubsection{Weak supervision and its uses in model validation}
\label{sec:weak}

A major motivation for our approach is that it extends seamlessly to weak
(or partial) supervision, where instead of observing a true response, we
observe a partial version of it, which we represent as a set of labels
containing the true response value. Such weakly supervised settings are of
growing importance in statistical machine learning~\cite{RatnerDeWuSeRe16,
  RatnerBaEhFrWuRe17, CauchoisGuAlDu22} and, in our view,
are especially important in the lifecycle of a statistical model and
its supervision.
Consider a shopping setting in which a store uses a
machine-learned model to rank items to stock, e.g., which brands of milk to
carry; a shopper typically provides only partial feedback (purchasing a
single item) rather than a ranked list of all potential items, making such
feedback both easy to collect---one observes what shoppers buy
naturally---and partial.  To formalize, let $W^0_i \subseteq \mc Y$, for
$i=1,\ldots,m+1$, denote sets of potential labels.  For some distribution
$F_\textrm{weak}$ on $\mc X \times 2^{\mc Y}$, assume that we observe weakly
supervised data $\{(X_i^0, W_i^0)\}_{i=1}^{m+1} \sim F_\textrm{weak}$
instead of (strongly) supervised data $\{(X_i^0, Y_i^0)\}_{i=1}^{m+1} \sim
F$ as before. We assume we have a scoring function $s : \mc{X} \times
\mc{Y} \to \R$ as usual.

Now for any $x \in \mc X$ and $W \subset \mc{Y}$, define the
\textit{min-score}
\begin{equation}
    s_\textrm{min}(x, W) \defeq \inf_{y \in W} s(x, y), \label{eq:weak}
\end{equation}
the most optimistic score given the partial label information.  The
min-scores $s_{\min}(X_i^0, W_i^0)$ still give rise to conformal p-values just
as before: \citet[Theorem 2]{CauchoisGuAlDu22} show that the
normalized rank $\pi_{j}^0$ of the $j$th calibration point's score
\begin{align*}
    \pi_{j}^0 & = \frac{1}{m+1} \sum_{i=1}^m \ones \Big\{ s_\textrm{min}(X_i^0, W_i^0) \leq s_\textrm{min}(X_{j}^0, W_{j}^0) \Big\} + \frac{1}{m+1}, \quad j=1,\ldots,m+1,
\end{align*}
follows a uniform distribution on $\{1/(m+1),\ldots,1\}$ so long as $(X_{m+1}^0, W_{m+1}^0) \sim F_\textrm{weak}$ (and we break ties randomly).  Therefore, we may replace the standard scores $s(X,Y)$ appearing in \eqref{eq:p-value} with the min-scores $s_\textrm{min}(X, W)$ in \eqref{eq:weak} and proceed---even with weak labels.

\subsection{Detection}
\label{sec:goals-detection}

We return to and formalize our goal of detecting newly difficult $R \in
\mcRmcX$.  Assume we have a calibration set $\{(X_i^0, Y_i^0)\}_{i=1}^m
\simiid F$, an independent test set $\{(X_i, Y_i)\}_{i=1}^n$, a scoring
function $s$, and a finite collection of subpopulations $\mcRmcX \subseteq
2^{\mc X}$ that partition $\mc X$: we wish to test which (if any) of the
regions exhibit changing performance (noting that we could take the full set
$\RXs = \{\mc{X}\}$).
In Section~\ref{sec:detection}, we show how to use certain
localized p-values, in a construction similar
to what \citet{LeiWa14} develop, to provide false discovery control for
discovered populations.
Letting $\mc{R}\opt \subset \RXs$ denote the collection of
changing (non-null) subpopulations,
in Algorithm~\ref{alg:p-filter} we show how a Benjamini-Yekutieli-type
procedure~\cite{BenjaminiYe01} provides false discovery control.
In particular, the global null hypothesis $H_0$ that
$(X_i^0, Y_i^0) \simiid F$ and $(X_j, Y_j) \simiid F$ imply the
region-based nulls
\begin{equation}
  s(X_j, Y_j) \stackrel{\textup{dist}}{=}
  s(X_i^0, Y_i^0)
  ~~ \mbox{when}~ X_i^0, X_j \in R
  \label{eqn:region-null}
\end{equation}
for $R \in \RXs$. Then we show that for a given
desired level $\alpha$, Algorithm~\ref{alg:p-filter} returns
an estimated collection of subpopulations $\hmcR$
that
control the subpopulation-level false discovery rate
\begin{equation} \label{eq:FDR}
  \textrm{FDR}(\hmcR; \mcRopt)
  \defeq \E \left[
    \frac{|\hmcR \setminus \mcRopt|}{\max\{|\hmcR|, 1\}}\right],
\end{equation}
guaranteeing that under the nulls~\eqref{eqn:region-null} we have
$\textrm{FDR}(\hmcR; \mcRopt) \le \alpha$.

\subsection{Identification}
\label{sec:goals-identification}

Often of more interest than controlling subpopulation-level false discovery
rate~\eqref{eq:FDR} is to recover the worst-performing subpopulations.  For
example, we may seek to simply interpret the subpopulations or use them to
boost model accuracy through refitting. A natural second goal is therefore
to directly identify the subpopulations showing degraded model
performance. In Section~\ref{sec:recovery}, we work in a stylized
model of this setting---based on the nulls~\eqref{eqn:region-null}---to
investigate recovery error. Under the null
$H_0$ that the distributions of the test $(X_j, Y_j)_{j=1}^n$ and validation
$(X_i^0, Y_i^0)_{j=1}^m$ are identical and exchangeable, then the $p$-values
\begin{equation}
  \label{eqn:p-value-all}
  \begin{split}
    \pi^{\textup{discrete}}_j & \defeq \frac{1}{m + 1} \sum_{i = 1}^m
    \indic{s(X_i^0, Y_i^0) \le s(X_j, Y_j)} + \frac{1}{m + 1} \\
    \pi_j
    & \defeq \pi^{\textup{discrete}}_j - \uniform\left[0,
      \frac{1}{m+1}\right]
    \end{split}
\end{equation}
are uniform on $\{\frac{1}{m+1}, \frac{2}{m+1}, \ldots, 1\}$ and $[0, 1]$,
respectively. Letting $\Phi$ denote the normal CDF, we see that under $H_0$
the Z-scores $Z_j \defeq \Phi^{-1}(\pi_j)$ are $\normal(0, 1)$.

In the identification setting, we assume that there exists a subpopulation
$R\opt \in \RXs$ corresponding to the set of $X$-space where the null fails
and leverage these Z-scores in a stylized Gaussian sequence model. Abusing
notation to set $R\opt = \{j \in [n] \mid X_j \in R\opt\}$, we
formalize identification as choosing an estimate $\Rest \subset [n]$ of this
non-null region, where we assume
\begin{equation}
  \label{eqn:structured-gaussian-sequence}
  Z_j \simiid \normal(\mu, \sigma^2) \textrm{ for} ~ j \in \Ropt,
  ~~~~
  Z_j \simiid \normal(0, \sigma^2) ~\textrm{for}~ j \not\in \Ropt
\end{equation}
for $\mu > 0$ an unknown elevated mean and $\sigma^2 > 0$ a
known variance.
In Section~\ref{sec:recovery}, we provide
sharp upper and lower bounds on the normalized
recovery error
\begin{equation}
  \frac{|\hat R \triangle \Ropt|}{|\Ropt|}, \label{eq:recovery-error}
\end{equation}
developing a regularized testing procedure that adapts
(nearly) optimally to both the size $|\Ropt|$ of the unknown set and
the unknown $\mu > 0$ representing model irregularity.


\subsection{Refitting}
\label{sec:goals-retraining}

Finally, it is natural to seek to boost model accuracy through refitting, by
identifying subpopulations with degraded performance.  We study this idea in the
same structured variant~\eqref{eqn:structured-gaussian-sequence} of the
canonical Gaussian sequence model as in the identification case. While the model
is simple relative to more sophisticated scenarios in the literature, in our
view it provides useful insights nonetheless, and it allows us to distinguish
new optimal refitting procedures from natural---but suboptimal---more classical
procedures. Modifying the notation~\eqref{eqn:structured-gaussian-sequence} to
be more evocative of a prediction model, we assume
\begin{equation}
  Y_i \mid X_i \simiid \normal(0, \sigma^2), \; i \notin R^\star, \quad \textrm{and} \quad Y_i \mid X_i \simiid \normal(\mu, \sigma^2), \; i \in R^\star, \label{eq:gaussian-data-generation-refitting}
\end{equation}
where we interpret the responses $Y_i$, $i=1,\ldots,n$, as model errors (e.g.,
residuals) that demonstrate degradation for $i \in \Ropt$.

Letting $\ones_R \in \{0, 1\}^n$ denote the vector with values $1$ for
indices $j \in R$ and 0 otherwise,
our goal then becomes to return an estimator $\what{\mu}$ close
to $\mu\subopt \defeq \mu \ones_{R\opt}$.
Our results in Section~\ref{sec:refitting}
show that if we use the identified anomalous set $\what{R}$
from Section~\ref{sec:goals-identification},
the ``refit'' estimator
\begin{equation}
  \what{\mu} \defeq
  \ave( \{ Y_i : i \in \Rest \} ) \cdot \ones_{\Rest}
  \label{eq:refit-estimator}
\end{equation}
is minimax rate-optimal for estimating $\mu_\star$ in the subpopulation
model~\eqref{eq:gaussian-data-generation-refitting}; this is in contrast to
standard maximum likelihood estimators.


\newcommand{\independent}{\perp\!\!\!\!\perp}

\section{Detection}
\label{sec:detection}

Following the plan we outline in
Sections~\ref{sec:goals-detection}--\ref{sec:goals-retraining}, we begin
with our methodology for detecting subpopulations that show degraded model
performance.  Assume we have a calibration set $\{(X_i^0, Y_i^0)\}_{i=1}^m$,
an independent test set $\{(X_i, Y_i)\}_{i=1}^n$, a scoring function $s :
\mc{X} \times \mc{Y} \to \R$ (typically fit on a training set independent of
the validation and test data), and a collection of subpopulations $\mcRmcX
\subseteq 2^{\mc X}$.

Given our goal to test the distributional equality~\eqref{eqn:region-null}
while controlling the subpopulation-level false discovery
rate~\eqref{eq:FDR}, we aggregate region-specific p-values. For $R \in
\RXs$, we define the (random) index sets
\begin{equation*}
  I(R) \defeq \left\{i \in \{1, \ldots, m\} \mid X_i^0 \in R\right\},
  ~~~
  J(R) \defeq \left\{j \in \{1, \ldots, n\} \mid X_j \in R\right\}.
\end{equation*}
Our null is that conditional on $X \in R$ we have both
\begin{equation*}
  (X_i^0, Y_i^0) \mid X_i^0 \in R \simiid F_R
  ~~ \mbox{and} ~~
  (X_j, Y_j) \mid X_j \in R \simiid F_R
\end{equation*}
for some joint law $F_R$ on $(X, Y) \mid X \in R$.
Then conditional on the
(random) index sets $I(R)$ and $J(R)$,
the values $s(X_i^0, Y_i^0)$ and $s(X_j, Y_j)$
for $i \in I(R), j \in J(R)$ are exchangeable.
Moreover, if regions $R, R' \in \RXs$ are disjoint, then
whenever $R \neq R'$ we have the independence
\begin{equation}
  \label{eqn:independence-of-regions}
  \left\{(X_i^0, Y_i^0)_{i \in I(R)}, (X_j, Y_j)_{j \in J(R)}\right\}
  \independent
  \left\{(X_i^0, Y_i^0)_{i \in I(R')}, (X_j, Y_j)_{j \in J(R')}\right\}
\end{equation}
conditional on $\{I(R), J(R), I(R'), J(R')\}$, and
moreover, if $\RXs$ partitions $\mc{X}$ so that all $R \in \RXs$ are
disjoint, then we have the
mutual independence~\eqref{eqn:independence-of-regions} conditional
on the collection $\{I(R), J(R)\}_{R \in \RXs}$ of indices.
With these distributional identities,
we consider the normalized rank of the $j$th test point,
defining
\begin{equation}
  \label{eqn:j-region-p-value}
  \pi_j(R) \defeq \frac{1}{|I(R)| + 1}
  \sum_{i \in I(R)} \indic{s(X_i^0, Y_i^0) \le s(X_j, Y_j)}
  + \frac{1}{|I(R)| + 1}
\end{equation}
for $j \in J(R)$, tacitly abusing notation to allow $\pi_j$ to represent the
continuous p-value as in the construction~\eqref{eqn:p-value-all}.  We then
have the distribution-free guarantee that $\pi_j(R) \sim \uniform[0,
  1]$ (which holds no matter $F$ by the exchangeability
of $s(X_i^0, Y_i^0)$ and $s(X_j, Y_j)$ for $i \in I(R)$, $j \in J(R)$;
see~\cite[Prop.~2, Sec.~3.2]{LeiWa14} for
a related construction). We therefore consider the regional nulls
\begin{equation*}
  H_{0,R} : \pi_j(R) \sim \uniform[0, 1]
  ~~
  \mbox{for}~ j ~ \mbox{such that}~ X_j \in R.
\end{equation*}

There are several methods to aggregate the individual $p$-values
$\{\pi_j(R)\}_{j \in J(R)}$ into valid $p$-values for $H_{0,R}$~\cite{VovkWa20,
  HeardRu18}, where we recall that
$\pi$ is valid if
$\P(\pi \le u) \le u$ for $u \in [0, 1]$.
As we wish to detect regions where
the values $\pi_j(R)$ in~\eqref{eqn:region-p-value} are large,
we use the aggregated values
\begin{equation}
  \label{eqn:region-p-value}
  \pi(R) \defeq
  2 \frac{1}{|J(R)|} \sum_{j : X_j \in R}
  (1-\pi_j(R)), \quad R \in \RXs,
\end{equation}
where the factor of $2$ guarantees validity~\citep{VovkWa20}, so
\begin{equation}
  \label{eqn:region-valid-p-value}
  \P_{H_{0,R}}(\pi(R) \le u \mid J(R), I(R)) \le u
\end{equation}
for all $u \in [0, 1]$, guaranteeing in turn that $\P_{H_{0,R}}(\pi(R) \le
u) \le u$ as desired.  With these valid $p$-values, it is natural to apply a
Benjamini-Hochberg-Yekutieli~\citep{BenjaminiHo95, BenjaminiYe01,
  BenjaminiBo14, RamdasBaWaJo19} stepwise algorithm for rejecting
regions, as we encapsulate in Algorithm~\ref{alg:p-filter}, where
we make a correction for possible dependence between the
$\pi(R)$ if the regions are not disjoint.
In the algorithm we index the regions
by $l = 1, \ldots, N$ so
$\RXs = \{R_1, \ldots, R_N\}$, and we let $\pi(R_{(1)}) \le
\pi(R_{(2)}) \le \cdots \le \pi(R_{(N)})$ be the associated order
statistics.

\begin{algorithm}
  \caption{Benjamini-Hochberg-Yekutieli procedure for detecting subpopulations}
  \label{alg:p-filter}
  \begin{algorithmic}
    \STATE {\bf input:} calibration set $\{(X_i^0, Y_i^0)\}_{i=1}^m$; test set $\{(X_i,Y_i)\}_{i=1}^n$; level $\alpha \in (0,1)$;
    \STATE ~~~~ scoring function $s : \mc X \times \mc Y \to \R$;
    subpopulations $\RXs = \{R_1,\ldots,R_N\}$
    \FOR{$R \in \RXs$} 
    \STATE compute subpopulation $p$-values
    $\pi(R)$ as in~\eqref{eqn:region-p-value}
    \ENDFOR 
    \STATE
    \textbf{sort} $p$-values into order statistics
    $\pi(R_{(1)}) \le \pi(R_{(2)})
    \le \ldots \pi(R_{(N)})$
    \IF{regions $\RXs$ are disjoint}
    \STATE \textbf{compute} rejection index
    \begin{equation*}
      k_{\max} \defeq \max \left\{
      l \in \{1,\ldots, N\} : \pi(R_{(l)}) \le
      \frac{l}{N} \alpha \right\}
    \end{equation*}
    \ELSE
    \STATE \textbf{compute} rejection index
    \begin{equation*}
      k_{\max} \defeq \max \left\{
      l \in \{1,\ldots, N\} : \pi(R_{(l)}) \le
      \frac{l}{N \sum_{i = 1}^N 1/i} \alpha \right\}
    \end{equation*}
    \ENDIF
    \STATE
    \textbf{return} set
    $\what{\mc{R}} = \{R_{(1)}, \ldots, R_{(k_{\max})}\}$ of
    anomalous subpopulations,
    where $\what{\mc{R}} = \emptyset$ if $k_{\max} = 0$
    \end{algorithmic}
\end{algorithm}

An almost immediate result is the following, which shows
that Algorithm~\ref{alg:p-filter} controls the subpopulation-level
false discovery rate at level $\alpha$.
\begin{corollary}
  \label{corollary:FDR}
  Fix $\alpha \in (0,1)$.  Let $\{(X_i^0, Y_i^0)\}_{i=1}^m \simiid F$ be a
  calibration set, $\{(X_i,Y_i)\}_{i=1}^n$ an independent test set, and
  $s : \mc X \times \mc Y \to \R$ a fixed scoring function.
  Let $\RXs = \{R_1, \ldots R_N\}$ be a collection of subpopulations
  and $\mc{R}\opt \subset \RXs$ be the collection of non-null
  populations. Then Algorithm~\ref{alg:p-filter} returns
  a collection $\what{\mc{R}}$ satisfying
  \begin{equation*}
    \textrm{FDR}(\hmcR; \mcRopt)
    \defeq \E \left[
    \frac{|\hmcR \setminus \mcRopt|}{\max\{|\hmcR|, 1\}}\right]
    \le \frac{|\mcRopt|}{N} \alpha \le \alpha.
  \end{equation*}
\end{corollary}
\begin{proof}
  In the case that the regions $R \in \RXs$ are disjoint, then the
  mutual independence guarantee~\eqref{eqn:independence-of-regions}
  conditional on the index sets $\{I(R), J(R)\}_{R \in \RXs}$ means that
  the standard Benjamini-Hochberg procedure satisfies
  \begin{equation*}
    \E\left[\frac{|\what{\mc{R}} \setminus
        \mc{R}\opt|}{\max\{|\what{\mc{R}}, 1|\}}\right]
    = \E\left[\E\left[\frac{|\what{\mc{R}} \setminus
          \mc{R}\opt|}{\max\{|\what{\mc{R}}, 1|\}}
        \Biggmid \{J(R), I(R)\}_{R \in \RXs} \right]\right]
    \le \frac{\alpha |\mc{R}\opt|}{N}
  \end{equation*}
  as an immediate consequence of, e.g., \citet[Thm.~1.2]{BenjaminiYe01}. If
  the regions are arbitrary, then the correction factor $\sum_{i = 1}^l 1/i$
  in Alg.~\ref{alg:p-filter}, coupled with the marginal
  validity~\eqref{eqn:region-valid-p-value} of $\pi(R)$, gives the
  result~\cite[Thm.~1.3]{BenjaminiYe01}.
\end{proof}




Corollary~\ref{corollary:FDR} provides a testing guarantee at the level of
regional $p$-values, which is distinct from the typical results in the
detection and two-sample testing literature, which seek to test the global
null that $(X_i^0, Y_i^0) \simiid F$ and $(X_j, Y_j) \simiid F$. In this
sense, it shares similarities to more recent work on group
filtering~\cite{DaiBa16} and the $p$-filter
procedures~\cite{RamdasBaWaJo19}, which look at group-structured testing
regimes. While it would be interesting to leverage hierarchical or more
sophisticated group structures than those Algorithm~\ref{alg:p-filter}
addresses---simply distinguishing between a disjoint partition and
non-disjoint partitions, with a potentially conservative correction factor
in the latter case~\cite{BenjaminiYe01}---this might yield substantial
additional complexity. Additionally, in the treatment of most such
hierarchical and group-structured tasks~\cite[see, e.g., page
  2797]{RamdasBaWaJo19}, one must reject ``elementary'' hypotheses (in our
context, those corresponding to initial index-specific $p$-values
$\pi_j(R)$) before rejecting a group hypothesis $H_{0,R}$; because we only
test at the region level $R$, Algorithm~\ref{alg:p-filter} can still reject
regions even if individual $p$-values $\pi_j(R)$ could not be rejected (with
a correction for multiplicity $n$), because we typically think of regions as
consisting of a fairly large number of points.



\section{Identification}
\label{sec:recovery}

We turn to issues surrounding the identification of subpopulations that show
degraded model performance.  For some downstream tasks---e.g., interpreting
the subpopulations and using them to boost model accuracy through
refitting---it may useful to identify one worst-performing population rather
than as many as possible while controlling the subpopulation-level false
discovery rate~\eqref{eq:FDR}, especially in cases where the conservativism
of Algorithm~\ref{alg:p-filter} causes a loss in power.  Consequently, we
here detail methodology to identify subpopulations showing degraded model
performance.

Our model and problem formulation are as follows. Let $\mc{R}$
be the collection of indices associated to $\RXs$, i.e.,
$R \in \RXs$ corresponds to $\{j \in [n] \mid X_j \in R\} \in \mc{R}$.
We assume there is a
subpopulation $R\opt \in \mc{R}$ of unknown size with anomalous elements,
and we wish to recover this $R\opt$.  Consider the calibration $p$-values
\begin{equation}
  \label{eqn:calibration-p-value}
  \pi_j \defeq \frac{1}{m + 1}
  \sum_{i = 1}^m \indic{s(X_i^0, Y_i^0) \le s(X_j, Y_j)}
  + \frac{1}{m + 1},
\end{equation}
defined globally rather than in the region-specific
calculation~\eqref{eqn:j-region-p-value}. We expect
that for $j \in R\opt$, these
$\pi_j$ should be superuniform (i.e.,
to stochastically dominate a uniform random variable) as our assumption is
that the predictive model is no longer as accurate over $R\opt$.  We
formalize this by letting $\mu > 0$ and $\sigma > 0$ denote an unknown
signal strength and (known) noise level, then modeling the Z-scores $Z_i
\defeq \Phi^{-1}(\pi_{i})$, $i=1,\ldots,n$, as having elevated means via
\begin{equation}
  \label{eq:gaussian-data-generation}
  Z_{i} \mid X_i \simiid
  \normal(0, \sigma^2), \; i \notin \Ropt,
  \quad \textrm{and}
  \quad Z_{i} \mid X_i \simiid \normal(\mu, \sigma^2), \; i \in \Ropt.
\end{equation}
While model~\eqref{eq:gaussian-data-generation} is a simplification because
of its independence assumptions, when the score functions $s$ are accurate
we indeed expect that $\pi_i$ are uniform for $i \not \in R\opt$, so
normality should roughly hold~\cite{LeiWa14, LeiGSRiTiWa18}. Finally, though
the independence in \eqref{eq:gaussian-data-generation} need not hold in
general, it holds conditional on the calibration set (though in
doing so, normality may fail)). Nonetheless, the model
\eqref{eq:gaussian-data-generation} represents a stylized but
theoretically and empirically tractable setting in which we may study
identification and refitting to come.

Our final assumptions concern the size and complexity of the subpopulations
of interest, and we assume $\vc(\RXs) = d$, and that $|\Ropt| = k$ for some
$k \le \frac{n}{2}$. The scaling of $k$ differs slightly from the small
values the detection literature typically assumes~\cite{DonohoJi04,
  DonohoJi08}, which in its focus on sparse and weak effects usually sets $k
\ll \sqrt{n}$.  In contrast, given our focus on tracking deployed model
performance, many sizes $k$ are of interest.

With the model \eqref{eq:gaussian-data-generation}, we present an algorithm
to control the recovery error~\eqref{eq:recovery-error} using
subpopulation-level Z-scores, $Z_R = \frac{1}{\sqrt{|R|}} \sum_{i \in R}
Z_i$.  Our identification procedure searches for the subpopulation $R \in
\mc R$ attaining the largest value of $Z_R$ subject to a carefully
calibrated penalty that ensures power is not lost at the scale of the
largest subpopulations.
We summarize the procedure, a multi-scale scan
statistic~\citep{DumbgenSp01, Arias-CastroDoHu05, RufibachWa10, Walther10,
  WaltherPe20}, in Algorithm~\ref{alg:recovery}.

\begin{algorithm}
  \caption{\label{alg:recovery}
    Multi-scale procedure for identifying subpopulations}
  \begin{algorithmic}
    \STATE {\bf Input:} collection of subpopulations $\mc R \subset
    2^{\{1, \ldots, n\}}$ with VC-dimension $d = \vc(\RXs)$;
    \STATE ~~~ base Z-scores $Z_i$, $i=1,\ldots,n$; noise level $\sigma > 0$; size penalty $C>0$

    \STATE {\bf Initialize:} Compute subpopulation-level Z-scores:
    \[
    Z_R = \frac{1}{\sqrt{|R|}} \sum_{i \in R} Z_i, \quad ~ R \in \mc{R}
    \]

    \STATE \textbf{return} penalized maximizer
    \begin{equation}
      \what{R} \in \argmax_{R \in \mc{R}} \left\{
      Z_R - C \sigma \sqrt{d \log \frac{en}{|R| \vee d}} \right\}.
      \label{eq:multiscale}
    \end{equation}
  \end{algorithmic}
\end{algorithm}

\subsection{Theory for identification}

With our assumptions and algorithm in place,
we turn to theoretical guarantees associated with
Algorithm \ref{alg:recovery} and the associated fundamental limits.
In both, we use a signal-to-noise-rescaled version of the VC-dimension,
defining
\begin{equation*}
  \deffective(\mu) \defeq \frac{d\sigma^2}{\mu^2},
\end{equation*}
and let $X_1^n$ denote the test set covariates $X_1,\ldots,X_n$.  
With these, we present with an upper bound on the recovery error that
Algorithm \ref{alg:recovery} attains. Notably, our guarantee is adaptive to
the mean $\mu$, of which Algorithm~\ref{alg:recovery} has no knowledge.
\begin{theorem}
  \label{thm:scan-recovery-error}
  Let $\mcRmcX$ be a collection of subpopulations satisfying $\VC(\mcRmcX) =
  d < \infty$. Assume the model \eqref{eq:gaussian-data-generation} and that
  $\deffective(\mu) \lesssim k$. Then there exists a
  universal constant $C$ such that Algorithm~\ref{alg:recovery} with size
  penalty $C$ returns a region $\what{R}$ such that
  \begin{align*}
    \P\left[  \frac{|\what{R} \setdiff R\opt|}{|R\opt|}
      \ge C \cdot \frac{\deffective(\mu)}{k} \left[\log \left( \frac{n}{\deffective(\mu)}\right) + d^{-1}\log \frac{1}{\delta}
        \right] \Biggmid X_1^n  \right] \le \delta.
  \end{align*}
\end{theorem}

\noindent
We present a proof in Appendix~\ref{sec:proof-of-thm-scan-recovery-error}.

Theorem~\ref{thm:scan-recovery-error} roughly says that
Algorithm~\ref{alg:recovery}'s recovery error scales as $d \log(n/k) / k$,
divided by the (squared) signal-to-noise ratio $(\mu/\sigma)^2$.  The
scaling $d \log(n/k)$ stems from the metric entropy of $\mc R$ with respect
to the Hamming metric~\citep{Haussler95}; intuitively, the scaling suggests
that recovery is hard when Algorithm \ref{alg:recovery} must consider more
subpopulations, but is easier when the size of the subpopulation of interest
$|\Ropt| = k$ is large.  Therefore, we may interpret the overall scaling of
the bound as the (log) number of subpopulations that Algorithm
\ref{alg:recovery} must consider, divided by the number of anomalous test
points $k$ and the squared signal-to-noise ratio.


We complement Theorem~\ref{thm:scan-recovery-error} with a lower bound on
the recovery error that any estimator can attain, which again relies on
$\deffective = \frac{\sigma^2}{\mu^2} d$ and relates the sample size,
VC-dimension $d$ of the collection of regions, cardinality $k$ of each
region $R$, and the signal-to-noise ratio $\frac{\mu^2}{\sigma^2}$.  For a
numerical constant $c > 0$ (whose value we do not specify but which the
proof of Theorem~\ref{thm:recovery-error-lower-bound} makes necessary), we
let
\begin{equation}
  \label{eqn:threshold-value}
  T(n, k, d, \mu, \sigma) \defeq
  \max\left\{t \in \{1, \ldots, k\} \mid
  t \le \frac{c \sigma^2}{\mu^2} (d \wedge t)
  \log \frac{n - k + t}{t} \right\},
\end{equation}
Then, as we show in the proof of the theorem to come, again
using $\deffective = \frac{\sigma^2}{\mu^2} d$, we have
\begin{equation}
  \label{eqn:threshold-bound}
  T(n, k, d, \mu, \sigma) \ge
  \begin{cases}
    k & \mbox{if~} \frac{\mu^2}{\sigma^2} \le c
    \frac{d \log(n/k)}{k} \\
    \max\left\{d, \floor{\frac{c}{2} \deffective \log \frac{n - k}{c\,
        \deffective}}
    \right\}
    &  \mbox{if}~
    c \frac{d\log(n/k)}{k} < \frac{\mu^2}{\sigma^2} \le
    c \log \frac{n-k+d}{d} \\
    \floor{(n - k) \exp\left(-\frac{1}{c} \frac{\mu^2}{\sigma^2}\right)}
    & \mbox{if} ~
    c \log \frac{n-k+d}{d} \le \frac{\mu^2}{\sigma^2}
    \le c \log\left(n-k+1 \right).
  \end{cases}
\end{equation}
Then in Appendix~\ref{sec:proof-of-thm-recovery-error-lower-bound},
we prove the following theorem.
\begin{theorem}
  \label{thm:recovery-error-lower-bound}
  Let $1 \le d \leq k \le \frac{n}{2}$ and $\mu, \sigma > 0$.
  There exists a collection of regions $\mc{R}$ satisfying
  $\VC(\mc{R}) \le 2d$ and
  $|\{i \in [n] \mid X_i \in R\}| = k$ for each $R \in \mc{R}$ such
  that, if $R\opt$ is chosen uniformly from $\mc{R}$,
  then for any estimator $\what{R}$ we have
  \begin{equation*}
    \P\left(|\what{R} \setdiff R\opt| \ge T(n, k, d, \mu, \sigma)
    \mid X_1^n\right) \ge \frac{1}{4}
  \end{equation*}
  whenever $\frac{\mu^2}{\sigma^2} \le c \log(n - k + 1)$, where
  $T(n,k,d,\mu,\sigma)$ is the threshold value~\eqref{eqn:threshold-value}.
  Additionally, there exists a collection of regions $\mc{R}$
  satisfying $\VC(\mc{R}) \le 2d$ and $|\{i \in [n] \mid X_i \in R\}| = k$
  for each $R \in \mc{R}$ such that, under the same conditions,
  \begin{equation*}
   \E\left[|\what{R} \setdiff R\opt| \mid X_1^n \right]
    \ge \frac{d}{4} \exp\left(-\frac{\mu^2}{2
      \sigma^2} \right).
  \end{equation*}
\end{theorem}

A rough calculation considering the cases in~\eqref{eqn:threshold-bound}
shows that so long as the signal-to-noise ratio (SNR) is bounded as
$\frac{\mu^2}{\sigma^2} \lesssim \log n$, then for
numerical constants $0 < c, C < \infty$, we have
\begin{equation*}
  \P\left(|\what{R} \setdiff R\opt| \ge
  c \min\left\{k, \frac{\sigma^2 d}{\mu^2} \log \frac{n}{d},
  n \exp\left(-C \frac{\mu^2}{\sigma^2}\right)\right\}\right)
  \ge \frac{1}{4}.
\end{equation*}
Notably, when the SNR satisfies $\frac{\mu^2}{\sigma^2} \gg \log n$, then a
trivial procedure that simply chooses indices with large $Z_i$ is unlikely
to make any mistakes, as $\P(|Z_i| \ge \sigma \sqrt{2 \log n}) \le
\frac{1}{n}$ when $Z_i \sim \normal(0, \sigma^2)$. In the regime that
\begin{equation*}
  \deffective(\mu) \log \frac{n}{k} \le k,
\end{equation*}
this matches the upper bound in Theorem~\ref{thm:scan-recovery-error},
showing that Algorithm~\ref{alg:recovery} is indeed optimal---even among
procedures knowing $\mu$---at least in regimes where the size of the set $k$
to be recovered is reasonably large relative to the VC-dimension of $\RXs$.
In particular, the lower bound reveals a threshold effect: (asymptotically)
perfect recovery is impossible in general if the signal-to-noise ratio
$\mu/\sigma$ is smaller than $\sqrt{d \log(n/k) / k}$, matching the
threshold that Theorem \ref{thm:scan-recovery-error} assumes.


\subsection{Related testing and recovery results}

We situate Theorems~\ref{thm:scan-recovery-error}
and~\ref{thm:recovery-error-lower-bound} by comparing them with a few
related bounds in the literature. There is substantial interest to determine
thresholds for the signal-to-noise ratio $\frac{\mu^2}{\sigma^2}$
(relative to dimension, sample size, and sparsity level) to permit detection
and estimation in the combinatorial testing, Gaussian sequence model, and
high-dimensional regression literatures, including identifying
scenarios where the thresholds differ between detection and estimation.

In parametric regression, these thresholds
are substantially different. We look at a simplified case where the
dimension and sample size are identical, leveraging
\citet{Wainwright09b}. Here we consider vectors $\beta\opt \in \R^n$ with
$k$-sparse support, letting $R\opt = \{j \mid \beta_j\opt \neq 0\}$ with
$|R\opt| = k$ denote the true support. Let the minimal signal strength $\mu
\defeq \min_{j \in R\opt} |\beta\opt_j|$, and let $X_i \simiid \normal(0,
\frac{1}{n}I_n)$, $\noise_i \simiid \normal(0, \sigma^2)$, and use the
regression model $Z_i = X_i^T \beta\opt + \noise_i$, $i = 1, \ldots, n$.
Define the error measure
\begin{equation}
  \P\left( |\what{R} \setdiff R\opt| > 0\right).
  \label{eq:wainwright-error}
\end{equation}
Then with a bit of translation for appropriate dimensionality (as we set
$X_i \simiid \normal(0, (1/n) I_n)$), \citet[Thm.~2]{Wainwright09b}
establishes numerical constants $0 < c, C < \infty$ such that recovery in
the sense of~\eqref{eq:wainwright-error} is possible when
$\frac{\mu^2}{\sigma^2} \ge C \log(n - k)$ and impossible when
$\frac{\mu^2}{\sigma^2} \le c \log \frac{n}{k}$, making the thresholds
identical (to a numerical constant) when $k = o(n)$.  The detection story,
however, is different: \citet[Proposition~1 and Theorem~1]{Arias-Castro12}
establishes (in a slightly different fixed-design model with $X_i$ fixed to
$\ltwo{X_i} = 1$) that detection---testing for the presence of a $k$-sparse
vector with minimal non-zero entry $\mu$ against an all-zeros vector
$\zeros$---has error tending to 1 or 0 when $\frac{\mu}{\sigma} k \to 0$ or
$\frac{\mu}{\sigma} k \to \infty$, respectively. With such linear
measurements, then, there is a substantial difference between detection and
estimation.

In the case of structured testing and detection problems in the
model~\eqref{eq:gaussian-data-generation}, however, detection and
identification become more similar.
In the paper perhaps most salient to our approach,
\citet{Addario-BerryBrDeLu10} focus on
the Bayes testing risk
\begin{equation*}
  p\opt(\mc{R}) \defeq \inf_{\what{\psi}}
  \left\{\P_\emptyset(\what{\psi} = 1) + \frac{1}{|\mc{R}|} \sum_{R \in \mc{R}}
  \P_R(\what{\psi} = 0)\right\}
\end{equation*}
for tests $\what{\psi}$ of $R = \emptyset$ against $R \in \mc{R}$. One
consequence of their results, roughly, follows.  Let $\mcRmcX$ be a VC-class
with $\VC(\mcRmcX) = d$, and assume $\mc{R}$ consists of sets with support
size $k$. Then \cite[Prop.~2.2 and remarks following]{Addario-BerryBrDeLu10}
shows that if $\frac{\mu}{\sigma} \ge C \sqrt{\frac{d \log n +
    \log(1/\delta)}{k}}$, then $p\opt(\mc{R}) \le \delta$.  Under an
additional symmetry condition on the sets $\mc{R}$
(see~\cite[Sec.~5]{Addario-BerryBrDeLu10}),
\citeauthor{Addario-BerryBrDeLu10}'s Theorem~6.2 (and remarks afterward)
sketch out that $p\opt(\mc{R}) \ge \half$ whenever $\mu \le c \sqrt{\frac{d
    \log\frac{n}{k}}{k}}$. Our matching upper and lower bounds in
Theorems~\ref{thm:scan-recovery-error}
and~\ref{thm:recovery-error-lower-bound} extend these results to recovery
settings, even when $\mu$ is unknown, showing that the minimax testing rate
and recovery rates essentially coincide. In the identification of
populations with altered performance, recovery may not be substantially
harder than detection.

\section{Refitting}
\label{sec:refitting}

Throughout the current paper, we argue that detecting and identifying
subpopulations showing degraded model performance is central to a number of
downstream tasks in real-world statistical systems.  
For example, detection signals that a deployed model may be working unexpectedly
and requires intervention.  
Relatedly, identification can produce subpopulations that we may interpret and
use for performance tracking.  
Additionally, it is natural to seek to use identification to boost model
performance by somehow exploiting locality; in the current section, we examine
doing so by leveraging the scan-type recovery method from Section
\ref{sec:recovery}. We consider several natural strategies for fitting locally
adaptive models, which we review below.  
Throughout, we let $\hat \P, \hat \P_m$ denote the empirical measures associated
with the training and calibration sets, respectively, and we write $\hat
\mu(X;\hat \P)$ for a model that we fit using $\hat \P$ but evaluate at $X$.

We start by reviewing a few strategies for fitting localized models and
aggregating the models together; these roughly break down into three categories.
\begin{itemize}
    \item \textbf{Pure local.}  A simple but effective strategy for exploiting
    local information is to fit separate models and invoke the best one at
    test-time, similar to the approach we describe in
    Section~\ref{sec:goals-retraining}.  Concretely, let us assume that we have
    already identified $s$ regions $\hat R_1,\ldots,\hat R_s$. Then, we may
    proceed by fitting $s$ local models $\hat \mu(\cdot; \hat \P_{\hat R_1}),
    \ldots, \hat \mu(\cdot; \hat \P_{\hat R_s})$, where $\hat \mu(\cdot; \hat
    \P_{\hat R_j}$) for $i=1,\ldots,s$ denotes a model fitted using the samples
    $\{(X_i, Y_i) : i \in \hat R_j\}$.  Given an unseen test point $(X,Y)$, we
    compute
    \[
        j_{\min} \in \underset{j=1,\ldots,s}{\argmin} \dist(\hat R_j, X),
    \]
    where $\dist(A,x) \defeq \inf_{y \in A} \| y - x \|_2$ is the usual
    point-to-set distance between $A,x$.  Then, we form $\hat \mu(X; \hat
    \P_{\hat R_{j_{\min}}})$ to make a prediction at $X$.

    \item \textbf{Aggregated local.}  Another strategy is to fit several localized models $\hat \mu(\cdot; \hat \P_{\hat R_1}), \ldots, \hat \mu(\cdot; \hat \P_{\hat R_s})$, just as in the pure local strategy, but then aggregate the predictions \citep{Schapire89,FreundSc95,Breiman96,Breiman96b,Breiman96d,Tsybakov04}.  That is, given a test point $(X,Y)$ and some carefully chosen weights $w_1,\ldots,w_s \in \R$, we form the prediction
    \[
        \sum_{i=1}^s w_i \hat \mu(X; \hat \P_{\hat R_i}).
    \]
    
    \item \textbf{Shared strength.}  A final strategy is to fit local models that share statistical strength somehow.  For example, we may fit several local models via a kind of group regularized M-estimation (common in early approaches to multi-task learning) \citep{Caruana97,JacoBaVe08,JalaliSaRuRa10}.  Alternatively, we can fit a single global model but then adapt it in a certain way to each local region, e.g., through a boosting-type procedure \citep{GaoWuBuSvSuKhShZh19}.
\end{itemize}
The pure local strategy is especially popular in practice, so we focus on it
here.  Notably, the pure local strategy above generalizes the approach that we
describe in Section~\ref{sec:goals-retraining}: when we use
Algorithm~\ref{alg:recovery} to identify a single anomalous region (so that
$s=1$) and we define the local estimator $\hat \mu$ as in
\eqref{eq:refit-estimator}, then we essentially recover the strategy from
Section~\ref{sec:goals-retraining}.  In
Theorems~\ref{thm:two-step-estimation-error}
and~\ref{thm:estimation-error-lower-bound} below, we show that this pure
local-type strategy is in fact minimax optimal in the subpopulation
model~\eqref{eq:gaussian-data-generation-refitting}.  We also go beyond the
(stylized) subpopulation model~\eqref{eq:gaussian-data-generation-refitting},
and demonstrate the pure local strategy's efficacy along with that of the other
two archetypal strategies---aggregated local and shared strength---through
a detailed empirical evaluation that follows in Section~\ref{sec:experiments}.

\subsection{Theory for refitting}
\label{sec:refitting-theory}

Working now in the idealized Gaussian sequence model, with i.i.d.~samples $(X_1,Y_1),\ldots,(X_n,Y_n)$ following the data generating process~\eqref{eq:gaussian-data-generation-refitting}, we propose a variant of Algorithm~\ref{alg:recovery} that refits the natural estimator $\hat \mu_0 = \zeros$, where the goal is now to find to a pure local estimator $\hat \mu$ that estimates the underlying mean $\muopt$ instead of simply recovering the anomalous region $\Ropt$.  We present our procedure for refitting in Algorithm~\ref{alg:refitting}, where we reuse the definitions for the quantities $\mcRmcX$, $\mc R$, $\mcRopt = \{ \Ropt \}$, $d$, $\ell$, $\mu$, and $\sigma$ from Section~\ref{sec:recovery}.

\begin{algorithm}
  \caption{\label{alg:refitting}
    Two-step multi-scale procedure for refitting}
  \begin{algorithmic}
    \STATE {\bf Input:} collection of subpopulations $\mc R \subset
    2^{\{1, \ldots, n\}}$ with VC-dimension $d = \vc(\RXs)$;
    \STATE ~~~ model errors $Y_i$, $i=1,\ldots,n$; noise level $\sigma > 0$; size penalty $C>0$

    \STATE {\bf Initialize:} Compute subpopulation-level model errors:
    \[
    Y_{R} \defeq \frac{1}{\sqrt{|R|}} \sum_{i \in R} Y_i, \quad ~ R \in \mc{R}
    \]

    \STATE Compute penalized maximizer:
    \begin{equation*}
      \Rest \in \argmax_{R \in \mc{R}} \left\{
      Y_R - C \sigma \sqrt{d \log \frac{en}{|R| \vee d}} \right\}
    \end{equation*}
    
    \STATE \textbf{return} refit estimator
     \begin{equation}
        \hat \mu = \textrm{ave}( \{ Y_i : i \in \Rest \} ) \cdot \ones_{\Rest}. \label{eq:two-step-estimator}
    \end{equation}
  \end{algorithmic}
\end{algorithm}

We use the squared $\ell_2$ error $\| \hat \mu - \muopt \|_2^2$ to measure the quality of our refit estimator, where $\muopt = \mu \cdot \ones_{\Ropt}$.  The estimator $\hat \mu_0 = \zeros$ achieves a squared $\ell_2$ error of $k \mu^2$, so it is of particular interest to determine the conditions under which the output of Algorithm~\ref{alg:refitting} improves on $\hat \mu_0$, i.e., to study when refitting can hope to beat the ``generic'' model $\hat \mu_0$.
The following result gives a bound on the error of the localized estimator \eqref{eq:two-step-estimator} holding with high probability, and delineates such conditions; the proof of the result is in Section \ref{sec:proof-of-thm-two-step-estimation-error}.

\begin{theorem} \label{thm:two-step-estimation-error}
Let $\mcRmcX$ be a collection of subpopulations satisfying $\VC(\mcRmcX) =
  d < \infty$. Define the effective dimension
  $\deffective(\mu) \defeq \frac{d\sigma^2}{\mu^2}$.  Assume the model~\eqref{eq:gaussian-data-generation-refitting} and that the underlying signal is
  strong enough that $\deffective(\mu) \lesssim k$. Then there exists a universal
  constant $C$ such that the estimator $\hat \mu$~\eqref{eq:two-step-estimator} with size penalty $C$ satisfies
  \begin{align*}
    \P\left[ \| \hat \mu - \muopt \|_2^2
      \ge C \sigma^2 \left[ d\log \left( \frac{n}{\deffective(\mu)}\right) + \log \frac{1}{\delta}
        \right] \Biggmid X_1^n  \right] \le \delta.
  \end{align*}
\end{theorem}

The rate of Theorem \ref{thm:two-step-estimation-error} essentially reflects that of Theorem \ref{thm:scan-recovery-error},  as we may (heuristically) interpret the $\ell_2$ error in our setting as quantifying the difficulty of estimating $\Ropt$ in addition to that of estimating $\mu \cdot \ones_{\Ropt}$ given $\Ropt$.  
In particular, the rate of Theorem \ref{thm:two-step-estimation-error} reveals that refitting helps when $d \sigma^2 \log \frac{n}{\deffective(\mu)} \lesssim k\mu^2$, i.e., when $\deffective(\mu) \log \frac{n}{\deffective(\mu)} \lesssim k$,  which requires the signal strength $\mu \gtrsim \sigma \sqrt{\frac{d \log \frac{n}{k}}{k}}$.  Thus, the regime when refitting is profitable coincides with the regime where detection and recovery are asymptotically achievable.  

Of course, we may ask whether the rate of Theorem~\ref{thm:two-step-estimation-error} is optimal.  
The next result provides a lower bound on the error that any estimator can achieve in the model \eqref{eq:gaussian-data-generation-refitting}, and is the analog of Theorem~\ref{thm:recovery-error-lower-bound} for refitting.
The lower bound again matches the upper bound given in Theorem \ref{thm:two-step-estimation-error} so long as $\deffective(\mu) \log \frac{n}{k} \le k$. The proof of the result is in Section \ref{sec:proof-of-thm-estimation-error-lower-bound}.

\begin{theorem} \label{thm:estimation-error-lower-bound}
 Let $1 \le d \leq k \le \frac{n}{2}$ and $\mu, \sigma > 0$.
There exists a collection of regions $\mc{R}$ satisfying $\VC(\mc{R}) \le 2d$ and
$|\{i \in [n] \mid X_i \in R\}| = k$ for each $R \in \mc{R}$ such
that,  if $\Ropt$ is chosen uniformly from $\mc{R}$, then for any estimator $\what \mu$, 
\begin{equation*}
  \E \Big[ \norm{\what{\mu} - \muopt}_2^2  \mid X_1^n \Big] 
    \ge \frac{T(n, k, d, \mu, \sigma)\mu^2}{32}
\end{equation*}
whenever $\frac{\mu^2}{\sigma^2} \le c \log(n - k + 1)$, where
$T(n,k,d,\mu,\sigma)$ is the threshold value~\eqref{eqn:threshold-value}.

Additionally, there exists another collection of regions $\mc{R}$ satisfying $\VC(\mc{R}) \le 2d$ and $|\{i \in [n] \mid X_i \in R\}| = k$ for each $R \in \mc{R}$ such that, in the same conditions, for any estimator $\what \mu$, 
\begin{equation*}
\E\left[ \norm{\what{\mu} - \muopt}_2^2 \mid X_1^n \right]
  \ge \frac{d \mu^2}{8} \exp\left(-\frac{\mu^2}{2
    \sigma^2} \right).
\end{equation*}
\end{theorem}

\subsection{Comparison with the MLE}
\label{sec:refitting-sure}

In the context of the Gaussian sequence model, a natural alternative to the
strategy we propose in Section \ref{sec:refitting-theory} is to use the maximum
likelihood estimate, which we may then tune via any model selection criterion,
e.g., Stein's unbiased risk estimate (SURE) \citep{Stein81}.  Concretely, this
estimator uses the average of the observations on the candidate support $R \in
\mc R$, zero off $R$, and chooses $R$ to minimize, e.g., the SURE
criterion---which we focus on in what follows.  To introduce the estimator, let
us write $Y = (Y_1,\ldots,Y_n)$.  Additionally, write $\bar Y_R$, for $R \in \mc
R$, to mean $(\bar Y_R)_i = \ave( \{ Y_i : i \in R \} )$ if $i \in R$, and
$(\bar Y_R)_i = 0$ if $i \notin R$.  Finally, we write $\widehat \df(\bar Y_R)$
for any unbiased estimate of the degrees of freedom of $\bar Y_R$, i.e.,
\[
    \E \big[ \widehat \df(\bar Y_R) \big] = \df(\bar Y_R) \defeq \frac{1}{\sigma^2} \sum_{i=1}^n \Cov( (\bar Y_R)_i, \, Y_i).
\]
Then, we form:
\begin{equation}
    \hat R_\textrm{SURE} \in \argmin_{R \in \mc R} \Bigg\{ \| Y - \bar Y_R \|_2^2 + 2 \sigma^2 \widehat \df(\bar Y_R) \Bigg\}, \quad \textrm{and} \quad \hat \mu_{\textrm{SURE}} = \bar Y_{\hat R_\textrm{SURE}}. \label{eq:sure-tuned-mle}
\end{equation}
In the above setting, we have that $\df(Y_R) = 1$, i.e., SURE is equivalent to
the maximum likelihood estimator.  It is interesting to compare the performance
of the natural (SURE-tuned) MLE with the localized estimator
\eqref{eq:sure-tuned-mle}.  The following result gives an error bound for the
SURE-tuned MLE \eqref{eq:sure-tuned-mle}.


\begin{lemma} \label{lem:sure-tuned-mle-estimation-error} Let $\mcRmcX$ denote a
collection of regions satisfying $\VC(\mcRmcX) = d < \infty$.  Let $\mc R$
contain only subsets of size at most $k$ in addition to the empty set. Assume
the model \eqref{eq:gaussian-data-generation-refitting} Then, the SURE-tuned MLE
$\hat \mu_\SURE$ in \eqref{eq:sure-tuned-mle} satisfies

\[
    \E \| \hat \mu_\SURE - \muopt \|_2^2 \lesssim \sigma^2 d \log(n/d).
\]
\end{lemma}

The proof of the result is in Section
\ref{sec:proof-of-sure-tuned-mle-estimation-error}.  Though studying the risk of
a SURE-tuned estimator is difficult in general, in the setting
\eqref{eq:gaussian-data-generation-refitting}, we may leverage recent results
due to~\citet{TibshiraniRo19} and~\citet{CauchoisAlDu21} that provide relatively
easy-to-use characterizations of the risk of the SURE-tuned MLE in order to
prove the result.

Theorem \ref{thm:estimation-error-lower-bound} from earlier indicates that the
rate in Lemma \ref{lem:sure-tuned-mle-estimation-error} is (slightly)
suboptimal---even though the SURE-tuned MLE has knowledge of the correct region
size $k$.  To see why, let us consider the simple situation where the collection
of regions $\mc{R} = \{ \{1\}, \{2\}, \ldots, \{1,\ldots,n\} \}$ contains all
singletons $\{i\}$ for $i \in [n]$ in addition to the full set $[n]$ itself,
with $\Ropt = [n]$ so that the underlying mean vector $\muopt$ has full support.
It follows that the SURE-tuned MLE requires the underlying signal be strong
enough so that $\mu \gtrsim \sqrt{\sigma\log (n) / n}$ to successfully recover
$\Ropt$, whereas our Algorithm \ref{alg:recovery} only requires $\mu \gtrsim
\sigma/\sqrt{n}$. This translates into an estimation error rate of $\sigma^2
\log n$ for SURE vs.~simply $\sigma^2$ for Algorithm
\ref{alg:refitting}---highlighting the importance of the penalty appearing in
both Algorithms \ref{alg:recovery} and \ref{alg:refitting}.

However, the estimator $\hat \mu_\SURE$ could still be useful, especially in
situations when the components of the underlying mean vector $\muopt$ can vary.
Indeed, let us assume that $\muopt \in \R^n$ and that
\begin{equation}
  Y_i \mid X_i \simiid \normal(0, \sigma^2), \; i \notin R^\star, \quad \textrm{and} \quad Y_i \mid X_i \simiid \normal((\muopt)_i, \sigma^2), \; i \in R^\star, \label{eq:hetero-gaussian-data-generation-refitting}
\end{equation}
which generalizes the model \eqref{eq:gaussian-data-generation-refitting} and
allows the signal over the anomalous subpopulation to vary in magnitude.  
Now for $R \in \mc R$ write $Y_R$ to mean $(Y_R)_i = Y_i$ if $i \in R$, and
$(Y_R)_i = 0$ if $i \notin R$, so that the SURE-tuned MLE in the model
\eqref{eq:hetero-gaussian-data-generation-refitting} is:
\begin{equation}
    \hat R_\textrm{SURE} \in \argmin_{R \in \mc R} \Bigg\{ \| Y - Y_R \|_2^2 + 2 \sigma^2 \widehat \df(Y_R) \Bigg\}, \quad \textrm{and} \quad \hat \mu_{\textrm{SURE}} = Y_{\hat R_\textrm{SURE}}. \label{eq:hetero-sure-tuned-mle}
\end{equation}
In the above setting, we have that $\df(Y_R) = |R|$, i.e., SURE is equivalent to
Mallows's $C_p$ \citep{Mallows73}.  The following result gives an $\ell_2$ error
bound for the SURE-tuned MLE in the more general model
\eqref{eq:hetero-sure-tuned-mle}.  The proof of the result is similar to that of
Lemma \ref{lem:sure-tuned-mle-estimation-error}, and is in Section
\ref{sec:proof-of-hetero-sure-tuned-mle-estimation-error}.

\begin{lemma} \label{lem:hetero-sure-tuned-mle-estimation-error} Let $\mcRmcX$
  denote a collection of regions satisfying $\VC(\mcRmcX) = d < \infty$.  Let
  $\mc R$ contain only subsets of size at most $k$ in addition to the empty set.
  Assume the model \eqref{eq:hetero-gaussian-data-generation-refitting} and that
  $\min \{ k, \| \muopt \|_2^2 / \sigma^2 \} \gtrsim d \log(n/d)$. Then, the
  SURE-tuned MLE $\hat \mu_\SURE$ in \eqref{eq:hetero-sure-tuned-mle} satisfies
  \[
      \E \| \hat \mu_\SURE - \muopt \|_2^2 \lesssim \min \{ k \sigma^2, \| \muopt \|_2^2 \}.
  \]
\end{lemma}



\section{Numerical examples}
\label{sec:experiments}

Finally, we turn to empirically validating our inferential methodology.  Throughout, we focus on two evaluation criteria that are important in practice.
\begin{itemize}
    \item \textit{Change in model accuracy.}  A key use for subpopulations is refitting models by leveraging subpopulation information, e.g., as we discussed in Section \ref{sec:refitting}.  Ideally, the retrained models demonstrate improvements in accuracy on the subpopulations, without degrading overall performance too much.  In what follows, we examine both global model performance, as well as local performance arising from the subpopulations our methodology and a few baselines generate.
    
    \item \textit{Interpretability.}  As practitioners frequently interpret subpopulations---with the interpretation often guiding downstream decision-making---we inspect and interpret the recovered subpopulations throughout our experiments, as a sanity check to see if they are sensible.
\end{itemize}

As we see it, the use of structure throughout our methodology, in the form of the regions $\mcRmcX$, is central.  Of course, when local variation is present in the data, structure helps with interpretability.  However, an important point is that structure is also key to improving model performance, since it works as a regularizer, i.e., trading bias for variance when estimating subpopulations.  As a result, we expect our methodology to be useful in problems with signal-to-noise ratios that are not too large.  Additionally, reflecting on the theoretical guarantees put forth over the last few sections, we can expect our method to do well when the underlying subpopulations are sizable, i.e., in the sense of having large enough (local) sample size and/or signal strength.  Finally, as is clear from the discussion we gave in Section \ref{sec:weak}, we can expect our methodology to be nonetheless useful when we have access to weak supervision.

Therefore, we consider experiments with the following three real-world data sets.  The first is a time series, where the goal is to forecast the incidence of COVID-19 at a county-level across the United States, based on just a handful of noisy features.  This is an important but difficult problem, with significant local trends, quickly changing ambient conditions, and relatively weak overall signal.  On the other hand, the second problem we consider is classifying satellite imagery by country, where we expect a clearer signal but weaker subpopulations, which is the opposite of the situation with the COVID-19 time series.  Finally, we consider a popular sentiment analysis data set, where we intentionally weaken the supervision (details below).  In each of these data sets, we investigate different strategies for retraining the model, which we take from Section \ref{sec:refitting}.

\subsection{COVID-19 forecasting}
\label{sec:experiments-covid}

As mentioned, our goal is to predict the fraction of people testing positive for COVID-19, at each of $L = 3{,}140$ United States counties over $T = 34$ weeks from January through the beginning of August in 2021, based on some demographic features that we describe later.  As a non-stationary time series, this problem naturally fits into our framework, since an a priori fixed global model of course cannot adapt to the underlying distributional changes.  Moreover, locality plays a central role: generally speaking, a fundamental challenge in epidemiological forecasting (certainly true for the current data set) is ensuring the global patterns do not ``swamp'' the local trends, i.e., developing methodology sensitive to local fluctuations.

\paragraph{Data.}  The data we use comes from the DELPHI group at Carnegie Mellon University, one of the Center for Disease Control and Prevention's five national centers of excellence \citep{ArnoldBiBrCoFaGrMaReTi21}.  For each of $t=1,\ldots,T$ weeks, and at each of $\ell=1,\ldots,L$ locations (i.e., counties), we observe a real-valued response $Y_{\ell, t} \in [0,1]$, $\ell=1,\ldots,L$, $t=1,\ldots,T$, measuring the actual fraction of people that have COVID-19.

To keep the dimensionality of the data manageable, we consider just three features, which are trailing (i.e., smoothed) averages over the past seven days.  The first feature is simply the number of COVID-19 cases per 100{,}000 people, smoothed over the week, at each county.  The second is the number of doctor vists for COVID-like symptoms, smoothed over the week, at each county.  The third is the number of people who responded to a Facebook survey indicating that they have seen COVID-like symptoms in their county, smoothed over the week.

We standardize both the features and responses so that they lie in $[0,1$], and collect the features into vectors $X_{\ell, t} \in \R^3$, $\ell=1,\ldots,L$, $t=1,\ldots,T$.  The foregoing setup is very similar to the one the DELPHI team actually uses to produce real-time COVID-19 forecasts \citep{Tibshirani20}.

\paragraph{Methods.}  Each method we consider works by taking two passes over the data.  During the first pass, each method estimates subpopulations (if needed), i.e., subsets of $\{1,\ldots,L$\}.  We perform the model fitting and forecasting steps on the second pass, potentially using the estimated regions from the first pass.  We mention that in actual practice, we do not really require the first pass, because we often use a combination of prior knowledge and additional data to identify regions.

\textit{Identifying subpopulations.}  We consider three natural baselines that we describe briefly now, and give additional details on later.  The first baseline is a pure global strategy, i.e., the first baseline does not actually compute or use any subpopulation information.  The other two baselines, as well as our method, are localized strategies.  For a fixed number of regions each having size $r \leq L$, the second baseline simply chooses $r$ points uniformly at random to form a single region at each time step.  The third baseline and our method both use locality, but in different ways.  During the first pass, both of these methods use the data at (i) time $t$ and $t+1$, for $t=1,5,9,\ldots$, to form $X_{\ell, t}$ and $Y_{\ell,t+1}$, $\ell=1,\ldots,L$, respectively, and fit a global model (described below); (ii) time $t+1$ and $t+2$ for calibration; and (iii) time $t+2$ and $t+3$ to compute the p-values, as in Section \ref{sec:bkgd}.  (To be clear, we require the data from two adjacent time steps in order to form both $X_{\ell, t}$ and $Y_{\ell,t+1}$, $\ell=1,\ldots,L$.)  The second baseline treats the $r$ points with the largest p-values (irrespective of any structure) at time $t+3$, as a single region.  On the other hand, we determine a single (hardest) region at time $t+3$ by using the output of Algorithm \ref{alg:recovery}, with $\mc R$ set to the collection of Euclidean balls centered around each county's geographic position and containing at most $r-1$ other counties.  To sum up, each method except for the first two baselines finishes the first pass with a list of estimated regions, e.g., $(\hat R_1, \ldots, \hat R_{T-3})$.

\textit{Fitting global and local models.}  The second pass works as follows.  The first baseline fits a single global model to all of the data available at times $t$, $t+1$, and $t+2$, for $t=1,5,9,\ldots$.  On the other hand, the two other baselines and our method just fit local models to the data at times $t$, $t+1$, and $t+2$.  In particular, each of these methods fits local models to the data at time $t$ and $t+1$, with the $j$th local model fit to the data belonging to region $\hat R_{t+j-1}$, for $j=1,\ldots,s$, such that $t+j-1 \leq T-3$; in our experiments, we simply fix $s=5$.  These three methods then use the data at times $t+1$ and $t+2$ to aggregate the local models together, in the ways that we describe below.  We evaluate model accuracy at time $t+3$.

Letting $\hat R_{s+1} = \{1,\ldots,L\}$, we fit both the global and local models via least absolute deviations regression, i.e., for a fixed $t$, we compute
\begin{equation} \label{eq:LAD}
    \begin{array}{ll}
        (\hat \alpha_j^{(t)}, \hat \beta_j^{(t)}) \in \underset{\alpha \in \R, \beta \in \R^p}{\argmin} & \sum_{\ell \in \hat R_{t+j-1}} \big| Y_{\ell,t+1} - ( \alpha + X_{\ell,t}^T \beta ) \big|, \quad j=1,\ldots,s+1.
    \end{array}
\end{equation}
Now let $\hat \mu^{(t)}_j(X_{\ell, t+1}) = \hat \alpha_j^{(t)} + X_{\ell, t+1}^T \hat \beta_j^{(t)}$, for $j=1,\ldots,s+1$.  Also, for a small constant $c$, let
\[
    g(z) = \log \Big( \frac{z+c}{1-z+c} \Big)
\]
denote the logit link function, which we pad by $c$ in order to avoid division by zero (we set $c = 0.01$ in our experiments).  Then, the global model makes a prediction at $X_{\ell, t+1}$ by simply forming
\begin{equation}
    g^{-1} \big( \hat \mu^{(t)}_{s+1}(X_{\ell, t+1}) \big). \notag 
\end{equation}

\textit{Aggregating local models.}  To aggregate the local models, we consider each of the three broad strategies we described earlier in Section \ref{sec:refitting}.  In particular, we consider two kinds of aggregated local strategies: linear stacking \citep{Breiman96d}, and simple averaging.  Concretely, in stacking, we let $U^{(t+1)} \in \R^{L \times s+1}$ denote a matrix of local model predictions on the data available at time $t+1$, i.e., $U^{(t+1)}_{\ell j} = \hat \mu^{(t)}_j(X_{\ell, t+1})$, for $\ell=1,\ldots,L$, $j=1,\ldots,s+1$, and obtain the weights associated with each local model at times $t+1$ and $t+2$ by solving the constrained regression problem
\begin{equation} \label{eq:stacking}
  \begin{array}{ll}
      \underset{w \in \R^{s+1}}{\minimize} & \sum_{\ell=1}^L \big( Y_{\ell, t+2} - U^{(t+1)}_{\ell \cdot} w \big)^2 \\
      \subjectto & \; w \geq 0, \; \ones^T w = 1.
  \end{array}
\end{equation}
Let $\hat w^{(t+1)} \in \R^{s+1}$ denote a solution to \eqref{eq:stacking}.  Then, we form
\begin{equation}
    g^{-1} \Big( \big\langle \hat \mu^{(t)}(X_{\ell, t+2}), \, \hat w^{(t+1)} \big\rangle \Big), \label{eq:stacking-prediction}
\end{equation}
to make an aggregate prediction at $X_{\ell, t+2}$.  Notice that stacking requires half of the available data (i.e., at times $t$ and $t+1$) to fit the local models, and the other half of the data (i.e., at times $t+1$ and $t+2$) to fit the weights associated with the local models.  Therefore, we also consider taking a simple unweighted average of the raw predictions of the local models that we fit to \textit{all} of the data available at times $t$, $t+1$, and $t+2$, before passing the average through the sigmoid, as in \eqref{eq:stacking-prediction}.

As for a shared strength strategy, we consider a multi-task learning-type approach.  Given fitted local coefficients $\smash{\hat \beta_j^{(t)}}$ for each region $j=1,\ldots,s+1$, as in \eqref{eq:LAD}, we fit an aggregate model with a fixed regularization strength $\lambda \geq 0$, by solving the following regularized least absolute deviations regression problem:
\[
    \begin{array}{ll}
        (\hat \alpha^{(t)}_\lambda, \hat \beta^{(t)}_\lambda) = \underset{\alpha \in \R, \beta \in \R^p}{\argmin} \Bigg\{ \sum_{\ell=1}^L \big| Y_{\ell,t+1} - ( \alpha + X_{\ell,t}^T \beta ) \big| + \lambda \cdot \sum_{j=1}^{s+1} \big\| \beta - \hat \beta_j^{(t)} \big\|_2^2 \Bigg\}.
    \end{array}
\]
We tune $(\hat \alpha^{(t)}_\lambda, \hat \beta^{(t)}_\lambda)$ by picking the value of $\lambda \in \Lambda = \{2^{-10}, 2^{-9}, \ldots, 2^9, 2^{10}\}$ that gives the smallest error on the data available at times $t+1$ and $t+2$.  Letting $\hat \mu^{(t)}_\lambda(X_{\ell, t+1}) = \hat \alpha^{(t)}_\lambda + X_{\ell, t+1}^T \hat \beta^{(t)}_\lambda$, the error measure we consider is the median relative absolute deviation, i.e.,
\begin{equation}
  \textrm{median} \Bigg( \Bigg\{ \frac{\big| Y_{\ell, t+2} - \hat \mu^{(t)}_\lambda(X_{\ell, t+1}) \big|}{\big| Y_{\ell, t+2} - Y_{\ell, t+1} \big|} \Bigg\}_{\ell=1}^L \Bigg). \label{eq:median-scaled-error}
\end{equation}

The error measure \eqref{eq:median-scaled-error} is, of course, robust to excessive influence from a small number of densely populated counties.  Moreover, the denominator in \eqref{eq:median-scaled-error} represents the error that a simple ``strawman'' attains, i.e., using only the response values at the previous time step to make predictions.  Therefore, we can interpret the measure \eqref{eq:median-scaled-error} as the reduction in loss relative to a simple baseline, with values closer to one indicating worse performance, and those closer to zero indicating better performance.  Naturally, we also use \eqref{eq:median-scaled-error} when reporting our numerical results, which we present next, and as our scoring function during the first pass that we described earlier.

\paragraph{Results.}  \textit{Global performance.}  We begin by looking at overall performance, i.e., the median relative absolute deviation over all locations $\ell=1,\ldots,L$, and time points $t=1,\ldots,T$, for each of the four methods we described earlier, i.e., a purely global model along with three aggregated models, which identify subpopulations in different ways: (i) according to the output of Algorithm \ref{alg:recovery}, both with and without the penalty in \eqref{eq:multiscale}; (ii) based on the counties with the largest score \eqref{eq:median-scaled-error}, i.e., irrespective of any structure; and (iii) uniformly at random.  We then combine the predictions of the local models to produce the aggregated models, by following the strategies we described above.

Table \ref{tab:covid-budget-25-pct} shows the results, when the subpopulation size $r \leq L/4$.  As we mentioned at the beginning of this section, we expect our methodology to not degrade overall performance too badly, i.e., to essentially perform on par with the global model.  Interestingly, our method actually outperforms the other methods, including the global model, for three out of the four retraining strategies.  The differences are most pronounced when using the two aggregated local strategies we described above (averaging and stacking), whereas performance is comparable when using either the shared strength or pure local strategy.  As an alternative viewpoint, Figure \ref{fig:covid-budget-25-pct-time-series} shows the median relative error at each time step, as in \eqref{eq:median-scaled-error}, when we use stacking.
We can see that our methodology has more stable performance over time.

Tables \ref{tab:covid-budget-20-pct} and \ref{tab:covid-budget-16-pct} again show the global error, for $r \leq L/5$ and $r \leq L/6$, respectively.  Of course, we do not expect our methodology to outperform the global model uniformly, for all values of the maximum region size.  Indeed, we can see from the two tables that our methodology either performs best, or comparable to the best in a few cases.  In particular, our methodology seems to work well when we use stacking or simple averaging, and is roughly on par with the other approaches when we use either the shared strength or pure local strategies.  It is worth keeping in mind that in these latter cases, the (small) differences in performance come with the benefit of interpretability, as we discuss later.  Still, the good performance of our method is slightly surprising (and encouraging), as we did not perform any tuning, e.g., of the metric or maximum size used to construct the regions that our method uses.

\textit{Local performance.}  Now we turn to briefly investigating local performance.  In the absence of any ``ground truth'' subpopulations of interest (recall this was the reason we required the first pass that we described before), we simply compare the distributions of errors \eqref{eq:median-scaled-error}, for Algorithm \ref{alg:recovery} vs.~those of the pure global benchmark, across the hardest subsets that Algorithm \ref{alg:recovery} identifies.  Of course, we expect Algorithm \ref{alg:recovery} to exhibit better local performance than the global model in this case.  We show a Q-Q plot in Figure \ref{fig:covid-budget-25-pct-local}, where we compare the quantiles of the distributions of errors (over all locations $\ell=1,\ldots,L$, and time points $t=1,\ldots,T$), for Algorithm \ref{alg:recovery} vs.~the global model.  We use stacking and set the subpopulation size $r \leq L/4$, for Algorithm \ref{alg:recovery}.  From the figure, we can indeed see that Algorithm \ref{alg:recovery} has better local performance.

\textit{Interpretability.}  Finally, we inspect and interpret a few of the regions themselves, again when $r \leq L/4$.  We show the regions that Algorithm \ref{alg:recovery} produces on the 22nd of January 2021, 29th of January 2021, 16th of April 2021, and 30th of August 2021, in Figures \ref{fig:covid-alg2-01-22-2021}, \ref{fig:covid-alg2-01-29-2021}, \ref{fig:covid-alg2-04-16-2021}, and \ref{fig:covid-alg2-07-30-2021}, respectively.  We also consider the regions that a ``naive'' baseline generates on the same days, i.e., the baseline that forms regions simply based on the counties with the highest scores.  We show these latter regions in Figures \ref{fig:covid-naive-01-22-2021}, \ref{fig:covid-naive-01-29-2021}, \ref{fig:covid-naive-04-16-2021}, and \ref{fig:covid-naive-07-30-2021}, respectively.  It is interesting to interpret the regions.  On the 22nd and 29th of January 2021---widely recognized as two weeks with the highest incidence of COVID-19 in the United States at the time---our methodology (as in Figures \ref{fig:covid-alg2-01-22-2021} and \ref{fig:covid-alg2-01-29-2021}) identifies two regions that seem to reflect the movement of the virus across the country (cf.~Figures \ref{fig:covid-true-01-22-2021} and \ref{fig:covid-true-01-29-2021}).  Of course, as we expect, the regions from Algorithm \ref{alg:recovery} are in fact structured, meaning that they do not exclusively contain only the ``hardest'' counties, which can help with interpretability.  On the other hand, the corresponding naive regions simply contain the hardest counties with no real structure present whatsoever.

On the 16th of April 2021---after several weeks of implementing precautionary measures---the state of Michigan saw a sudden spike in the incidence of COVID-19, which our methodology evidently completely captures; see Figure \ref{fig:covid-alg2-04-16-2021}.  On the other hand, the corresponding naive region (see Figure \ref{fig:covid-true-04-16-2021}) does not include the entire state of Michigan, but rather just a few of the Michigan counties with the highest incidence of COVID-19, along with counties from other states.

Finally, on the 30th of August 2021, outbreaks began to emerge throughout the country, due to the rise of the Delta variant---with Arkansas and Missouri being two of the worst states.  Again, our methodology, which we show in Figure \ref{fig:covid-alg2-04-16-2021}, completely captures these two states.

\clearpage

\begin{table}
    \centering
    \begin{tabular}[]{l llll}
      \toprule 
      & \multicolumn{4}{c}{Retraining strategy} \\
      \cmidrule{2-5}
      Subpopulation identification strategy & Averaging & Stacking & Multi-task & Pure local \\
      \midrule
      \midrule
      Algorithm \ref{alg:recovery} & \textbf{0.7862} & \textbf{0.7885} & 0.8032 & \textbf{0.7859} \\
      Algorithm \ref{alg:recovery}, unpenalized & 0.7892 & 0.8026 & 0.8025 & 0.7968 \\
      Hardest points & 0.8231 & 0.8813 & 0.8070 & 0.8022 \\
      Uniformly at random & 0.8240 & 0.8170 & 0.8110 & 0.8180 \\
      Pure global & 0.7909 & 0.7909 & \textbf{0.7909} & 0.7909 \\
      \bottomrule
    \end{tabular}
    \caption{The median relative absolute deviation over all locations $\ell=1,\ldots,L$, and time points $t=1,\ldots,T$, when the subpopulation size $r \leq L/4$.  We highlight the best (i.e., lowest) error, for each retraining strategy, in bold.}
  \label{tab:covid-budget-25-pct}
\end{table}


\begin{table}
    \centering
    \begin{tabular}[]{l llll}
      \toprule 
      & \multicolumn{4}{c}{Retraining strategy} \\
      \cmidrule{2-5}
      Subpopulation identification strategy & Averaging & Stacking & Multi-task & Pure local \\
      \midrule
      \midrule
      Algorithm \ref{alg:recovery} & 0.7955 & 0.7933 & 0.8044 & 0.8058 \\
      Algorithm \ref{alg:recovery}, unpenalized & 0.7951 & \textbf{0.7859} & 0.8097 & 0.8048 \\
      Hardest points & 0.8289 & 0.8832 & 0.8068 & 0.8047 \\
      Uniformly at random & 0.8279 & 0.8091 & 0.8096 & 0.8190 \\
      Pure global & \textbf{0.7909} & 0.7909 & \textbf{0.7909} & \textbf{0.7909} \\
      \bottomrule
    \end{tabular}
    \caption{The median relative absolute deviation over all locations $\ell=1,\ldots,L$, and time points $t=1,\ldots,T$, when the subpopulation size $r \leq L/5$.  We highlight the best (i.e., lowest) error, for each retraining strategy, in bold.}
  \label{tab:covid-budget-20-pct}
\end{table}

\begin{table}
  \centering
  \begin{tabular}[]{l llll}
    \toprule 
    & \multicolumn{4}{c}{Retraining strategy} \\
    \cmidrule{2-5}
    Subpopulation identification strategy & Averaging & Stacking & Multi-task & Pure local \\
    \midrule
    \midrule
    Algorithm \ref{alg:recovery} & 0.8110 & \textbf{0.7907} & 0.8072 & 0.8336 \\
    Algorithm \ref{alg:recovery}, unpenalized & 0.8124 & 0.7923 & 0.8056 & 0.8315 \\
    Hardest points & 0.8378 & 0.8525 & 0.8017 & 0.8340 \\
    Uniformly at random & 0.8256 & 0.8092 & 0.8126 & 0.8187 \\
    Pure global & \textbf{0.7909} & 0.7909 & \textbf{0.7909} & \textbf{0.7909} \\
    \bottomrule
  \end{tabular}
  \caption{The median relative absolute deviation over all locations $\ell=1,\ldots,L$, and time points $t=1,\ldots,T$, when the subpopulation size $r \leq L/6$.  We highlight the best (i.e., lowest) error, for each retraining strategy, in bold.}
\label{tab:covid-budget-16-pct}
\end{table}
\clearpage

\begin{figure}
  \centering
  \includegraphics[width=0.9\linewidth]{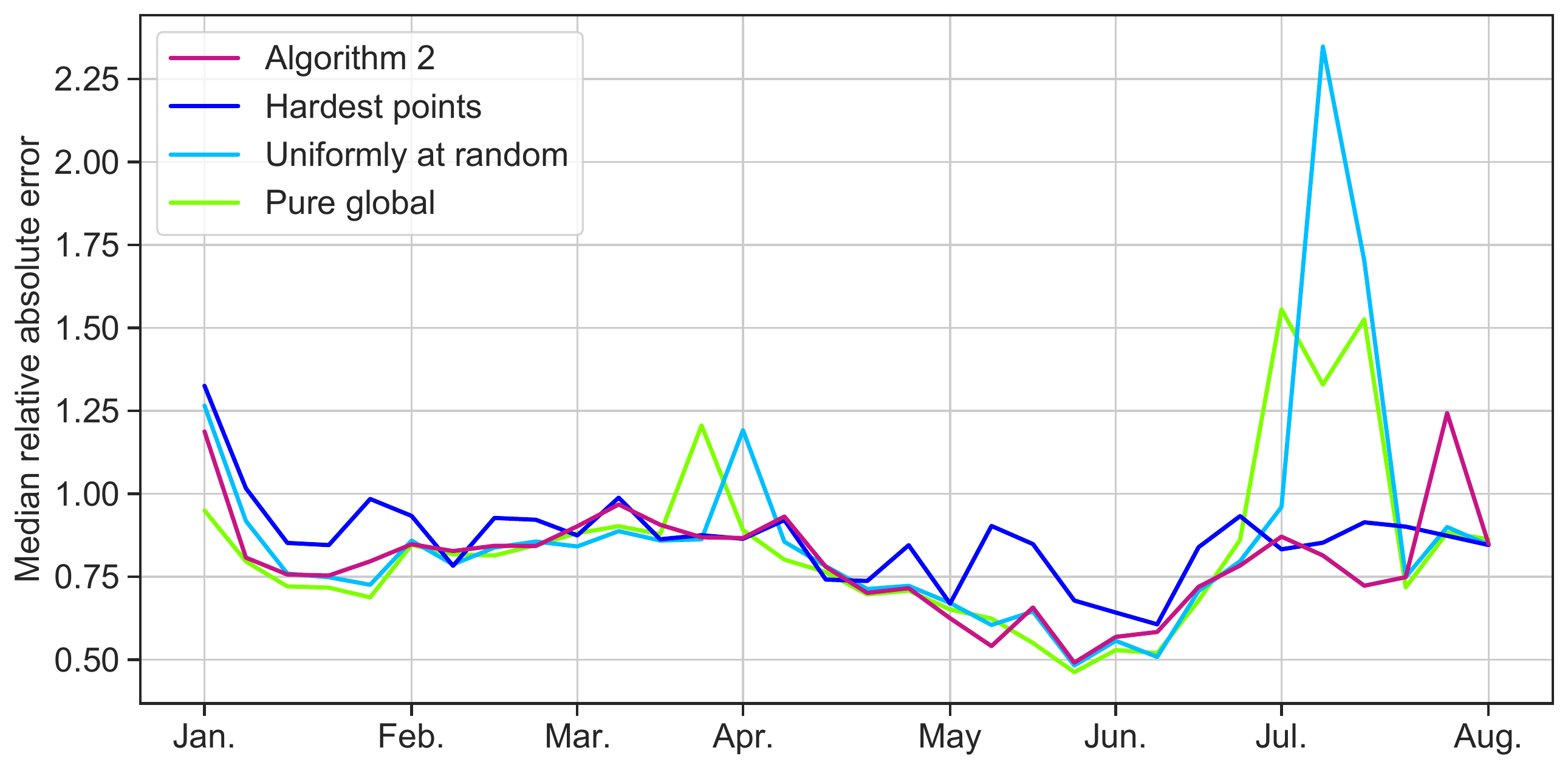}
  \caption{The median relative absolute deviation over all locations $\ell=1,\ldots,L$, at each time point $t=1,\ldots,T$, as in \eqref{eq:median-scaled-error}, when we use stacking and the subpopulation size $r \leq L/4$.}
  \label{fig:covid-budget-25-pct-time-series}
\end{figure}

\begin{figure}
  \centering
  \includegraphics[width=0.67\linewidth]{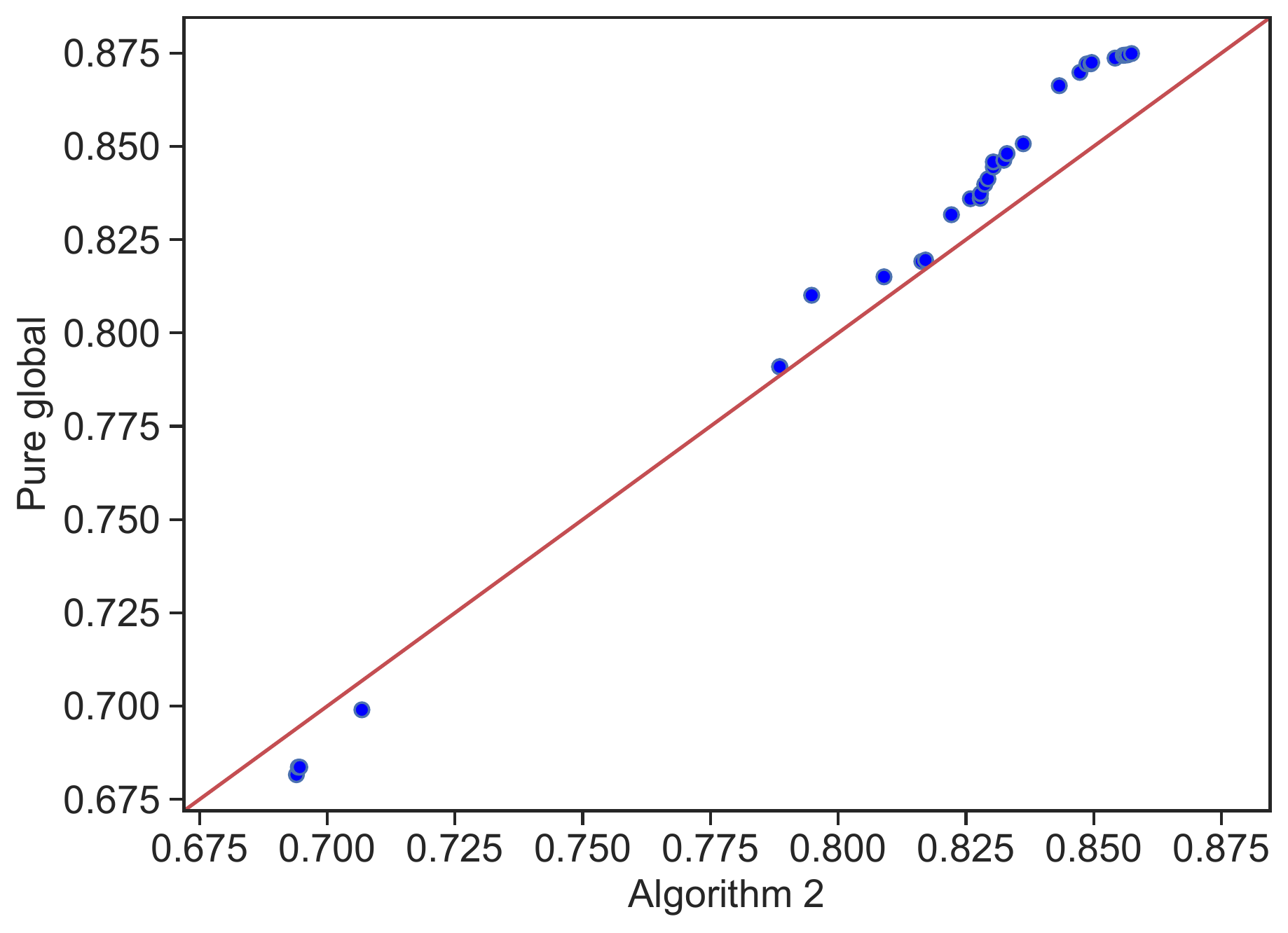}
  \caption{Q-Q plot comparing the distribution of Algorithm \ref{alg:recovery}'s median relative absolute deviation (over all locations $\ell=1,\ldots,L$, and time points $t=1,\ldots,T$) on the hardest regions it identifies, vs.~those of the pure global model on the same regions.  We use stacking and set the subpopulation size $r \leq L/4$, for Algorithm \ref{alg:recovery}.}
  \label{fig:covid-budget-25-pct-local}
\end{figure}

\clearpage
\begin{figure}
    \centering
    \includegraphics[width=0.9\linewidth]{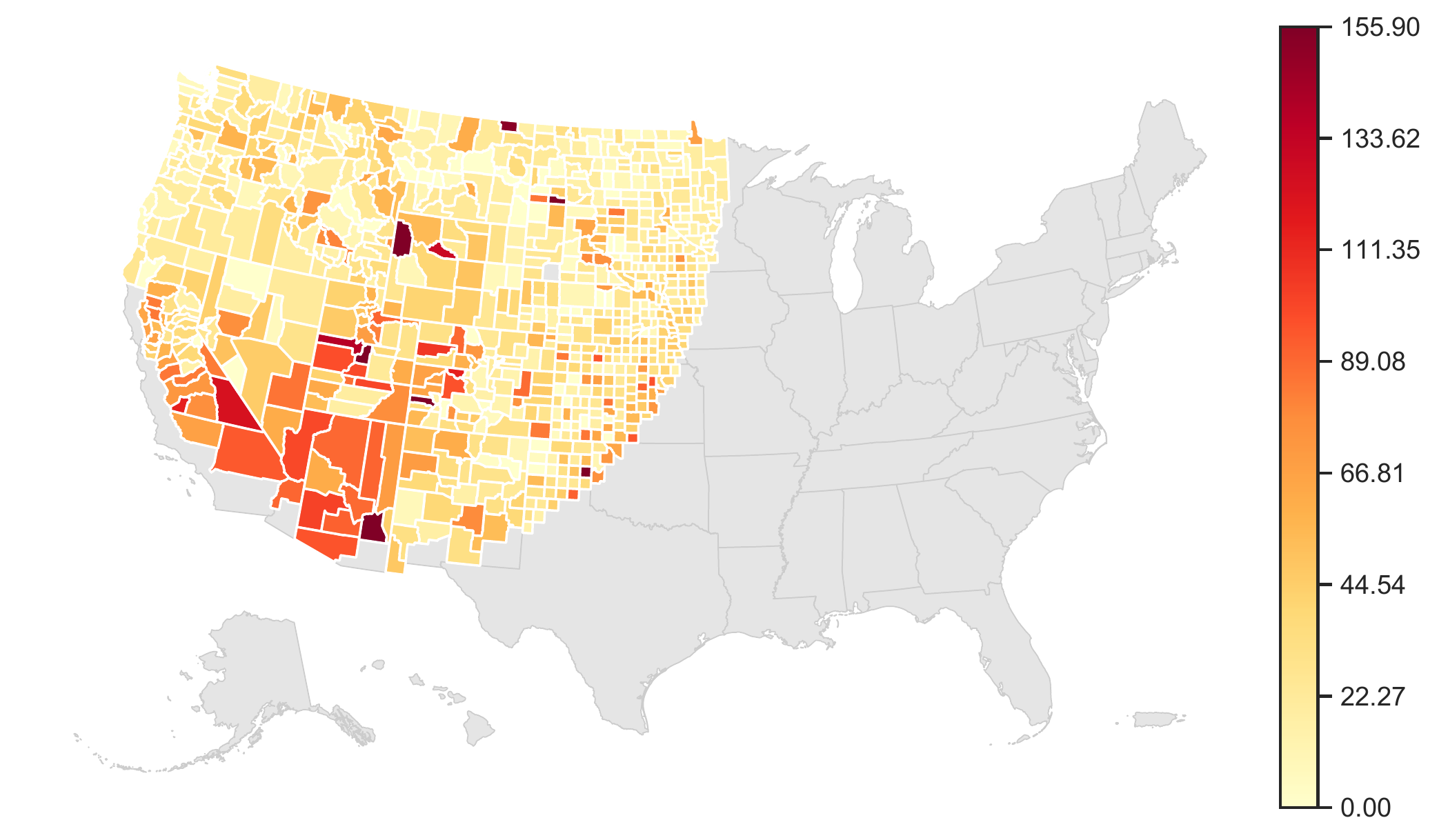}
    \caption{The hardest region that Algorithm \ref{alg:recovery} produces on the 22nd of January 2021, when we require the region size $r \leq L/4$.  We color the counties according to the true number of COVID-19 cases per 100{,}000 people, smoothed over the previous week.}
    \label{fig:covid-alg2-01-22-2021}
\end{figure}

\begin{figure}
  \centering
  \includegraphics[width=0.9\linewidth]{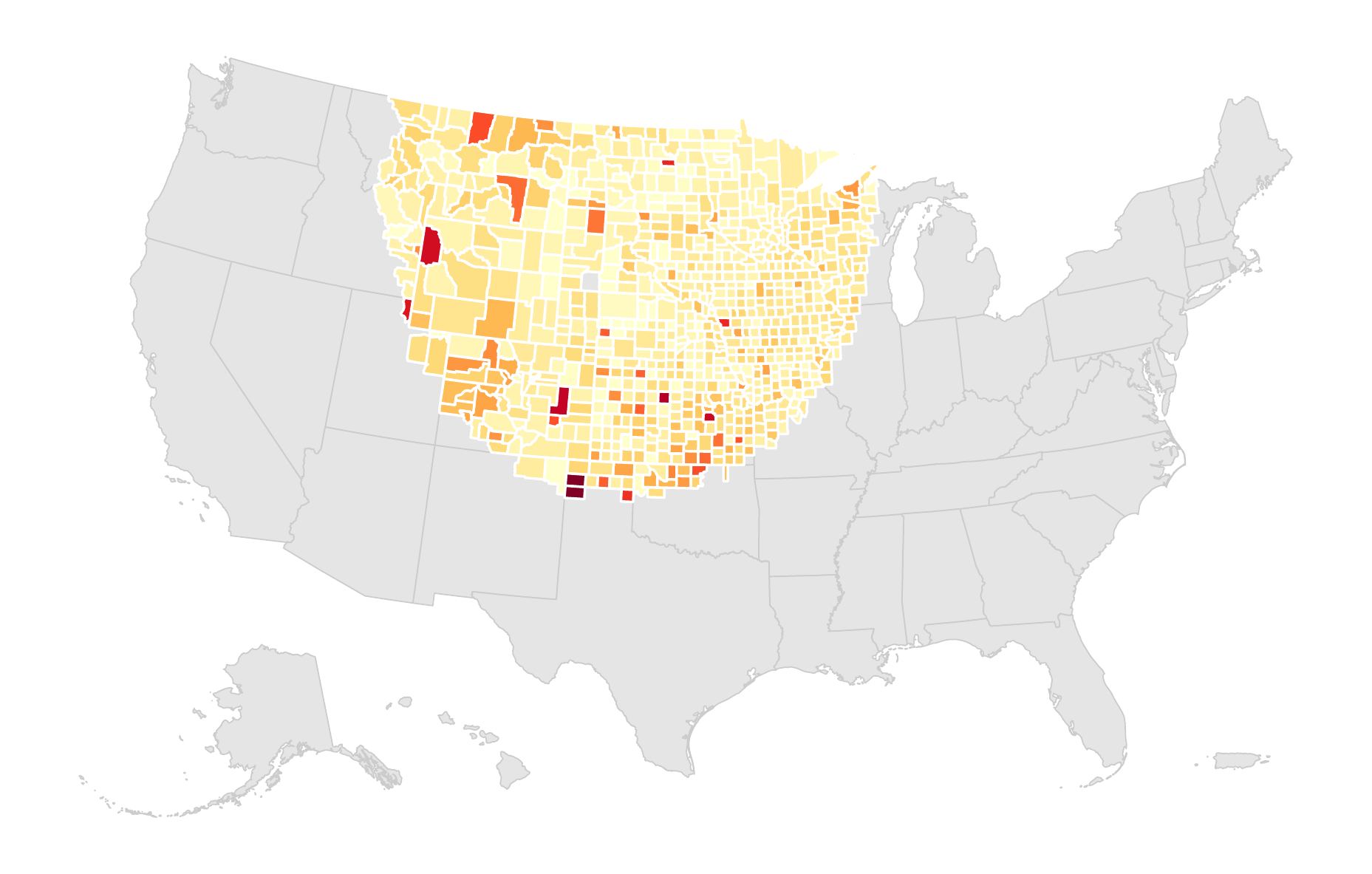}
  \caption{The hardest region that Algorithm \ref{alg:recovery} produces on the 29th of January 2021, when we require the region size $r \leq L/4$.  We color the counties according to the true number of COVID-19 cases per 100{,}000 people, smoothed over the previous week.}
  \label{fig:covid-alg2-01-29-2021}
\end{figure}
\clearpage

\clearpage
\begin{figure}
    \centering
    \includegraphics[width=0.9\linewidth]{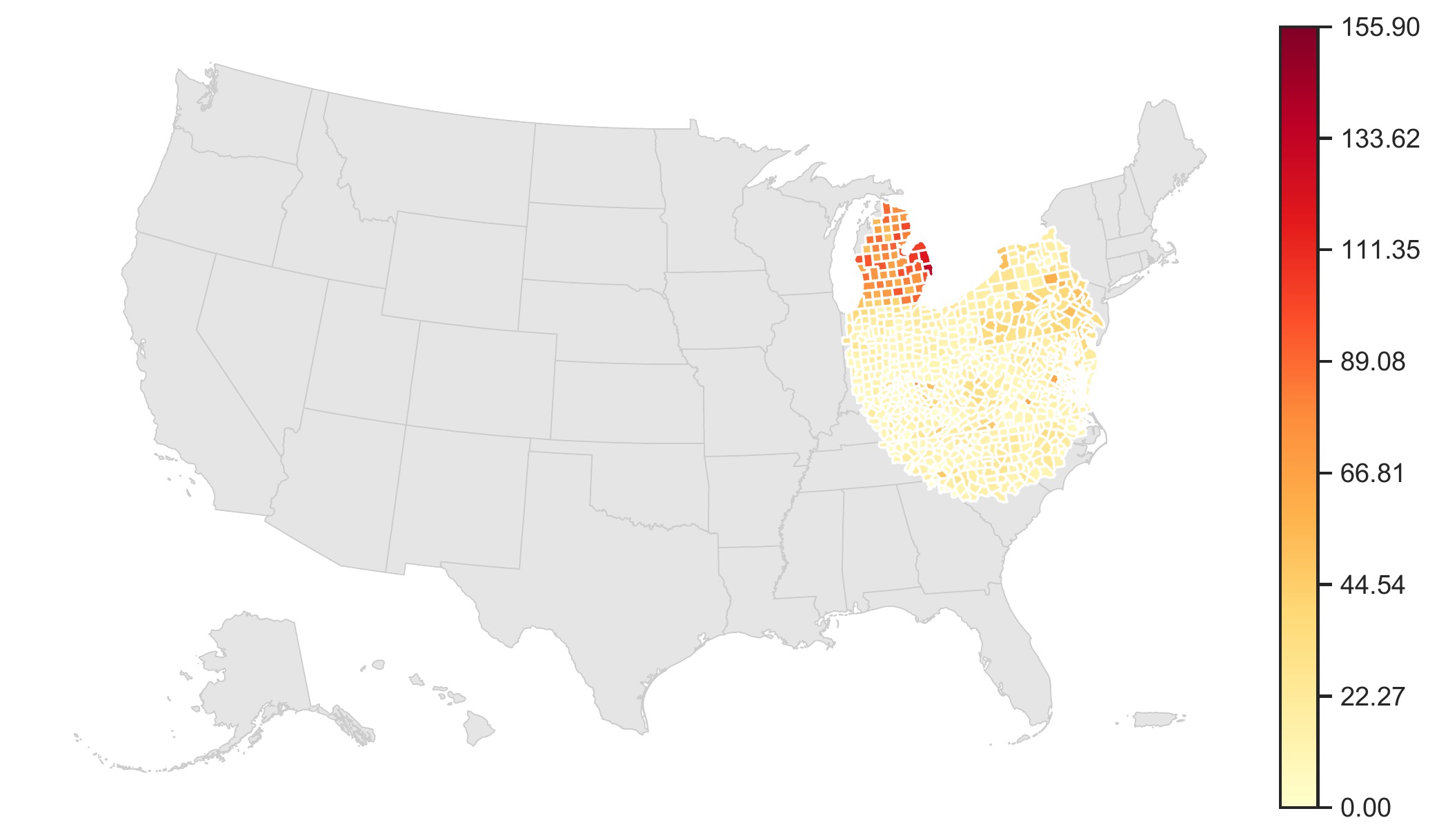}
    \caption{The hardest region that Algorithm \ref{alg:recovery} produces on the 16th of April 2021, when we require the region size $r \leq L/4$.  We color the counties according to the true number of COVID-19 cases per 100{,}000 people, smoothed over the previous week.}
    \label{fig:covid-alg2-04-16-2021}
\end{figure}

\begin{figure}
  \centering
  \includegraphics[width=0.9\linewidth]{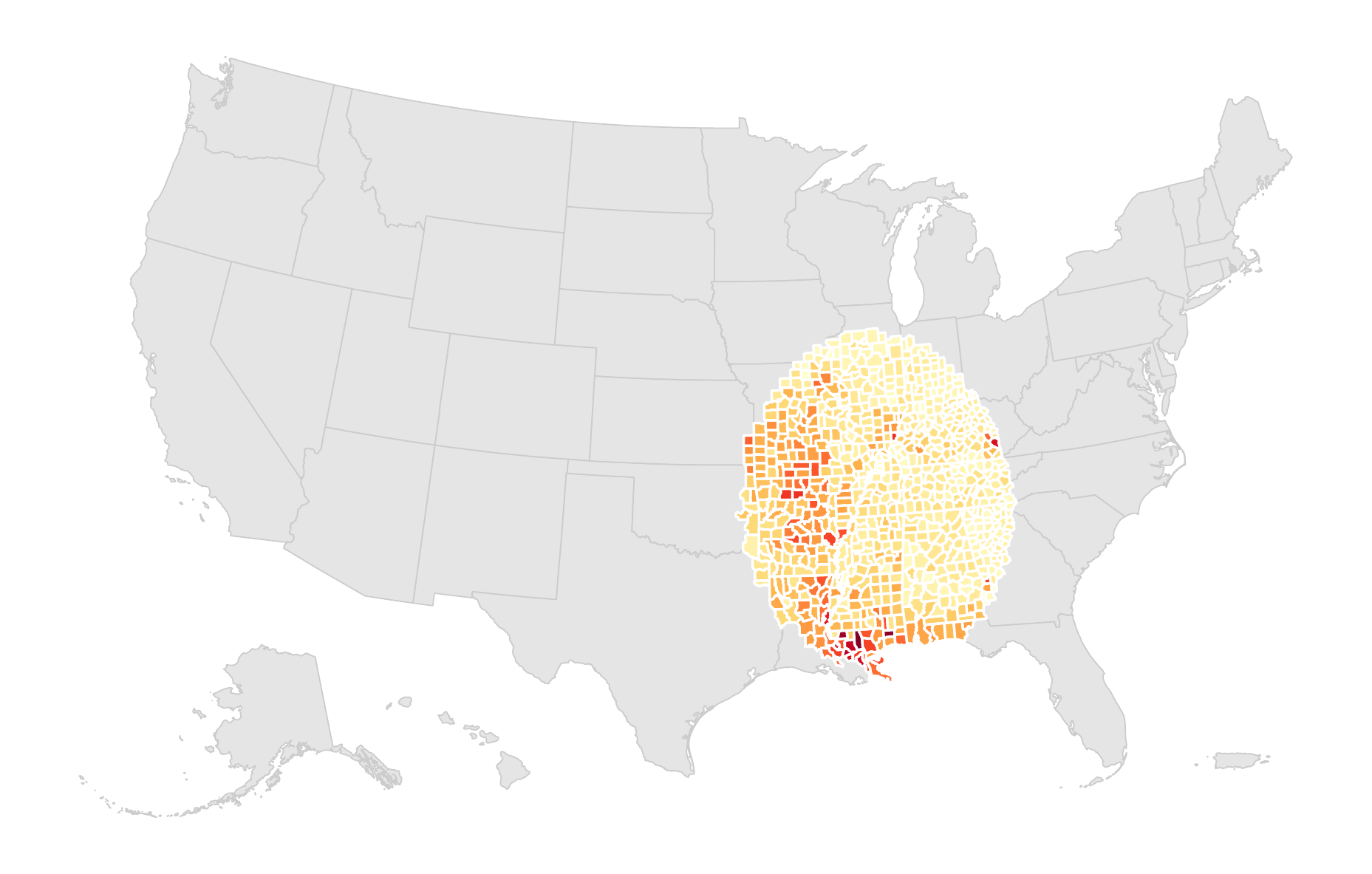}
  \caption{The hardest region that Algorithm \ref{alg:recovery} produces on the 30th of August 2021, when we require the region size $r \leq L/4$.  We color the counties according to the true number of COVID-19 cases per 100{,}000 people, smoothed over the previous week.}
  \label{fig:covid-alg2-07-30-2021}
\end{figure}
\clearpage

\clearpage
\begin{figure}
    \centering
    \includegraphics[width=0.9\linewidth]{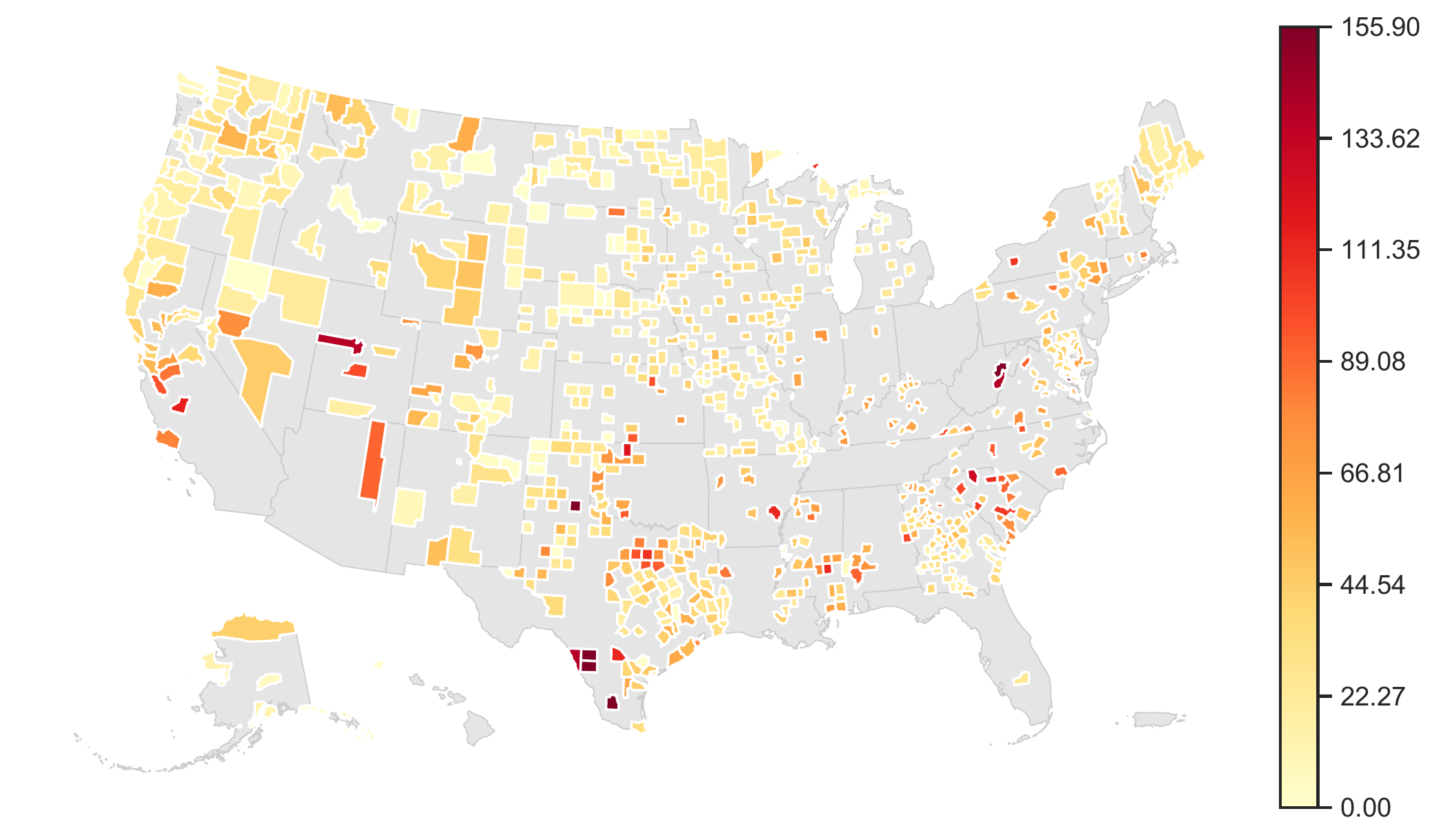}
    \caption{The hardest region the naive baseline produces on the 22nd of January 2021, when we require the region size $r \leq L/4$.  We color the counties according to the true number of COVID-19 cases per 100{,}000 people, smoothed over the previous week.}
    \label{fig:covid-naive-01-22-2021}
\end{figure}

\begin{figure}
  \centering
  \includegraphics[width=0.9\linewidth]{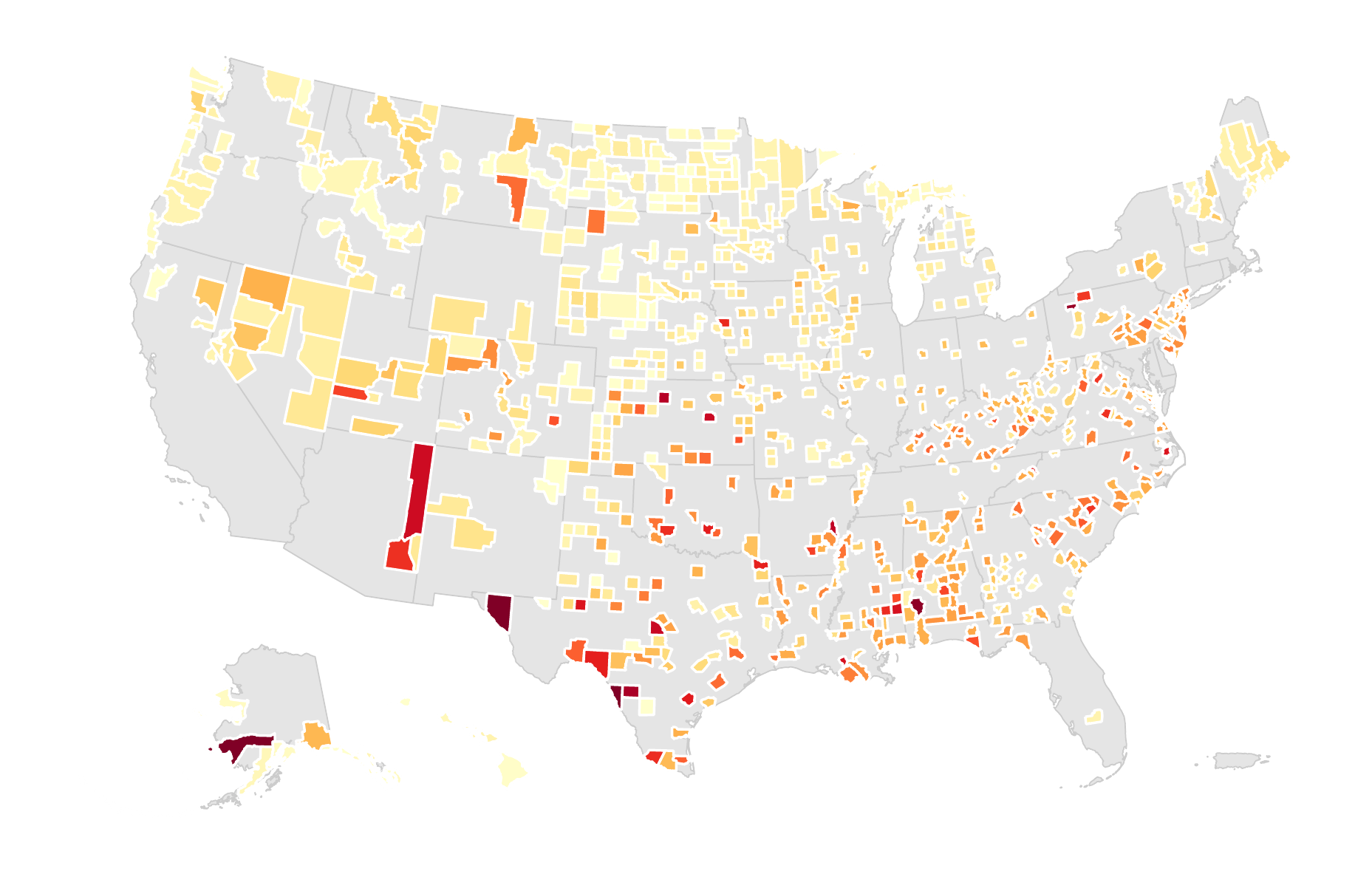}
  \caption{The hardest region the naive baseline produces on the 29th of January 2021, when we require the region size $r \leq L/4$.  We color the counties according to the true number of COVID-19 cases per 100{,}000 people, smoothed over the previous week.}
  \label{fig:covid-naive-01-29-2021}
\end{figure}
\clearpage

\clearpage
\begin{figure}
    \centering
    \includegraphics[width=0.9\linewidth]{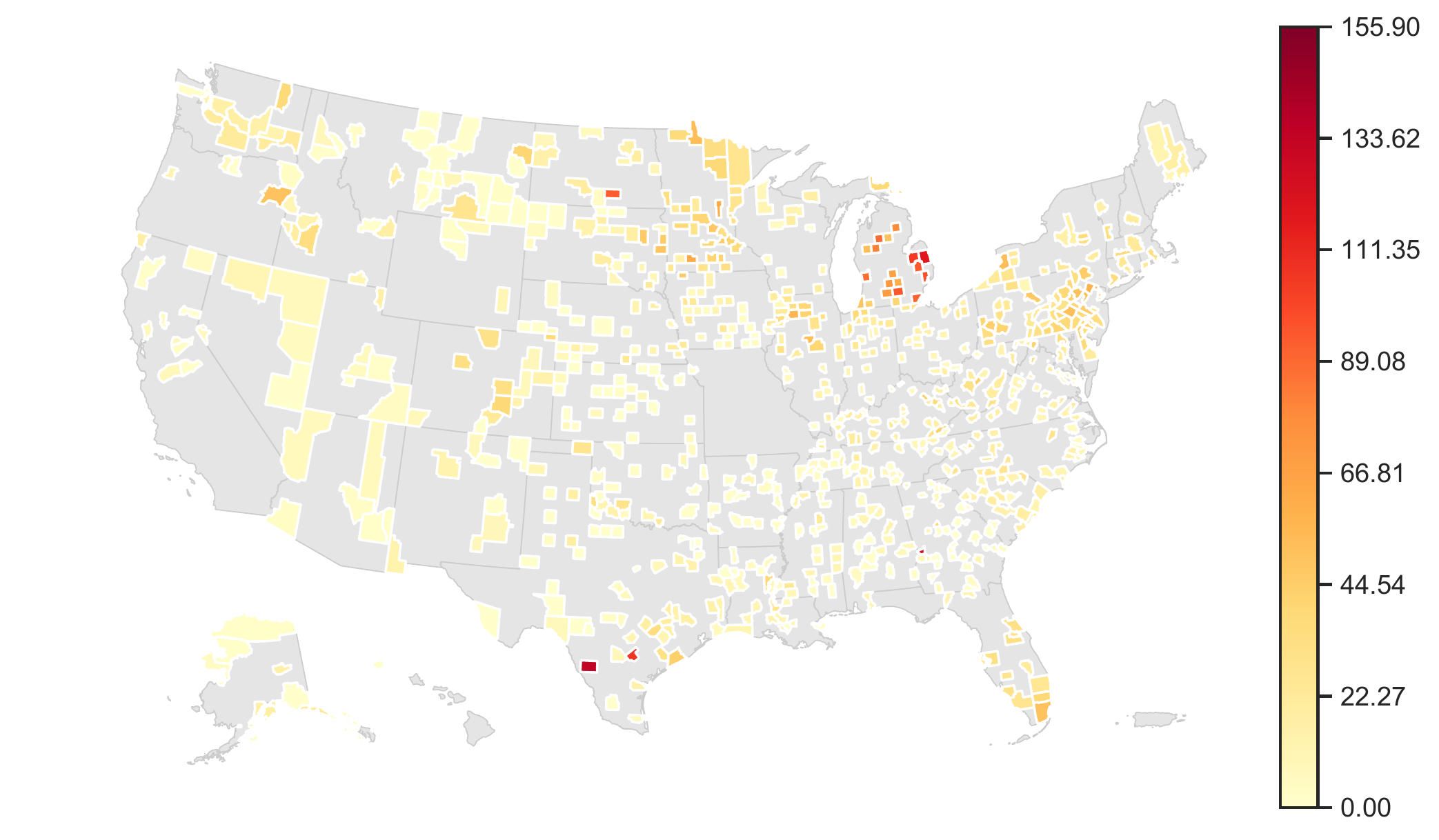}
    \caption{The hardest region the naive baseline produces on the 16th of April 2021, when we require the region size $r \leq L/4$.  We color the counties according to the true number of COVID-19 cases per 100{,}000 people, smoothed over the previous week.}
    \label{fig:covid-naive-04-16-2021}
\end{figure}

\begin{figure}
  \centering
  \includegraphics[width=0.9\linewidth]{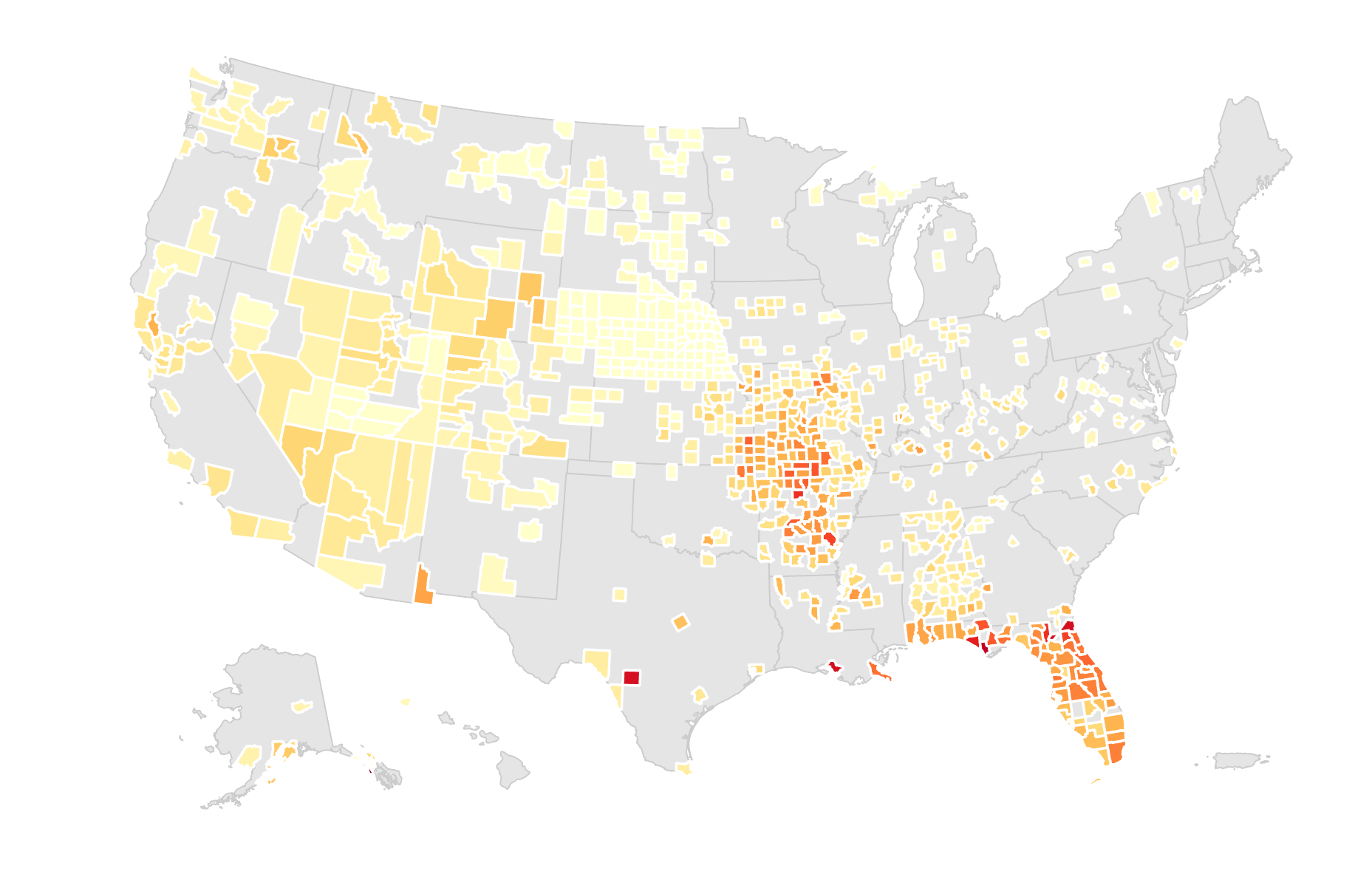}
  \caption{The hardest region the naive baseline produces on the 30th of August 2021, when we require the region size $r \leq L/4$.  We color the counties according to the true number of COVID-19 cases per 100{,}000 people, smoothed over the previous week.}
  \label{fig:covid-naive-07-30-2021}
\end{figure}
\clearpage

\clearpage
\begin{figure}
    \centering
    \includegraphics[width=0.9\linewidth]{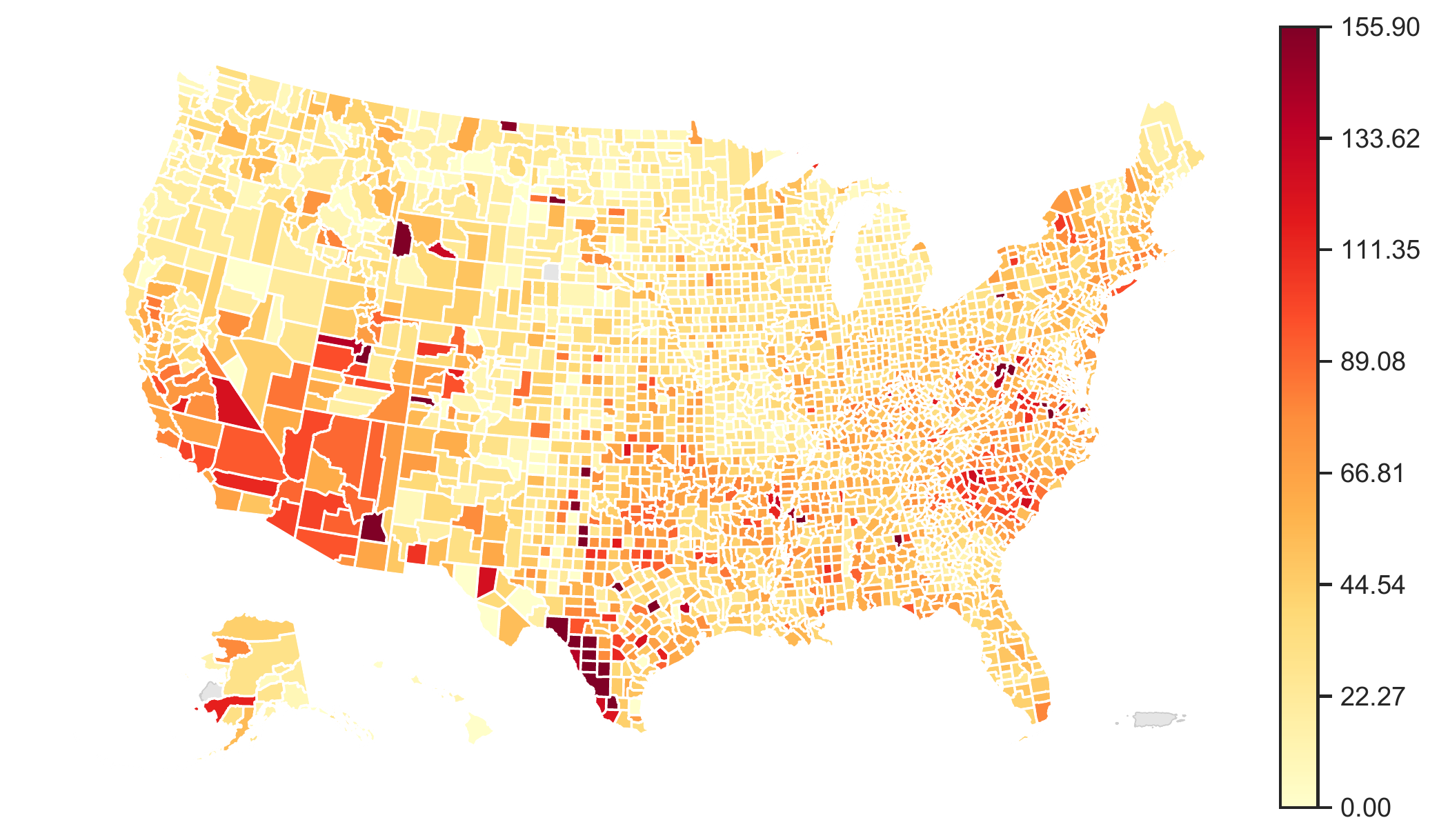}
    \caption{United States counties, which we color according to the true number of COVID-19 cases per 100{,}000 people, smoothed over the week of January 22, 2021.}
    \label{fig:covid-true-01-22-2021}
\end{figure}

\begin{figure}
  \centering
  \includegraphics[width=0.9\linewidth]{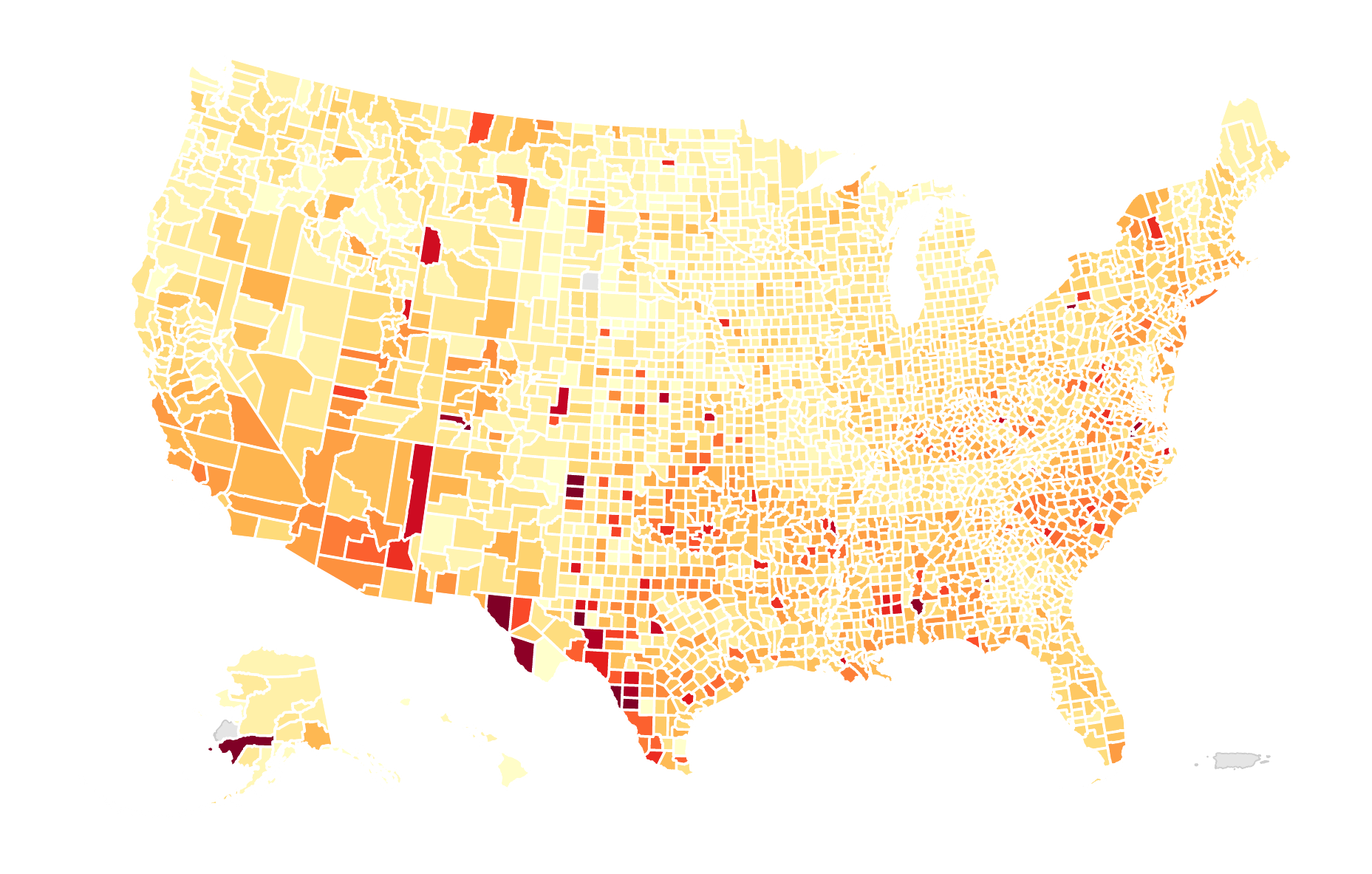}
  \caption{United States counties, which we color according to the true number of COVID-19 cases per 100{,}000 people, smoothed over the week of January 29, 2021.}
  \label{fig:covid-true-01-29-2021}
\end{figure}
\clearpage

\clearpage
\begin{figure}
    \centering
    \includegraphics[width=0.9\linewidth]{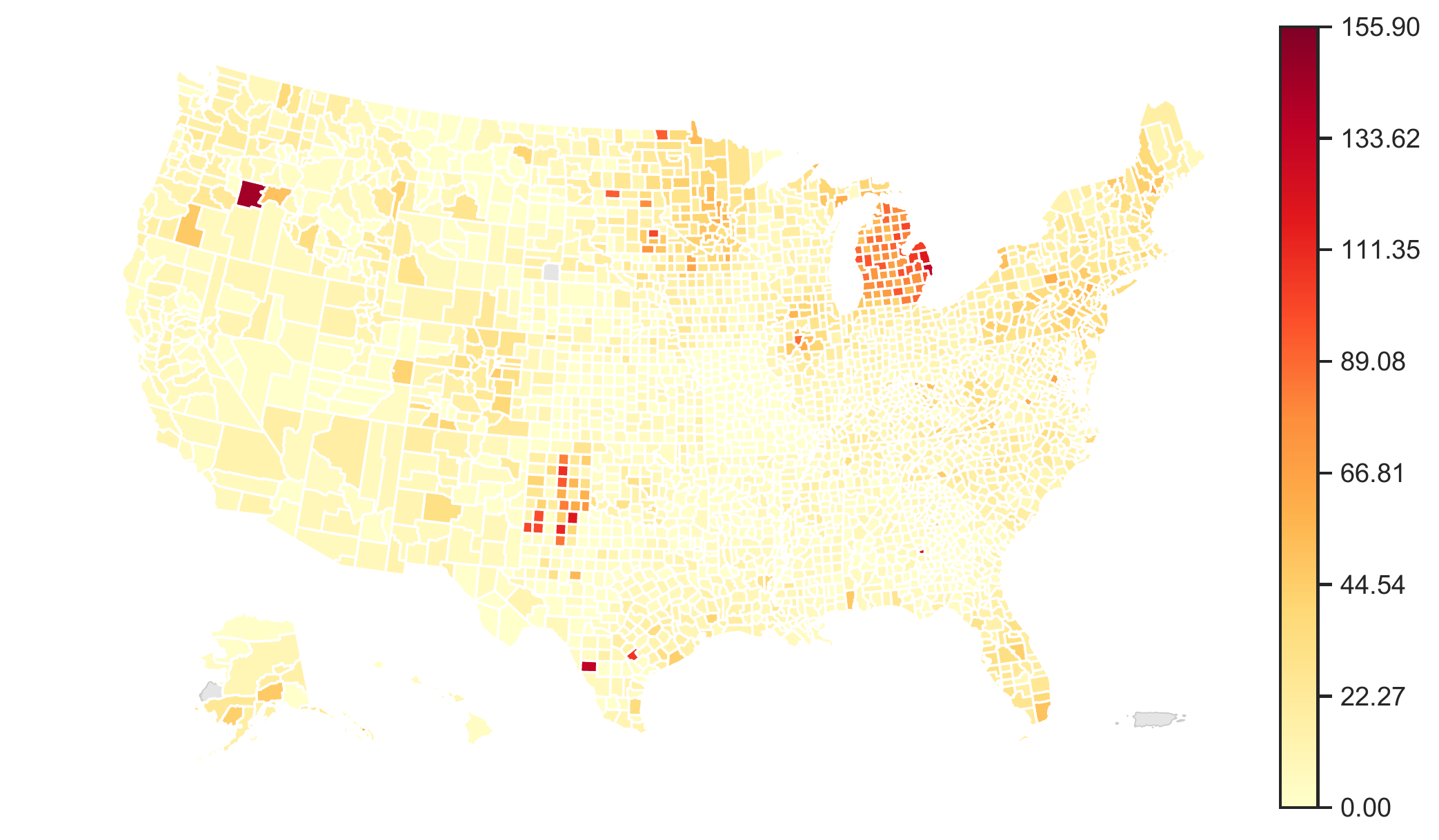}
    \caption{United States counties, which we color according to the true number of COVID-19 cases per 100{,}000 people, smoothed over the week of April 16, 2021.}
    \label{fig:covid-true-04-16-2021}
\end{figure}

\begin{figure}
  \centering
  \includegraphics[width=0.9\linewidth]{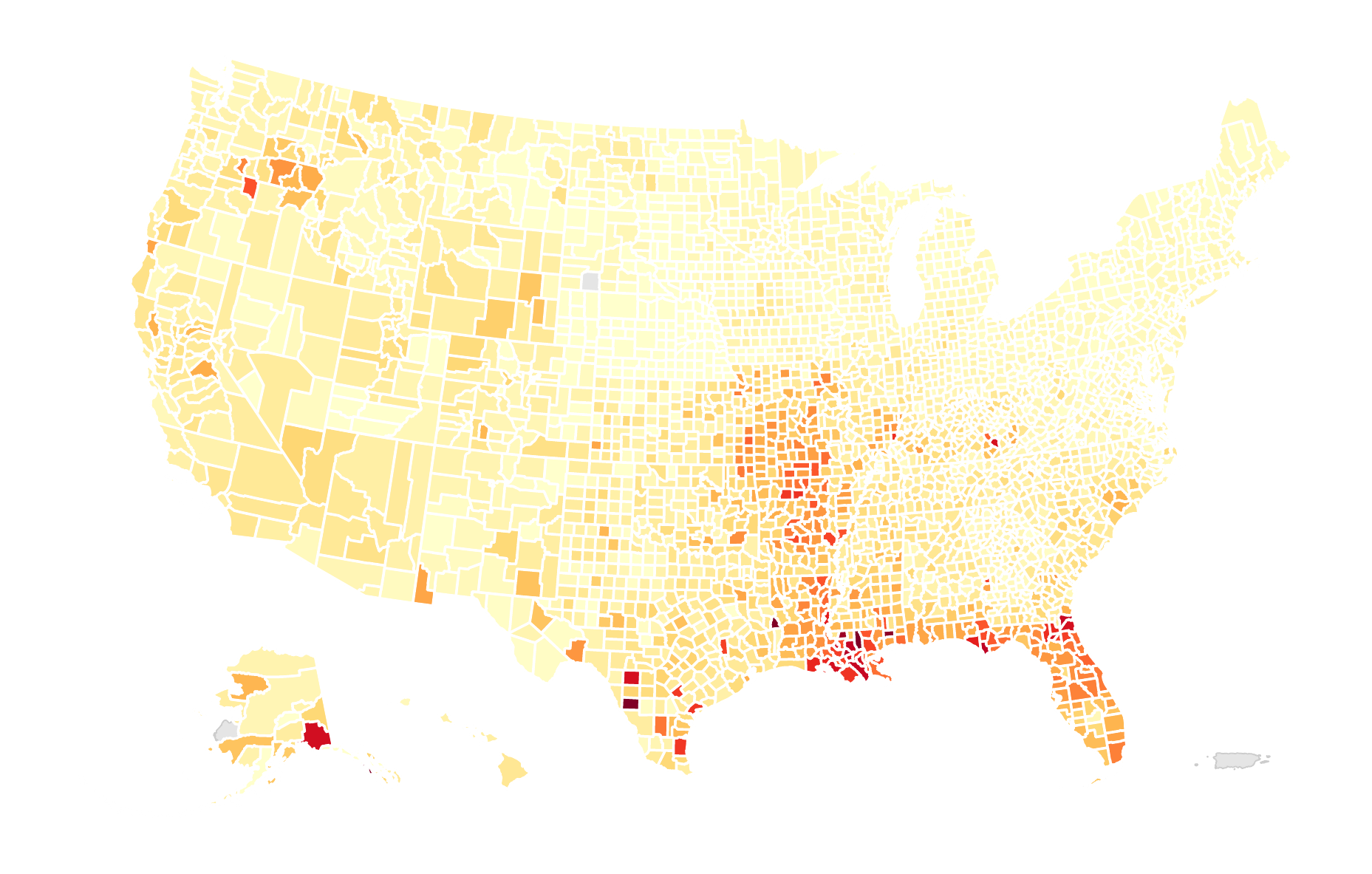}
  \caption{United States counties, which we color according to the true number of COVID-19 cases per 100{,}000 people, smoothed over the week of August 30, 2021.}
  \label{fig:covid-true-07-30-2021}
\end{figure}
\clearpage

\subsection{Distribution shift adaptation}
\label{sec:experiments-poverty}

We now turn to a different experimental set-up, and test our methods on datasets with built-in distribution shifts.
The WILDS project~\cite{KohSaMaXiZhBaHuYaPhBeLeKuPiLeFiLi20} gathers supervised learning datasets in which each instance has a ``group" or domain attribute (sometimes several), such as the country or location the instance comes from,  the identity of the reviewer that gave a certain rating, or the hospital/specific machine that produced the medical image to study.
The existence of such attributes allows us to consider training, validation and test data as mixtures of sub-populations, that is distributions $\{ P_g \}_{g \in \mc{G}}$, where $\mc{G}$ is the set of all different groups---countries, reviewers,hospitals---that form the entire dataset. 

For both WILDS datasets that we investigate---poverty mapping~\cite{KohSaMaXiZhBaHuYaPhBeLeKuPiLeFiLi20} and Amazon reviews~\cite{KohSaMaXiZhBaHuYaPhBeLeKuPiLeFiLi20}---we follow the same general experimental procedure, which replicates a scenario where practitioners, aiming to improve their model and with limited additional (labeled) data available, need to decide how to best allocate their resources and where to gather new data instances.  

\begin{enumerate}[1)]
\item We first train a model on a training set containing only a fraction of the entire groups,  i.e.  we train our model on $P_0^\text{train} = \sum_{g \in \mc{G}_0} \alpha^{\text{train}}_g P_g$ for a certain choice of mixture coefficients $\alpha^{\text{train}}_g > 0$ and $\mc{G}_0 \subsetneq \mc{G}$, and we compute non-conformity scores on an independent calibration set coming from the same restricted distribution $P_0^\text{calib} = P_0^\text{train}$.
\item On a first test set,  which is now a mixture of all different sub-groups present in the dataset, i.e.  $P_0^\text{test} = \sum_{g \in \mc{G}} \alpha^\text{test}_g P_g$ with $\alpha^\text{test}_g > 0$ for all $g \in \mc{G}$, we identify a hard region $R \in \mc{R}$ using Algorithm~\ref{alg:recovery}, using as p-values the ranks of each test non-conformity score among all calibration scores.

\item We then refit a model by augmenting the training set with additional  independent data from $P_0^\text{test} \mid X \in R^\text{hard}$, and compare it to two different baselines: one where the training set receives additional independent data from $P_0^\text{test}$ (``random") and one where the training set receives data points that are neighbors of test instances with the highest ranks (``hardest"). 
\item We eventually test the performance of each refitted model on a second independent test set (from $P_0^\text{test}$).
\end{enumerate} 

\begin{remark}
In our experiments, we choose every coefficient $\alpha_g^\text{train}, \alpha_g^\text{test}$ proportionally to the amount of instances from the sub-population in the entire dataset available.
\end{remark}

The goal of our experimental procedure is two-fold. 
First, since our initial model did not have access to any sample from $\{ P_g \}_{g \notin \mc{G}_0}$, we expect it to perform poorly on these, and hence to detect a region $R^\text{hard}$ comprising mostly of examples from these unseen groups,  which would correspond to having 
\begin{align}
\label{eqn:hypothesis-hard-region-is-unseen}
P_0^\text{test} \mid X \in R^\text{hard} \simeq \sum_{g \in \mc{G} \setminus \mc{G}_0} \alpha_g^\text{hard} P_g.
\end{align}
In particular,  we expect our procedure to be less sensitive to noise and outliers than the more naive ``Hardest" method, which simply includes samples with very high scores and does not take any  feature structure into account. 

If our first hypothesis~\eqref{eqn:hypothesis-hard-region-is-unseen} holds (at least partially), we then would expect, during the second training phase, a larger improvement in performance on these sub-groups with our method than with the two other baseline procedures, which add the same amount of data to the training set, but in a less targeted fashion. 
We thus hope that our method shows better or equivalent average performance on $P_0^\text{test}$,  but even more so that it significantly outperforms both baselines on each sub-population $\{ P_g \}_{g \in \mc{G} \setminus \mc{G}_0}$.  

Crucially, in these experiments, we only use knowledge of the group or protected attribute $g \in \mc{G}$ to construct the distributions $P_0^\text{train}$ and $P_0^\text{test}$: none of the methods has access to that piece of information to choose which instances to train with.
Even if we expect the model to display group-heterogeneous performance, our method (Alg.~\ref{alg:recovery}) cannot use it directly as a discriminant: the hypothesis is that examples from the same group should also cluster, at least partially, in the feature space.

\subsubsection{Poverty mapping}
We first experiment with the poverty map dataset~\cite{KohSaMaXiZhBaHuYaPhBeLeKuPiLeFiLi20}, where we aim to predict the poverty level across spatial regions from satellite imagery, precisely their asset wealth index.
A notable challenge of this problem is the scarcity of poverty level measurements in some regions of the world, especially in comparison with the wide availability of unlabeled satellite imagery: this calls for models robust and adaptive to geographical distribution shifts, and allows us to test our methodology.
The group or sub-population $g \in  \mc{G}$ of each instance is the country where the image comes from; the data originates from $|\mc{G}| =23$ different countries,  among which four of them ($\mc{G} \setminus \mc{G}_0 = \{ \text{Cameroon, Ghana, Malawi, Zimbabwe} \}$) only appear in the test distribution $P_0^\text{test}$.

We train all our models using the default network architecture and hyper-parameters in the WILDS package, minimizing the average least-squares loss---this corresponds to the ERM algorithm with a ResNet18-MS model. 
By doing so, we make sure that the distribution of each respective training set is the only difference between our different models that we compare.

In our experiments, when applying Alg.~\ref{alg:recovery} we vary one additional parameter $\delta \in (0, 1)$, which controls the maximum size of the hard region that Alg.~\ref{alg:recovery} can detect. The reason why we need such parameter is simple: in real datasets, it is plausible that large sub-populations of the data (and not simply are actually \textit{much} harder to predict or classify that some others, hence with ranks significantly higher than uniform: this could (and in some cases, would) lead Alg~\ref{alg:recovery} to focus on regions that are potentially too large to be of practical use. This is why we focus on detecting regions $R^\text{hard}$ such that $P_0^\text{test}(X \in R^\text{hard}) \le \delta$: our goal is to detect reasonably small regions, with the hope that they overlap with hard out of domain instances.
Finally, the set of regions on which we apply our detection method is the set of euclidean balls around the test points in the first test set, with the caveat that we use as feature vector $x \in \R^d$ the output of the pooling layer that precedes the last layer (and not the initial image itself), thus allowing the dimension of the problem to be lower.

We display our results in the three plots comprising Figure~\ref{fig:poverty_map_figs}, and summarize them in Table~\ref{tab:poverty-map-evaluation}. They are consistent with our initial expectations: Algorithm~\ref{alg:recovery} offers a bigger performance improvement to the Baseline model than the more naive ``Hardest point" and ``Random" methods, across the whole range of different $\delta$, whether in terms of average error or out of domain error, the latter improvement being more significant.
Additionally, the difference in performance between each method tends to increase as function of $\delta$, meaning that for this specific dataset and sub-populations choices, it appears beneficial to allow for a large hard region, potentially because the performance of the model is particularly heterogeneous across different regions of the world.

\begin{table}
\centering
\begin{tabular}{lllll}
\toprule
     & & \multicolumn{3}{c}{Type of mean squared error} \\
      \cmidrule{3-5}
$\delta$ & Subpopulation identification strategy &  Average &    O.O.D Region &     Hard Region       \\
\midrule
\midrule
0.10 & Algorithm~\ref{alg:recovery}  &  \textbf{0.2157}(0.0195) &  \textbf{0.3028}(0.0179) &  \textbf{0.3872}(0.0267) \\
     & Hardest points &  0.2257(0.0128) &  0.3229(0.0114) &  0.4101(0.0289) \\
     & Uniformly at random &  0.2273(0.0133) &  0.3228(0.0128) &  0.4112(0.0289) \\
          & Baseline &  0.2586(0.0111) &  0.3333(0.0252) &  0.4847(0.0574) \\
     \midrule
0.15 & Algorithm~\ref{alg:recovery}  &  \textbf{0.2151}(0.0161) &  \textbf{0.3037}(0.0155) &  \textbf{0.3651}(0.0262) \\
     & Hardest points &  0.2288(0.0141) &  0.3276(0.0144) &  0.3904(0.0305) \\
     & Uniformly at random &  0.2298(0.0157) &  0.3244(0.0126) &  0.3891(0.0321) \\
          & Baseline &  0.2586(0.0111) &  0.3333(0.0252) &  0.4633(0.0528) \\
      \midrule
0.20 & Algorithm~\ref{alg:recovery}  &   \textbf{0.2215}(0.022) &  \textbf{0.3077}(0.0184) &  \textbf{0.3544}(0.0315) \\
     & Hardest points &  0.2258(0.0114) &  0.3234(0.0091) &  0.3694(0.0275) \\
     & Uniformly at random &  0.2376(0.0337) &  0.3301(0.0276) &  0.3744(0.0374) \\
          & Baseline &  0.2586(0.0111) &  0.3333(0.0252) &  0.4467(0.0388) \\
      \midrule
0.25 & Algorithm~\ref{alg:recovery}  &  \textbf{0.2166}(0.0245) &  \textbf{0.2988}(0.0204) &   \textbf{0.341}(0.0345) \\
     & Hardest points &  0.2284(0.0156) &  0.3265(0.0132) &  0.3606(0.0264) \\
     & Uniformly at random &  0.2277(0.0124) &  0.3206(0.0101) &  0.3581(0.0239) \\
          & Baseline &  0.2586(0.0111) &  0.3333(0.0252) &   0.4292(0.038) \\
      \midrule
0.30 & Algorithm~\ref{alg:recovery} &   \textbf{0.2092}(0.016) &  \textbf{0.2928}(0.0169) &  \textbf{0.3157}(0.0175) \\
     & Hardest points &    0.2300(0.0108) &  0.3261(0.0081) &  0.3435(0.0151) \\
     & Uniformly at random &  0.2269(0.0109) &   0.3231(0.009) &  0.3418(0.0157) \\
          & Baseline &  0.2586(0.0111) &  0.3333(0.0252) &  0.4032(0.0157) \\
\bottomrule
\end{tabular}
  \caption{Mean squared error, averaged over $M=10$ trials, for different values of $\delta \in \{0.1, 0.15, 0.2, 0.25, 0.3\}$, in the poverty map dataset. 
    We report three types of error: the average error over an independent test set from $P_0^\text{test}$ (data from every $\{ P_g \}_{g \in \mc{G}}$), the average error over the test set restricted to out of domain data (data only from $\{ P_g \}_{g \notin \mc{G}_0}$), and the average error over the hard region $P_0^\text{test} \mid X \in R^\text{hard}$ that Algorithm \ref{alg:recovery} detects---we of course expect our method to show better performance on the latter, so we report it more as a sanity check.
    We highlight the best (i.e., lowest) error, for each type of error strategy, in bold, and report the standard deviation over $M=10$ trials in parentheses.}
  \label{tab:poverty-map-evaluation}

\end{table}

\begin{figure}[h!]
    \centering
    \begin{overpic}[scale=0.5, 
    ]{%
      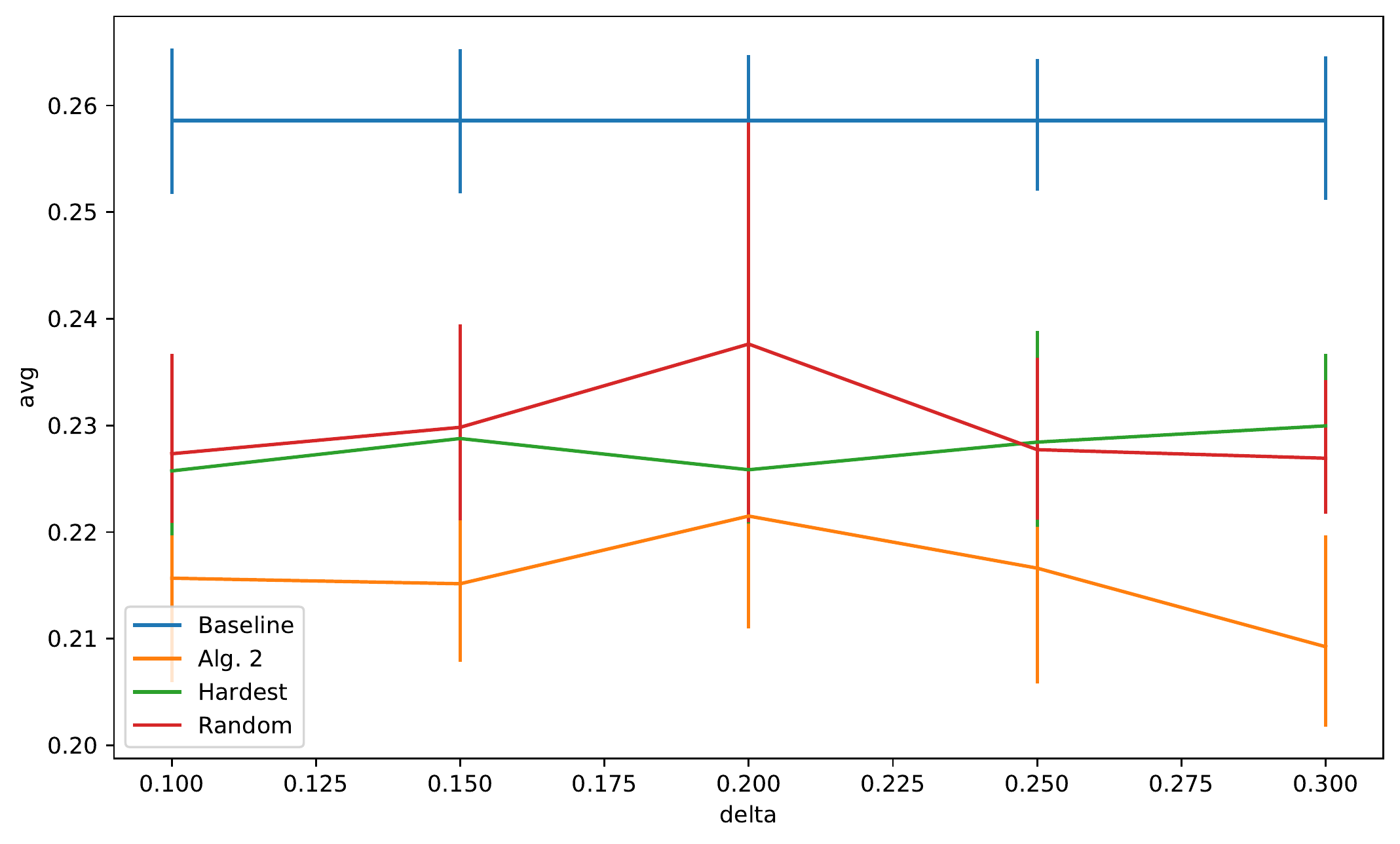}
    \put(0,10){
      \tikz{\path[draw=white, fill=white] (0, 0) rectangle (.2cm, 6cm)}
    }
    \put(0,10){\rotatebox{90}{
        \small Average mean squared error}
    }
    \put(40, 0){
      \tikz{\path[draw=white, fill=white] (0, 0) rectangle (4cm, .35cm)}
    }
    \put(50, 1){
      \small $\delta$}

%
    
  \end{overpic}   
    \\
    \begin{overpic}[scale=0.5, 
    ]{%
      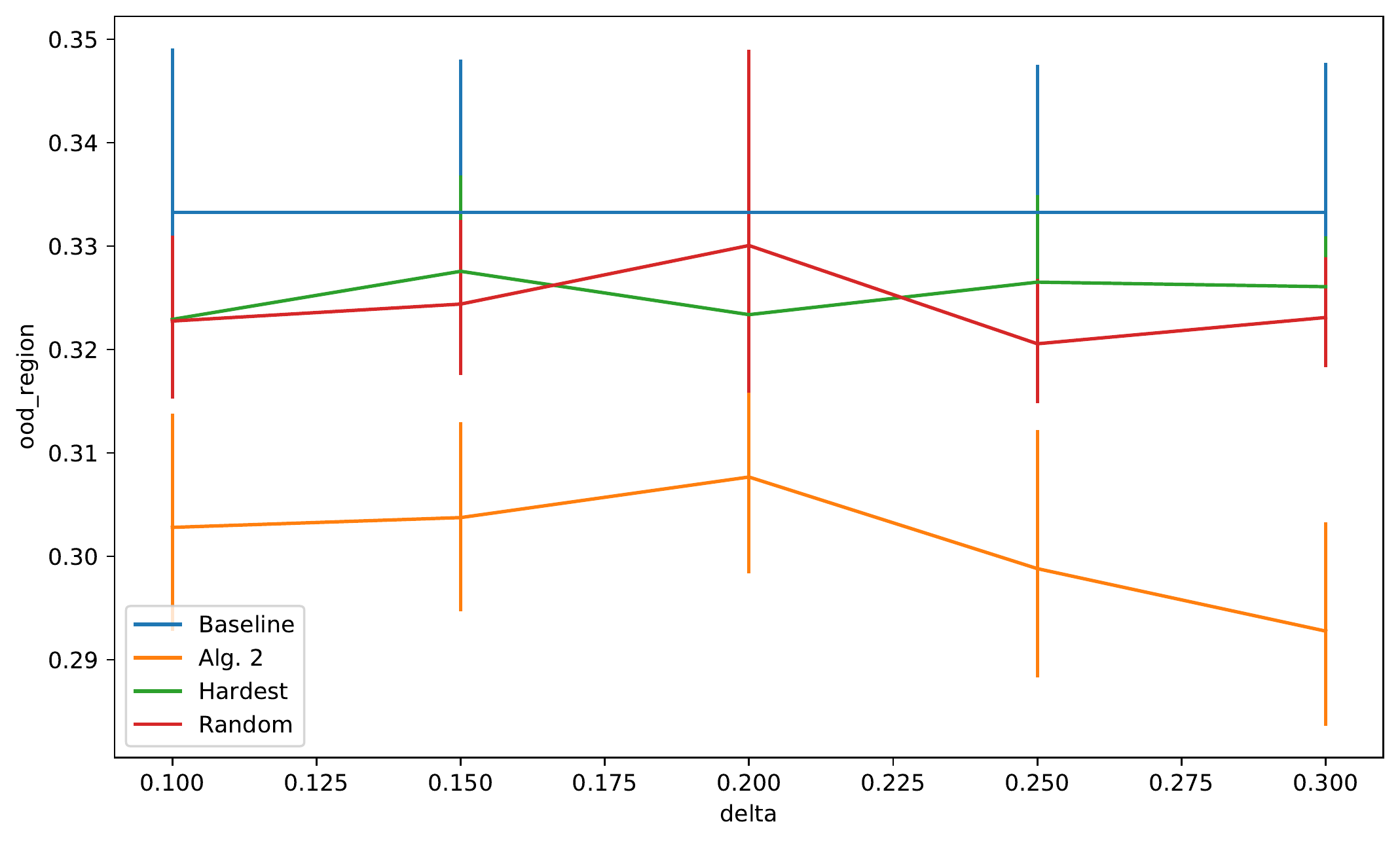}
    \put(0,10){
      \tikz{\path[draw=white, fill=white] (0, 0) rectangle (.2cm, 6cm)}
    }
    \put(0,10){\rotatebox{90}{
        \small O.O.D mean squared error}
    }
    \put(40, 0){
      \tikz{\path[draw=white, fill=white] (0, 0) rectangle (4cm, .35cm)}
    }
    \put(50, 1){
      \small $\delta$}

%
    
  \end{overpic}    
    \\
     \begin{overpic}[scale=0.5, 
    ]{%
      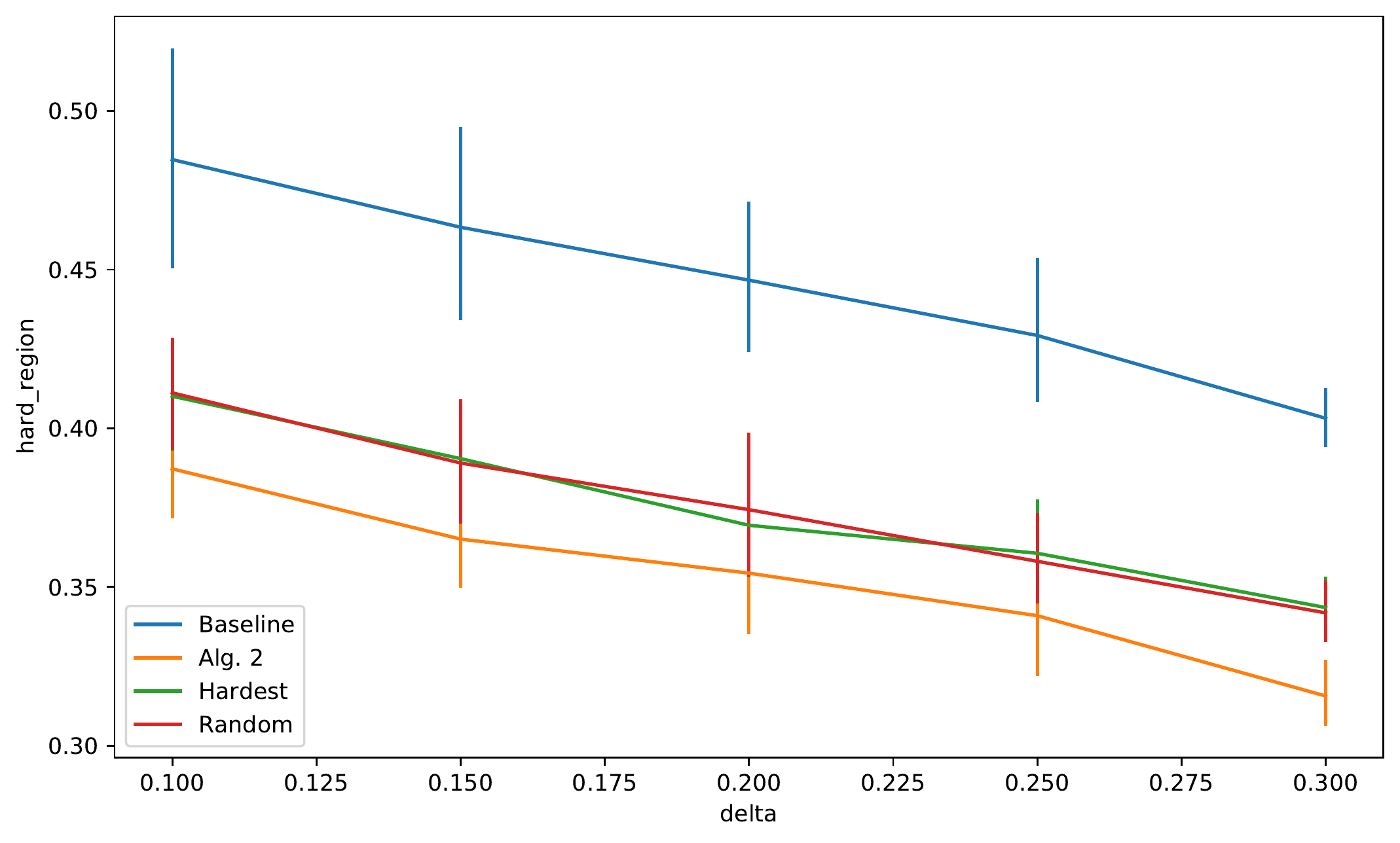}
    \put(0,10){
      \tikz{\path[draw=white, fill=white] (0, 0) rectangle (.2cm, 6cm)}
    }
    \put(0,10){\rotatebox{90}{
        \small Hard region mean squared error}
    }
    \put(40, 0){
      \tikz{\path[draw=white, fill=white] (0, 0) rectangle (4cm, .35cm)}
    }
    \put(50, 1){
      \small $\delta$}

%
    
  \end{overpic}  
    
    \caption{Mean squared error, averaged over $M=10$ trials, for different values of $\delta \in \{0.1, 0.15, 0.2, 0.25, 0.3\}$, in the poverty map dataset. 
    We report three types of error (one for each plot): the average error over an independent test set from $P_0^\text{test}$ (data from every $\{ P_g \}_{g \in \mc{G}}$), the average error over the test set restricted to out of domain data (data only from $\{ P_g \}_{g \notin \mc{G}_0}$), and the average error over the hard region $P_0^\text{test} \mid X \in R^\text{hard}$ that Alg. \ref{alg:recovery} detects---we of course expect our method to show better performance on the latter, so we only report it more as a sanity check.
    The ``Baseline" method is the initial model, consisting of data from $P_0^\text{train}$.
   We report error bars as twice the standard error over the $M$ trials.}
    \label{fig:poverty_map_figs}
\end{figure}

\subsubsection{Review rating prediction}
We next study the impact of weak supervision on our methods, experimenting on the Amazon review dataset~\cite{KohSaMaXiZhBaHuYaPhBeLeKuPiLeFiLi20}.
The goal here is to predict what rating on a scale from 1 to 5 some user left based on the comment they wrote; each particular user represents a different sub-population, and the out-of-domain region simply is a set of users for which none of their comments belongs to the training set.

The Amazon review dataset is fully supervised, meaning that all ratings $Y \in [5]$ are available. However,  for the purpose of testing our method in a partially labeled setting,  we introduce weak supervision in the first test set, i.e.\ when finding hard regions with Alg.~\ref{alg:recovery}. 
Specifically, instead of observing the actual rating $Y$, we assume that we only have access to a ``noisy" version of it, namely an interval $Y_\text{weak} \defeq [Y_\text{min}, Y_\text{max}] \subset [1,5]$ that contains the true rating $Y$. For instance,  if the initial true rating was $Y=4$, we could only observe $Y_\text{weak} = \{3,4,5\}$.
For simplicity, we introduce partial supervision in the following way: for each instance, $x,y \in \mc{X} \times [5]$, and for some real parameter $c>0$, we have
\begin{align*}
p( y_\text{weak} = [y_\text{min}, y_\text{max}] \mid x, y) \propto e^{- c \left(y_\text{max} - y_\text{min}\right)} \indic{ y \in y_\text{weak}},
\end{align*} 
which means that the distribution of the partial label only depends on the actual rating (probably too simplistic in practice), and that the probability of the size of the interval decreases exponentially.  The parameter $c>0$ controls the average size of the ``weak" label set: the bigger it is, the closer to full supervision we are.  We run our forthcoming experiments with values of $c>0$ such that $\E \left[ \left| Y_\text{weak} \right| \right] \in \{ 1.2, 1.5\} $, to compare two different noise levels of weak supervision.
We plot the distribution of the weak label size $|Y_\text{weak}|$ conditionally on the label $Y$ in Figure~\ref{fig:amazon_weak_set_size}.

Similarly to the poverty map experiment, we report the average accuracy of the different methods on three different groups: the entire distribution,  the out of domain region,  and the hard region Alg.~\ref{alg:recovery} unveils.
Additionally,  to evaluate  out-of-domain performance,  for each method and out-of-domain user, we compute the average accuracy and compare it to its baseline counterpart. This results, for each method,  in a distribution of the difference in accuracy over the set of O.O.D.\ users; we then report the c.d.f.\ of that distribution as a measure of improvement over out-of-domain reviews (see Figure~\ref{fig:amazon_review_supervision_per_user}).

To provide a comparison baseline, we run the same methods as in the previous Section in the full supervision setting,  and report our results in Figure~\ref{fig:amazon_review_supervision_per_user}A and Table~\ref{tab:amazon-review-full-supervision-evaluation}.
Our findings here are consistent with the conclusions we previously drew, in the sense that Alg.~\ref{alg:recovery} allows a small but significant improvement of performance over the more naive methods ``Hardest" and ``Random" methods.

The comparison for the partially labeled setting has more nuances. To run Alg.~\ref{alg:recovery}, we now use as test and calibration scores the min-scores as in Eqn.~\eqref{eq:weak}. 
In most instances,  especially in high accuracy tasks, they are equal to the true scores,which is why we would expect our results in the partially supervised setting to echo those in the fully supervised regime. 
This is only partially the case: when introducing small label noise (i.e., $\E \left[ \left| Y_\text{weak} \right| \right] = 1.2$), our method indeed generates models with higher accuracies in and out of domain,  as we outline in Table~\ref{tab:amazon-review-partial-supervision-evaluation-12} and~\ref{tab:amazon-review-partial-supervision-evaluation-15}, and Figures~\ref{fig:amazon_review_supervision_per_user}B/C.
On the other hand,  when weak supervision is inherently noisier (i.e, $\E |Y_\text{weak}| = 1.5$ is larger),  Table~\ref{tab:amazon-review-partial-supervision-evaluation-15} shows that the ``Hardest" method is on-par or even better than Alg.~\ref{alg:recovery} for larger sizes $\delta > 0$: it is possible that weak supervision combined with larger sizes of hard subsets have itself an implicit regularization effect on that more naive method, resulting in better performance.


\begin{figure}[h]
    \centering
    \begin{overpic}[scale=0.5, 
    ]{%
      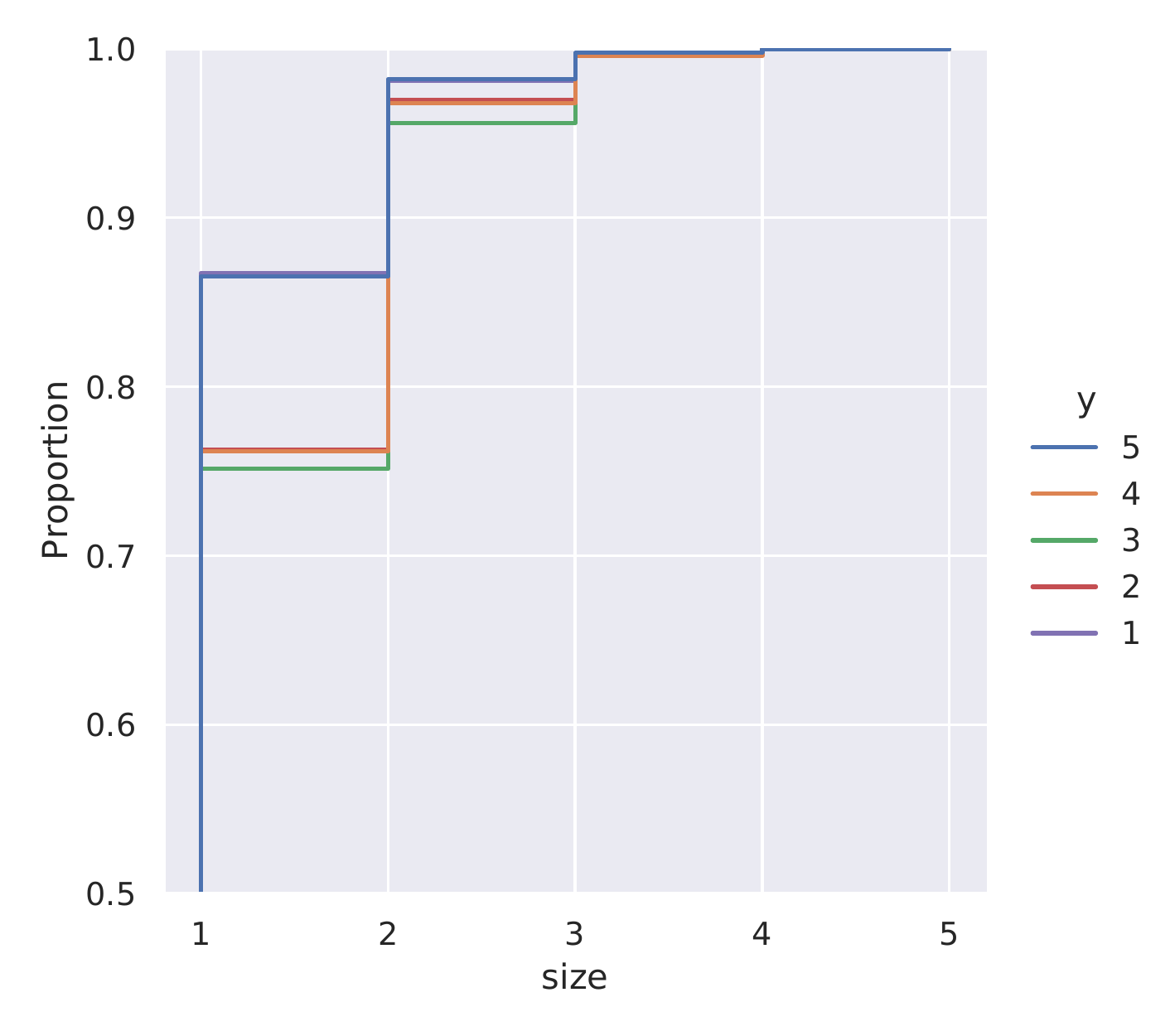}
    \put(2,10){
      \tikz{\path[draw=white, fill=white] (0, 0) rectangle (.2cm, 6cm)}
    }
    \put(2,20){\rotatebox{90}{
        \small $\P( |Y_\text{weak}| \le t \mid Y=y)$}
    }
    \put(40, 2){
      \tikz{\path[draw=white, fill=white] (0, 0) rectangle (4cm, .35cm)}
    }
    \put(47, 3){
      \small $t$}
  \end{overpic}   
     \begin{overpic}[scale=0.5, 
    ]{%
      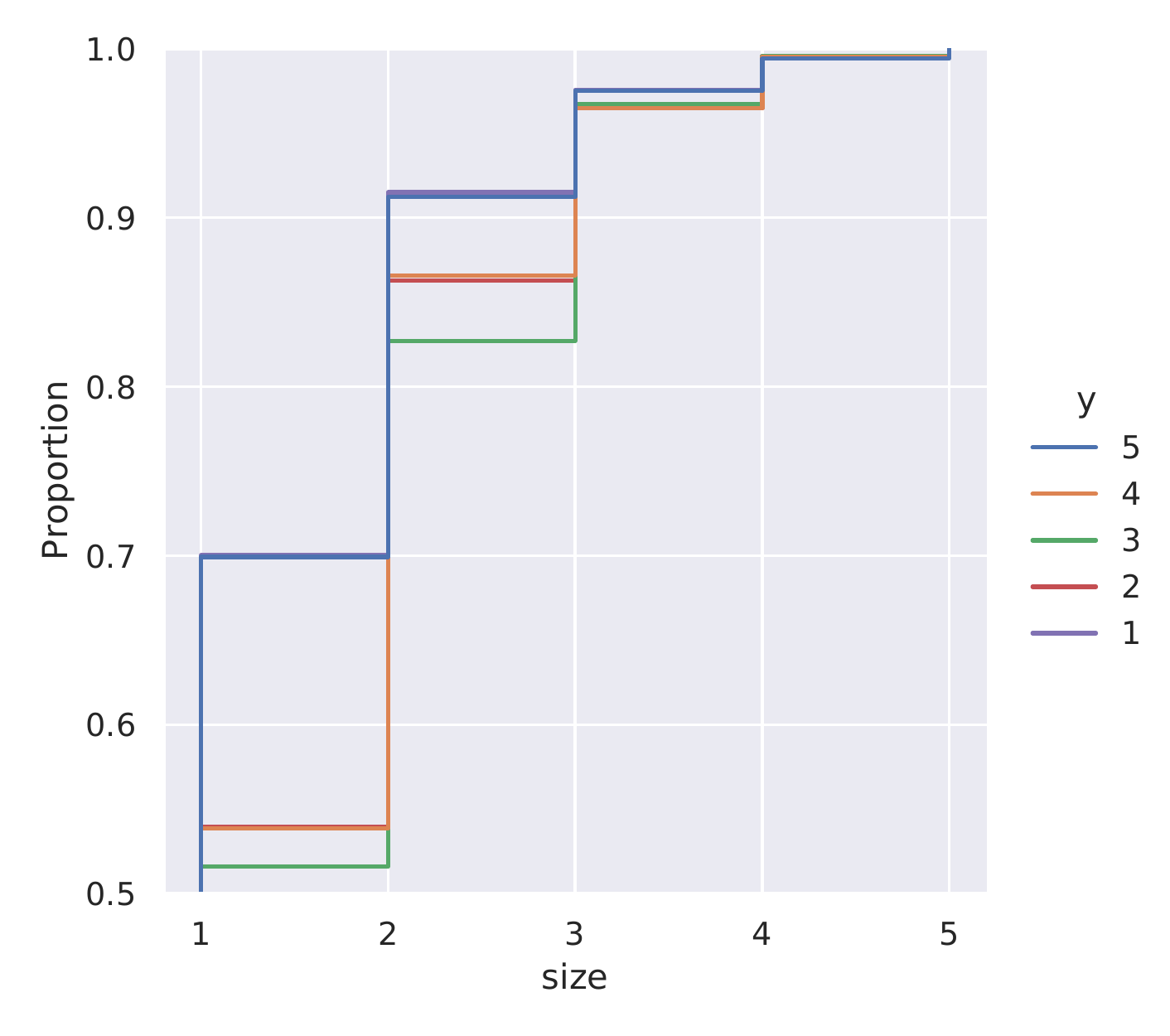}
    \put(2,10){
      \tikz{\path[draw=white, fill=white] (0, 0) rectangle (.2cm, 6cm)}
    }
    \put(2,20){\rotatebox{90}{
        \small $\P( |Y_\text{weak}| \le t \mid Y=y)$}
    }
    \put(40, 2){
      \tikz{\path[draw=white, fill=white] (0, 0) rectangle (4cm, .35cm)}
    }
    \put(47, 2){
      \small $t$}
    
  \end{overpic}  
    
    \caption{Distribution of the weak label set size $|Y_\text{weak}|$ in the Amazon review dataset experiment, for two different values of the parameter $c>0$ corresponding to respective average sizes $\E |Y_\text{weak}$ of $1.2$ (left plot) and $1.5$ (right plot).  
}
    \label{fig:amazon_weak_set_size}
\end{figure}

\begin{table}
\centering
\begin{tabular}{lllll}
\toprule
     & & \multicolumn{3}{c}{Type of  average accuracy} \\
      \cmidrule{3-5}
$\delta$ & Subpopulation identification strategy &  Average &    O.O.D Region &     Hard Region       \\
\midrule
\midrule
0.05 & Algorithm~\ref{alg:recovery} &  \textbf{0.7312}(0.0015) &  \textbf{0.7232}(0.0017) &  \textbf{0.4584}(0.0066) \\
     & Baseline &  0.7292(0.0018) &   0.7204(0.003) &  0.4386(0.0057) \\
     & Hardest points &  0.7297(0.0013) &  0.7205(0.0022) &  0.4458(0.0041) \\
     & Uniformly at random &     0.73(0.002) &  0.7211(0.0026) &  0.4463(0.0077) \\
     \midrule
0.10 & Algorithm~\ref{alg:recovery} &  \textbf{0.7327}(0.0014) &   \textbf{0.726}(0.0024) &  \textbf{0.5122}(0.0049) \\
     & Baseline &  0.7292(0.0018) &   0.7204(0.003) &  0.4919(0.0048) \\
     & Hardest points &  0.7312(0.0022) &   0.7234(0.003) &   0.5027(0.004) \\
     & Uniformly at random &  0.7311(0.0014) &   0.723(0.0022) &  0.5008(0.0034) \\
     \midrule
0.15 & Algorithm~\ref{alg:recovery} &  \textbf{0.7335}(0.0013) &  \textbf{0.7267}(0.0016) & \textbf{0.5433}(0.0031) \\
     & Baseline &  0.7292(0.0018) &   0.7204(0.003) &   0.525(0.0054) \\
     & Hardest points &   0.7327(0.002) &  0.7254(0.0038) &  0.5374(0.0035) \\
     & Uniformly at random &  0.7313(0.0018) &   0.7232(0.003) &  0.5329(0.0038) \\
\bottomrule
\end{tabular}

  \caption{Accuracy for different values of $\delta \in \{0.05, 0.10, 0.15\}$, in the Amazon review dataset in the fully supervised regime. 
    We report three types of accuracy: the average accuracy over an independent test set from $P_0^\text{test}$ (data from every $\{ P_g \}_{g \in \mc{G}}$), the average accuracy over the test set restricted to out of domain data (data only from $\{ P_g \}_{g \notin \mc{G}_0}$), and the average accuracy over the hard region $P_0^\text{test} \mid X \in R^\text{hard}$ that Algorithm \ref{alg:recovery} detects.
    We highlight the best accuracy for each type of population  and report the standard deviation over $M=5$ trials in parentheses.}
  \label{tab:amazon-review-full-supervision-evaluation}
\end{table}

\begin{table}
\centering
\begin{tabular}{lllll}
\toprule
     & & \multicolumn{3}{c}{Type of average accuracy} \\
      \cmidrule{3-5}
$\delta$ & Subpopulation identification strategy &  Average &    O.O.D Region &     Hard Region       \\
\midrule
\midrule
0.05 & Algorithm~\ref{alg:recovery} &   \textbf{0.732}(0.0012) &  \textbf{0.7251}(0.0021) & \textbf{ 0.4648}(0.0065) \\
     & Baseline &  0.7288(0.0021) &   0.721(0.0037) &  0.4338(0.0053) \\
     & Hardest points &  0.7303(0.0018) &   0.7226(0.003) &  0.4526(0.0059) \\
     & Uniformly at random &  0.7304(0.0018) &   0.723(0.0036) &   0.4472(0.007) \\
       \midrule
0.10 & Algorithm~\ref{alg:recovery} &   \textbf{0.7326}(0.002) &  \textbf{0.7268}(0.0037) & \textbf{0.5141}(0.0081) \\
     & Baseline &  0.7288(0.0021) &   0.721(0.0037) &  0.4925(0.0047) \\
     & Hardest points &  0.7309(0.0023) &  0.7239(0.0029) &  0.5029(0.0066) \\
     & Uniformly at random &  0.7311(0.0019) &    0.725(0.003) &   0.504(0.0073) \\
     \midrule
0.15 & Algorithm~\ref{alg:recovery} &  \textbf{0.7327}(0.0015) &  \textbf{0.7272}(0.0034) &   \textbf{0.558}(0.0162) \\
     & Baseline &  0.7288(0.0021) &   0.721(0.0037) &  0.5411(0.0206) \\
     & Hardest points &  0.7325(0.0023) &  0.7264(0.0027) &  0.5524(0.0177) \\
     & Uniformly at random &  0.7315(0.0017) &  0.7255(0.0032) &  0.5503(0.0183) \\
\bottomrule
\end{tabular}

  \caption{Accuracy for different values of $\delta \in \{0.05, 0.10, 0.15\}$, in the Amazon review dataset in the partially supervised regime, with an average size $\E \left[ |Y_\text{weak}| \right] = 1.2$. }
  \label{tab:amazon-review-partial-supervision-evaluation-12}
\end{table}

\begin{table}
\centering
\begin{tabular}{lllll}
\toprule
     & & \multicolumn{3}{c}{Type of average accuracy} \\
      \cmidrule{3-5}
$\delta$ & Subpopulation identification strategy &  Average &    O.O.D Region &     Hard Region       \\
\midrule
\midrule
0.05 & Algorithm~\ref{alg:recovery} &   \textbf{0.731}(0.0009) &  \textbf{0.7235}(0.0031) &  \textbf{0.4653}(0.0127) \\
     & Baseline &  0.7283(0.0012) &   0.7199(0.002) &  0.4431(0.0112) \\
     & Hardest points &  0.7296(0.0011) &  0.7223(0.0038) &  0.4538(0.0104) \\
     & Uniformly at random &  0.7296(0.0017) &  0.7219(0.0032) &  0.4498(0.0183) \\
     \midrule
0.10 & Algorithm~\ref{alg:recovery} &   0.7306(0.001) &  \textbf{0.7243}(0.0018) &  \textbf{0.5493}(0.0304) \\
     & Baseline &   0.729(0.0015) &  0.7214(0.0034) &  0.5308(0.0314) \\
     & Hardest points &   \textbf{0.7314}(0.001) &  0.7237(0.0029) &  0.5447(0.0333) \\
     & Uniformly at random &  0.7304(0.0012) &  0.7236(0.0037) &  0.5401(0.0332) \\
     \midrule
0.15 & Algorithm~\ref{alg:recovery} &  0.7314(0.0009) &  0.7244(0.0024) & \textbf{0.5758}(0.0149) \\
     & Baseline &  0.7264(0.0049) &  0.7186(0.0046) &  0.5577(0.0213) \\
     & Hardest points &  \textbf{0.7323}(0.0015) &   \textbf{0.7251}(0.003) &   0.5722(0.018) \\
     & Uniformly at random &  0.7314(0.0013) &  0.7244(0.0028) &  0.5685(0.0166) \\
\bottomrule
\end{tabular}

  \caption{Accuracy for different values of $\delta \in \{0.05, 0.10, 0.15\}$, in the Amazon review dataset in the partially supervised regime, with an average size $\E \left[ |Y_\text{weak}| \right] = 1.5$. }
  \label{tab:amazon-review-partial-supervision-evaluation-15}
\end{table}

\begin{figure}[h!]
    \centering
    \begin{overpic}[scale=0.5, 
    ]{%
      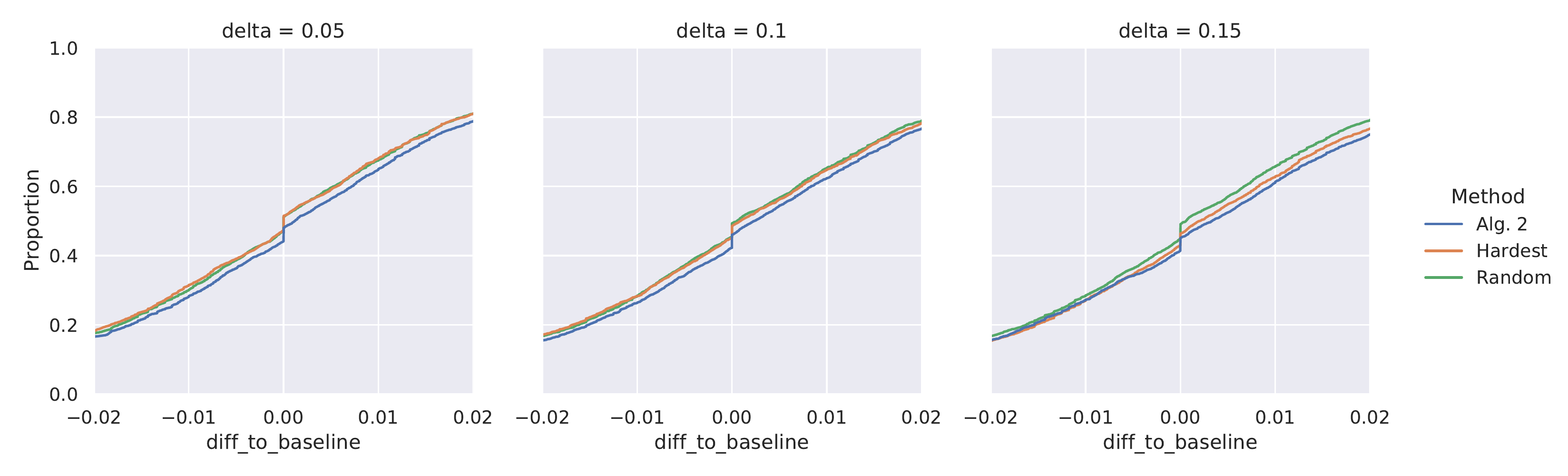}
     
    \put(0,4){
      \tikz{\path[draw=white, fill=white] (0, 0) rectangle (.3cm, 4cm)}
    }
    \put(0,3){\rotatebox{90}{
        \small $\P_{\text{user}} \left[\Delta_\text{Acc,User,Method} \le t \right]$}
    }
    \put(11,1){
      \tikz{\path[draw=white, fill=white] (0, 0) rectangle (4cm, .35cm)}
    }
    \put(38, 1){
      \tikz{\path[draw=white, fill=white] (0, 0) rectangle (4cm, .35cm)}
    }
        \put(65, 1){
      \tikz{\path[draw=white, fill=white] (0, 0) rectangle (4cm, .35cm)}
    }    
 
    \put(17,1){
      \small $t$ 
    }
    \put(46, 1){
       \small $t$ 
    }
        \put(73, 1){
       \small $t$ 
    }

    \put(13,27.5){
      \tikz{\path[draw=white, fill=white] (0, 0) rectangle (4cm, .3cm)}
    }
    \put(42, 27.5){
      \tikz{\path[draw=white, fill=white] (0, 0) rectangle (4cm, .3cm)}
    }
        \put(69, 27.5){
      \tikz{\path[draw=white, fill=white] (0, 0) rectangle (4cm, .3cm)}
    }    
    
    \put(15,28){
      \small $\delta = 0.05$
    }
    \put(42, 28){
      \small $\delta = 0.10$
    }
        \put(69, 28){
     \small $\delta = 0.15$
    }    
     \put(90,28){\textbf{A}}
  \end{overpic}   
  
      \begin{overpic}[scale=0.5, 
    ]{%
      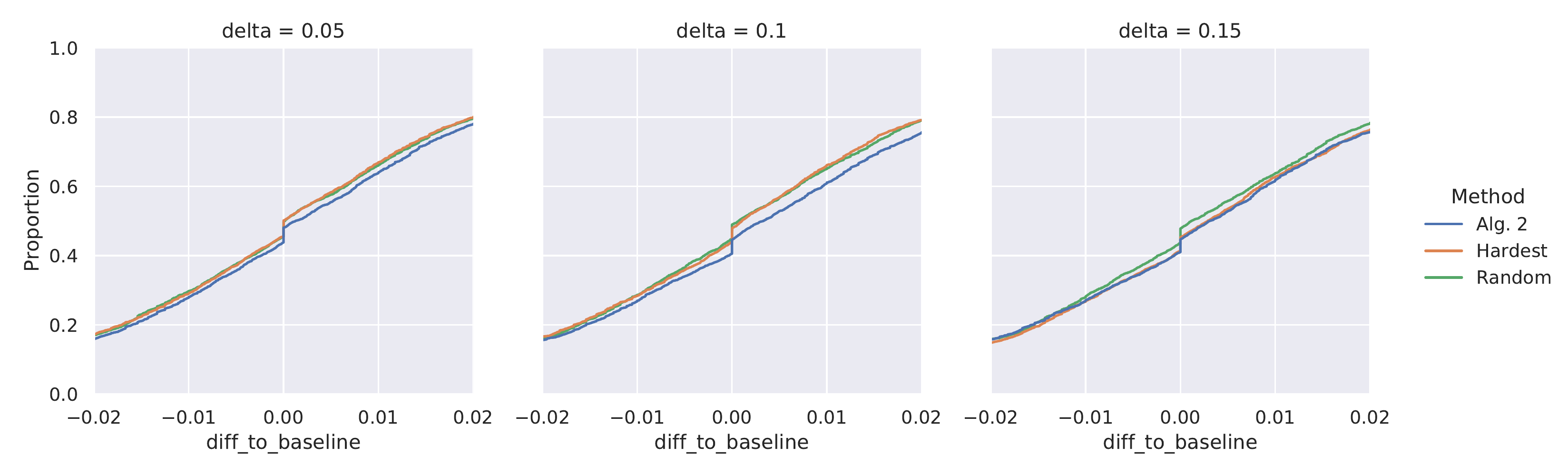}
     
    \put(0,4){
      \tikz{\path[draw=white, fill=white] (0, 0) rectangle (.3cm, 4cm)}
    }
    \put(0,3){\rotatebox{90}{
        \small $\P_{\text{user}} \left[\Delta_\text{Acc,User,Method} \le t \right]$}
    }
    \put(11,1){
      \tikz{\path[draw=white, fill=white] (0, 0) rectangle (3cm, .35cm)}
    }
    \put(38, 1){
      \tikz{\path[draw=white, fill=white] (0, 0) rectangle (4cm, .35cm)}
    }
        \put(65, 1){
      \tikz{\path[draw=white, fill=white] (0, 0) rectangle (4cm, .35cm)}
    }    
 
    \put(17,1){
      \small $t$ 
    }
    \put(46, 1){
       \small $t$ 
    }
        \put(73, 1){
       \small $t$ 
    }

    \put(13,27.5){
      \tikz{\path[draw=white, fill=white] (0, 0) rectangle (4cm, .3cm)}
    }
    \put(42, 27.5){
      \tikz{\path[draw=white, fill=white] (0, 0) rectangle (4cm, .3cm)}
    }
        \put(69, 27.5){
      \tikz{\path[draw=white, fill=white] (0, 0) rectangle (4cm, .3cm)}
    }    
    
    \put(15,28){
      \small $\delta = 0.05$
    }
    \put(46, 28){
      \small $\delta = 0.10$
    }
        \put(69, 28){
     \small $\delta = 0.15$
    }    
     \put(90,28){\textbf{B}}
  \end{overpic} 
  
    \begin{overpic}[scale=0.5, 
    ]{%
      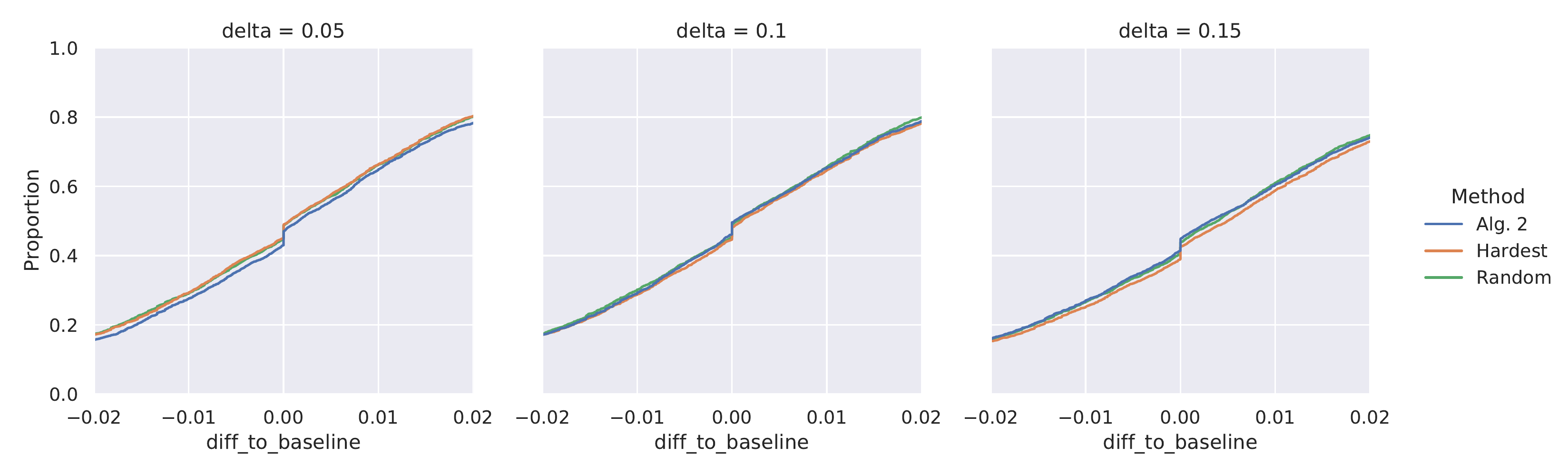}
    \put(0,4){
      \tikz{\path[draw=white, fill=white] (0, 0) rectangle (.3cm, 4cm)}
    }
    \put(0,3){\rotatebox{90}{
        \small $\P_{\text{user}} \left[\Delta_\text{Acc,User,Method} \le t \right]$}
    }
    \put(11,1){
      \tikz{\path[draw=white, fill=white] (0, 0) rectangle (4cm, .35cm)}
    }
    \put(38, 1){
      \tikz{\path[draw=white, fill=white] (0, 0) rectangle (4cm, .35cm)}
    }
        \put(65, 1){
      \tikz{\path[draw=white, fill=white] (0, 0) rectangle (4cm, .35cm)}
    }    
 
    \put(17,1){
      \small $t$ 
    }
    \put(42, 1){
       \small $t$ 
    }
        \put(73, 1){
       \small $t$ 
    }

    \put(13,27.5){
      \tikz{\path[draw=white, fill=white] (0, 0) rectangle (4cm, .3cm)}
    }
    \put(42, 27.5){
      \tikz{\path[draw=white, fill=white] (0, 0) rectangle (4cm, .3cm)}
    }
        \put(69, 27.5){
      \tikz{\path[draw=white, fill=white] (0, 0) rectangle (4cm, .3cm)}
    }    
    
    \put(15,28){
      \small $\delta = 0.05$
    }
    \put(42, 28){
      \small $\delta = 0.10$
    }
        \put(69, 28){
     \small $\delta = 0.15$
    }    
     \put(90,28){\textbf{C}}
  \end{overpic}   
 
    \caption{
	Results for the Amazon review dataset, with panel A being the fully supervised setting $\E \left[|Y_\text{weak}|\right] = 1$,  panel B having a small amount of weak supervision $\E \left[|Y_\text{weak}|\right] = 1.2$ and panel C having the noisiest labels $\E \left[|Y_\text{weak}|\right] = 1.5$.
	 Define for each user and each method $\Delta_\text{Acc,User,Method} \defeq \text{Accuracy}_{\text{Method}} - \text{Accuracy}_{\text{Baseline}}$ to be the difference in accuracy between the method and the baseline.
    We report the cumulative distribution function of that quantity over the set of out-of-domain users, hence lower c.d.f.s are better (as it means the distribution is stochastically larger).
    The ``Baseline" method has only seen data from $P_0^\text{train}$, so we expect all methods to improve upon it.
    We average each per-user accuracy over $M=5$ independent trials.
    }
    \label{fig:amazon_review_supervision_per_user}
\end{figure}
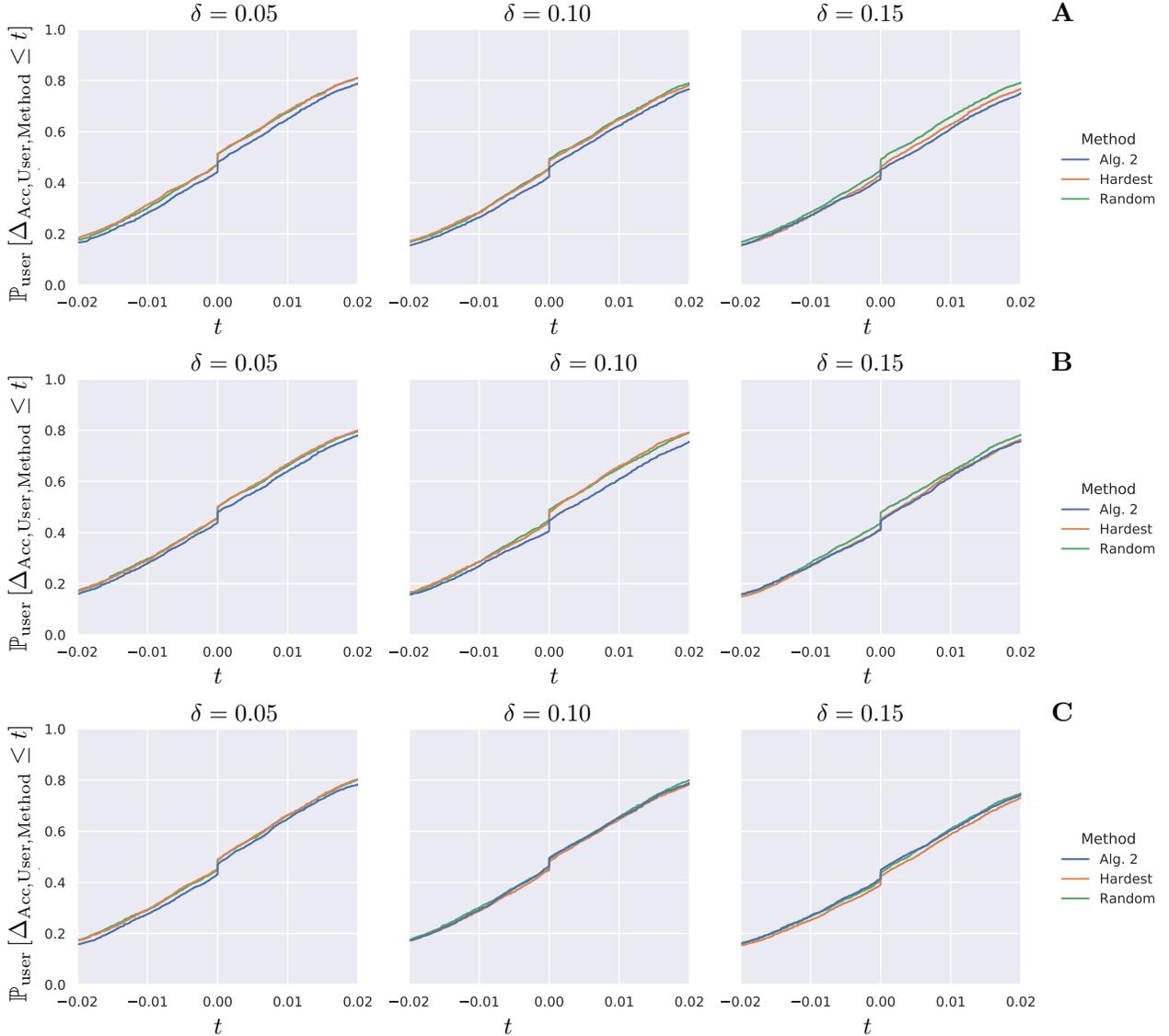

\clearpage

\section{Discussion}
\label{sec:conclusion}

We proposed inferential methodology for the localization and detection of subpopulations present in a data stream.  Though we focused heavily on the implications for model maintenance, the underlying ideas apply more broadly, and reduce to familiar existing methodology in special cases.  For example, when the class $\mc R$ is completely unstructured, i.e., $\mc R = 2^{[n]}$, then Algorithm \ref{alg:p-filter} essentially reduces to the Benjamini-Hochberg-type proposal of \citet{BatesCaLeRoSe21}, for unstructured one-class outlier detection.  On the other hand, when the class $\mc R$ is highly structured, e.g., satisfying certain geometric or graph-theoretic criteria, and we are additionally willing to make certain (parametric) assumptions about the data-generating process, then Algorithm \ref{alg:recovery} roughly becomes the familiar scan statistic (e.g., \citet{Kulldorff97,SharpnackKrSi13}).

There are other seemingly natural methodological approaches that we might have pursued.  As we mentioned in Section \ref{sec:intro}, two-sample testing is intimiately connected to the ideas in the current paper, and the well-known Kolmogorov-Smirnov test \citep{Kolmogorov33,Smirnov48,Andrews97} is probably one of the most widely used tools for nonparametric hypothesis testing.  However, it is not immediately clear (at least to us) how we might modify the Kolmogorov-Smirnov test to work without making strong distributional assumptions about the underlying black box machine learning model, or for the purpose of localization.  Additionally, it is reasonable to suggest that we use clustering for subgroup estimation, as part of the three-step approach to model refitting that we described in Section \ref{sec:refitting}.  However, it is also not clear what the type 1 and 2 error rates of such a procedure might be (and, moreover, how to control them).  Nonetheless, exciting recent work has drawn connections between classification and two-sample testing \citep{KimRaSiWa21}, and therefore this may indeed be a fruitful direction to investigate.

Finally, a direction that seems interesting to pursue is developing a truly sequential version of the methodology we laid out here, i.e., to marry the ideas from the broad literature on sequential testing \citep{Balsubramani14a,JohariPeWa15,JohariKoPeWa17,HowardRaMcSe20,HowardRaMcSe21}, with the ones in the current paper.  More generally, we hope the methodology in this paper motivates others to consider the many challenges related to real-time monitoring and maintenance.

\section*{Acknowledgements}
We thank Guenther Walther for helpful and encouraging comments on a draft of the paper.

\clearpage
\appendix


\newcommand{\indep}{\perp \!\!\! \perp}
\newcommand{\cordist}{d_{\textup{cor}}}
\newcommand{\dcor}{\cordist}

\newcommand{\reg}{\mathsf{reg}}
\newcommand{\Rcorset}{\mc{R}_{\textup{cor}}}

\section{Proofs}
\label{sec:appendix}

\subsection{Proof of Theorem \ref{thm:scan-recovery-error}}
\label{sec:proof-of-thm-scan-recovery-error}

We begin with a bit of notation. For $i \in [n]$, we have $Z_i = \mu
\indic{i \in R^\star} + \sigma \noise_i$ where $\noise_i \simiid
\normal(0,1)$.  Recalling that $\mc{R} = \RXs \cap \{X_i\}_{i=1}^n$,
for each $R \in \mc{R}$ we
define the localized noise
\begin{align*}
  \noise(R) \defeq \frac{1}{\sigma \sqrt{|R|}} \sum_{i \in R} \noise_i,
\end{align*}
which are marginally standard normal, and the following correlation
distance on sets in $\mc{R}$:
\begin{align*}
  \cordist(R_1, R_2)^2 \defeq
  \half \E[(\noise(R_1) - \noise(R_2))^2]
  = 1 - \frac{|R_1 \cap R_2|}{\sqrt{|R_1||R_2|}}.
\end{align*}

The starting point of our
proof is a type of basic inequality~\cite[cf.][]{Wainwright19}
relating the error in recovering $R\opt$ to penalized deviations
of $\noise(R)$. Recall that
$\Rest$ maximizes $Z_R - \sigma \reg(R)$ for the
penalty function $\reg(R) = C \sqrt{d \log \frac{en}{|R| \vee d}}$,
so that
by maximality of $\Rest$, we have
\begin{align*}
  \mu \frac{|\Rest \cap R\opt|}{\sqrt{|\Rest|}}
  + \sigma \noise(\Rest) - \sigma \reg(\what{R})
  & = Z_{\Rest} - \sigma \reg(\what{R}) \\
  & \ge Z_{R\opt} - \sigma \reg(R\opt)
  = \mu \sqrt{|R\opt|}
  + \sigma \noise(R\opt) - \sigma \reg(R\opt)
\end{align*}
Dividing by $\sqrt{|R\opt|}$ and rearranging, this is equivalent to
the basic inequality
\begin{equation}
  \label{eqn:basic-inequality}
  \cordist^2(\Rest, R\opt)
  \le \frac{\sigma}{\mu \sqrt{|R\opt|}}
  \left(\reg(R\opt) - \reg(\what{R}) + \noise(\what{R}) - \noise(R\opt)
  \right).
\end{equation}

We now proceed with a peeling argument by controlling the deviation on
the right-hand-side of the basic inequality~\eqref{eqn:basic-inequality}
over $\cordist$-balls around $R\opt$. Our first
step is to exhibit an equivalence between Hamming and correlation distances.
(See Sec.~\ref{sec:proof-dh-vs-cor} for a proof.)
\begin{lemma}
  \label{lem:dh-vs-cor}
  Let $R_1, R_2 \in \mc{R}$. Then
  \begin{equation*}
    \cordist^2(R_1, R_2) \le
    \frac{\dham(R_1, R_2)}{\max\{|R_1|, |R_2|\}}.
  \end{equation*}
  If additionally $\cordist^2(R_1, R_2) \le \half$, then
  \begin{align*}
    \frac{\dham(R_1, R_2)}{3\min\{|R_1|, |R_2|\}}
    \le \cordist^2(R_1, R_2).
  \end{align*}
\end{lemma}
To perform our peeling argument, for $\delta \in [0, 1]$ we define the
sets
\begin{equation*}
  \Rcorset(\delta) \defeq \left\{R \in \mc{R} \mid \cordist(R, R\opt)
  \le \delta \right\},
\end{equation*}
and for all $\ell \in \{1, \dots, n \}$, 
\begin{align*}
\mc{R}_\ell \defeq \left\{R \in \mc{R} \mid |R| \ge \ell \right\}.
\end{align*}
Using an entropy integral
bound~\cite[e.g.][Ch.~5.3]{Wainwright19}, we can then claim the following
lemma, whose proof we defer to
Section~\ref{sec:proof-correlation-dudley-integral}.
\begin{lemma}
  \label{lemma:correlation-dudley-integral}
  There exists a numerical constant $C$ such that, for $r \in [0, 1]$ and $\ell \in [n]$,
  \begin{equation*}
    \E\left[\sup_{R \in \Rcorset(r) \cap \mc{R}_\ell}
      \left|\noise(R) - \noise(R\opt)\right|\right]
    \le C \left\{d r^2\log\frac{en}{(d  \vee r^2 \ell)} \right\}^{1/2}.
  \end{equation*}
\end{lemma}

We combine the expectation bounds in
Lemma~\ref{lemma:correlation-dudley-integral} with Gaussian
Lipschitz concentration inequalities, along with the basic
inequality~\eqref{eqn:basic-inequality}, to obtain our final
desired result.  For any fixed $R_1, R_2$, we have
\begin{equation*}
  \ltwobigg{\frac{1}{\sqrt{|R_1|}} \ones_{R_1}
    - \frac{1}{\sqrt{|R_2|}} \ones_{R_2}}^2
  = 2 \dcor^2(R_1, R_2),
\end{equation*}
so that the function $f_R(z) \defeq \frac{1}{\sqrt{|R|}} \ones_R^T z -
\frac{1}{\sqrt{|R\opt|}} \ones_{R\opt}^T z$ is $\sqrt{2} r^2$-Lipschitz for
all $R \in \Rcorset(r)$. As a consequence, the concentration of Lipschitz
functions of Gaussian vectors~\cite[e.g.][Thm.~2.26]{Wainwright19} yields
that there exists a numerical constant $C$ such that for any $r \in [0, 1]$, $\ell \in [n]$,
and $t > 0$, we have
\begin{equation}
\label{eqn:prob-bound-on-sup-cor-min-size}
  \P\left(\sup_{R \in \Rcorset(r) \cap \mc{R}_\ell}
  \left|\noise(R) - \noise(R\opt)\right|
  \ge C \sqrt{d r^2 \log \frac{en}{d \vee (r^2 \ell)}
    + r^2 t} \right) \le \exp(-t).
\end{equation}
Lemma~\ref{lem:dh-vs-cor} additionally implies $\dham(R, R\opt) \le
  3 k \dcor^2(R, R\opt)$ for all $R \in \Rcorset(1/\sqrt{2})$, where
  we recall $k = |R\opt|$.  
  In particular, $|R| \ge |R\opt|(1 - 3 r^2) \ge \frac{k}{2}$, meaning $\Rcorset(r) \cap \mc{R}_{k/2} = \Rcorset(r)$ whenever $r^2 \le 1/6$.
By taking $t_i = \log \frac{2^i}{\delta}$ in the preceding display,
we sum over $i \ge 1$, obtaining the following uniform concentration
guarantee, which we state as a lemma.
\begin{lemma}
  \label{lemma:to-peel}
  Let $\{r_i\}_{i \ge 1} \subset [0, 1/\sqrt{6}]$ be any sequence. Then
  with probability at least $1 - \delta$,
  \begin{equation*}
    \sup_{R \in \Rcorset(r_i)} |\noise(R) - \noise(R\opt)|
    \le C \sqrt{d r_i^2 \log \frac{en}{r_i^2 k}
      + r_i^2\left(i + \log \frac{1}{\delta}\right)}
  \end{equation*}
  simultaneously for all $i \in \N$.
\end{lemma}

Before moving into the actual peeling argument,  we see that Lemma~\ref{lemma:to-peel} is only applicable when $r_i \le 1/\sqrt{6}$, hence we must first prove that, when chosen accordingly, the size penalty ensures that we have $\Rest \in \Rcorset(1/\sqrt{8}) \subset \Rcorset(1/\sqrt{6})$ with probability at least $1-\delta$.

For any $\ell \in [n]$, taking $r=1$ in equation~\eqref{eqn:prob-bound-on-sup-cor-min-size} yields
\begin{align*}
\P\left[  \sup_{R \in \mc{R}_\ell} \left\{ \noise(R) - \noise(R\opt) \right\} \ge  C_1 \left\{d \log\frac{en}{d \vee \ell} + \log \frac{1}{\delta}\right\}^{1/2} \right] \le \delta.
\end{align*}
Let $J = \lfloor \log(n/d) \rfloor$,  and apply the above inequality with $\ell_1 = 1, \ell_2=d,  \ell_3=d e, \dots,  \ell_{J+2} = d e^J$ respectively, and $\delta_1 = \delta e^{-(J+2)}, \dots, \delta_{J+2} = \delta e^{-1}$. By an union bound, we see that
\begin{align*}
\P\left[  \sup_{1 \le i \le J+2}  \sup_{R \in \mc{R}_{\ell_i}} \left\{ \noise(R) - \noise(R\opt) \right\} \ge  C_0 \left\{d \log\frac{en}{d \vee \ell_i} +  (J+2-i) + \log \frac{1}{\delta}\right\}^{1/2} \right] \le \delta.
\end{align*}
On the complement of this event,  for each $R \in \mc{R}$ such that $|R| \ge d$, we have $|R| \ge \ell_i=de^{i-2}$ for $i = \lfloor \log \frac{|R|}{d} \rfloor + 2$,  therefore 
\begin{align*}
\noise(R) - \noise(R\opt) &\le C_0 \left\{d \log\frac{en}{\ell_i} + \lfloor \log \frac{n}{d} \rfloor -  \lfloor \log \frac{|R|}{d} \rfloor + \log \frac{1}{\delta}\right\}^{1/2} \\
&\le C' \left\{d \log\frac{en}{|R|} + \log \frac{1}{\delta}\right\}^{1/2}
\end{align*}
for some universal constant $C' \ge 4C_0$. As a result, with probability at least $1-\delta$, for all $R \in \mc{R}$, it holds that
\begin{align}
\label{eqn:sup-noise-bounded-penalty}
\noise(R) - \noise(R\opt) - \reg(R) \le C' \sqrt{d \log\frac{en}{|R| \vee d} + \log \frac{1}{\delta}} - C\sqrt{d \log\frac{en}{|R| \vee d}}.
\end{align}
Combine now the uniform inequality~\eqref{eqn:sup-noise-bounded-penalty} with the basic inequality~\eqref{eqn:basic-inequality}, and assume that we choose to run Alg.~\ref{alg:recovery} with $C \ge C'$: with probability at least $1-\delta$, we must have
\begin{align*}
\cordist^2(\Rest, R\opt) &\le \frac{\sigma}{\mu \sqrt{|R\opt|}}\left( \reg(R\opt) + C' \sqrt{d \log\frac{en}{|\Rest| \vee d} + \log \frac{1}{\delta}} - C\sqrt{d \log\frac{en}{|\Rest| \vee d}} \right) \\
&\le \frac{\sigma}{\mu \sqrt{|R\opt|}}\left( \reg(R\opt) + C \sqrt{\log \frac{1}{\delta}} \right)
\end{align*}
hence if $\mu \ge \frac{8\sigma}{ \sqrt{|R\opt|}}\left( \reg(R\opt) + C \sqrt{\log(1/\delta)}\right)$, then we have $\cordist^2(\Rest, R\opt) \le 1/8$
with probability at least $1-\delta$.

The final step before performing our peeling argument on
the basic inequality~\eqref{eqn:basic-inequality} is to control
the deviations in the penalty terms $\reg(R)$. For this,
we have the nearly trivial bound that
\begin{equation}
  \label{eqn:r-penalty-bound}
  \reg(R\opt) - \reg(R)
  \le \frac{3}{2} \sqrt{d} \cdot \dcor^2(R, R\opt).
\end{equation}
Indeed, let $l = |R|$ and $k = |R\opt|$. When
$l \le k$, the result is trivial. When $l > k$, we use that
$\sqrt{a + b_0} \le \sqrt{a + b_1} + \frac{b_0 - b_1}{2 \sqrt{a + b_1}}$
by concavity
of $\sqrt{\cdot}$, and so
\begin{align*}
  \frac{1}{C} \left(\reg(R\opt) - \reg(\what{R})\right)
  & = \sqrt{d + d \log \frac{n}{k}} - \sqrt{d + d \log \frac{n}{l}}
  \le \frac{d \log \frac{l}{k}}{2 \sqrt{d + \log\frac{n}{l}}}
  \le \half \sqrt{d} \log \frac{l}{k}.
\end{align*}
Then we simply note that $\log \frac{l}{k}
= \log(1 + \frac{l - k}{k}) \le \frac{l - k}{k}
\le \frac{\dham(R, R\opt)}{k}$ and
apply Lemma~\ref{lem:dh-vs-cor}.

We can now apply a peeling argument. 
For $i = 4, 5, \ldots, 2 \log
(\frac{\mu}{\sigma} n)$,
define the shells $\mc{R}_i = \{R \in \mc{R} \mid 2^{-i} < \dcor^2(R, R\opt)
\le 2^{-i + 1}\}$. Use the shorthand
$\Delta^2 = \dcor^2(\what{R}, R\opt)$.
Then applying Lemma~\ref{lemma:to-peel} to
the shells $\mc{R}_i$, we combine
inequality~\eqref{eqn:r-penalty-bound} and the basic
inequality~\eqref{eqn:basic-inequality} to yield that
there exists a numerical constant $C$ such that, with
probability at least $1 - 2\delta$,
either $\Delta^2 \le \frac{\sigma^2}{\mu^2 n^2}$
or
\begin{align}
  \nonumber
  \Delta^2
  & \le C \frac{\sigma}{\mu \sqrt{k}}
  \sqrt{d \Delta^2 \log \frac{n}{k \Delta^2}
    + \Delta^2\left(\log\frac{\mu n}{\sigma} + \log \frac{1}{\delta}\right)}
  + C \frac{\sigma}{\mu} \sqrt{\frac{d}{k}} \cdot \Delta^2 \\
  & \le C \frac{\sigma}{\mu}
  \sqrt{\frac{d}{k}}
  \sqrt{\Delta^2 \log \frac{n}{k \Delta^2}
    + \Delta^2 \frac{\log \frac{n \mu}{\sigma \delta}}{d}}
  + C \frac{\sigma}{\mu}
  \sqrt{\frac{d}{k}} \Delta^2.
  \label{eqn:the-new-basics}
\end{align}
(Note that if $\what{R} \in \mc{R}_i$, we have
$\Delta^2 > 2^{-i}$, and $2^{-i+1} \le 1/8<1/6$.)

It is relatively straightforward to bound those
values $\Delta$ satisfying inequality~\eqref{eqn:the-new-basics}.
Indeed, by assumption in the theorem we have
$\frac{\sigma}{\mu} \sqrt{\frac{d}{k}} \le c$ for a (small) constant
$c$, subtracting $C \frac{\sigma}{\mu} \sqrt{\frac{d}{k}} \Delta^2$ from
each side of inequality~\eqref{eqn:the-new-basics} and dividing through
by $\Delta > 0$ yields
\begin{equation*}
  \Delta
  \le C \frac{\sigma}{\mu} \sqrt{\frac{d}{k}}
  \sqrt{\log \frac{n}{k \Delta^2} + \frac{1}{d} \log \frac{n \mu}{\sigma \delta}}
  = C \frac{\sigma}{\mu}
  \sqrt{\frac{d}{k}}
  \sqrt{2\log \frac{1}{\Delta}
    + \log\frac{n}{k}
    + \frac{1}{d} \log \frac{n \mu}{\sigma \delta}}.
\end{equation*}
We use the following observation:
\begin{observation}
  \label{observation:circularity}
  Let $0 < a \le 1/\sqrt{e}$. If $\Delta \le a \sqrt{\log \frac{1}{\Delta} +
    b}$, then $\Delta \le a \sqrt{2 \max\{b, \log \frac{1}{a}\}}$.
\end{observation}
\begin{proof}
  We provide the proof by contradiction. Assume that $\Delta > a
  \sqrt{2\max\{b, \log \frac{1}{a}\}}$, and consider two cases. In the
  first, assume that $\log \frac{1}{a} > b$, so that $\Delta > a \sqrt{2\log
  \frac{1}{a}}$.  Then by assumption, we have
  \begin{equation*}
    \sqrt{2 \log \frac{1}{a}} \le \sqrt{\log\frac{1}{\Delta} + b}
    \le \sqrt{\log\frac{1}{a} - \half \log \left(2 \log \frac{1}{a}\right)
      + b}
    < \sqrt{2 \log \frac{1}{a}},
  \end{equation*}
  a contradiction.
  Alternatively, assume $b \ge \log \frac{1}{a}$, so that
  $\Delta > a \sqrt{2 b}$. Then again by assumption, we have
  \begin{equation*}
    \sqrt{2 b} \le \sqrt{\log \frac{1}{\Delta} + b}
    < \sqrt{\log \frac{1}{a} - \half \log(2b) + b}
    \le \sqrt{2 b},
  \end{equation*}
  where we have used that $b \ge \log \frac{1}{a} \ge \half$.
  Again, this is a contradiction.
\end{proof}

Substituting the bound in Observation~\ref{observation:circularity}
into the preceding display, we obtain that
\begin{equation*}
  \dcor^2 (\what{R}, R\opt)
  \le C \frac{\sigma^2}{\mu^2} \frac{d}{k}
  \max\left\{\log \frac{k \mu^2 }{d \sigma^2},
  \frac{1}{d} \log \frac{n \mu}{\sigma \delta}
  + \log \frac{n}{k}
  \right\}.
\end{equation*}
Making a simplifying calculation to remove the lower order terms
$\frac{1}{d} \log \frac{n \mu}{\sigma}
\lesssim \frac{1}{d} \log \frac{n}{k}
+ \frac{1}{d} \log \frac{k \mu^2}{\sigma^2}$, this implies that for a numerical
constant $C''$, we have with probability at least $1 - 2\delta$ that
\begin{equation*}
  \dcor^2(\what{R}, R\opt)
  \le C'' \frac{\sigma^2}{\mu^2 k} 
  \left[ d \left(\log \frac{k \mu^2}{\sigma^2}
    + \log \frac{n}{d k}\right)
    + \log \frac{1}{\delta}\right].
\end{equation*}

\subsubsection{Proof of Lemma~\ref{lem:dh-vs-cor}}
\label{sec:proof-dh-vs-cor}

We assume without loss of generality that $|R_1| \ge |R_2|$. Then
the first inequality follows the observation that
\begin{align*}
  1 - \cordist^2(R_1, R_2)
  = \frac{|R_1 \cap R_2|}{\sqrt{|R_1| |R_2|}}
  \ge
  \frac{|R_1 \cap R_2|}{|R_1|}
  \ge 1 - \frac{|R_1 \setdiff R_2|}{|R_1|}
  = 1 - \frac{\dham(R_1, R_2)}{\max\{|R_1|, |R_2|\}}.
\end{align*}
For the second,
let $\cordist^2(R_1, R_2) = \delta_{12} \le \half$. Then we observe that
\begin{align*}
  1 - \delta_{12}
  = \frac{|R_1 \cap R_2|}{\sqrt{|R_1||R_2|}}
  & =\half \frac{|R_1| + |R_2| - |R_1 \setdiff R_2|}{\sqrt{|R_1||R_2|}} \\
  & = \half\left(\sqrt{\frac{|R_1|}{|R_2|}} +\sqrt{\frac{|R_2|}{R_1|}} \right)
  - \frac{|R_1 \setdiff R_2|}{\sqrt{|R_1||R_2|}},
\end{align*}
which is equivalent, with some rearrangement to 
\begin{align*}
  \frac{|R_1 \setdiff R_2|}{|R_2|} &= \delta_{12}\sqrt{|R_1|/|R_2|} + \half\left(\sqrt{|R_1|/|R_2|} - 1\right)^2
\end{align*}
On the other hand, we have $|R_2| \ge |R_1 \cap R_2| \ge (1-\delta_{12})
\sqrt{|R_1||R_2|}$, which directly implies that
$\sqrt{|R_1|/|R_2|} \le \frac{1}{1-\delta_{12}} \le 1 + 2 \delta_{12}$
as $\delta_{12} \le \half$.
We conclude that
\begin{align*}
  \frac{\dham(R_1, R_2)}{\min\{|R_1|, |R_2|\}}
  =
  \frac{|R_1 \setdiff R_2|}{|R_2|}
  \stackrel{(\star)}{\le}
  \frac{\delta_{12}}{1 - \delta_{12}}
  + \half \frac{\delta_{12}^2}{(1 - \delta_{12}^2)^2}
  \le
  \delta_{12} + 4\delta_{12}^2 \le 3 \delta_{12}
  = 3 \dcor^2(R_1, R_2),
\end{align*}
which is equivalent to the lemma.

\subsubsection{Proof of Lemma~\ref{lemma:correlation-dudley-integral}}
\label{sec:proof-correlation-dudley-integral}

Before beginning the proof proper, we state a simple observation
we will use frequently.
\begin{observation}
  \label{observation:gamma-integral}
  We have $\int_0^\delta \sqrt{\frac{1}{t} \log \frac{1}{t+y}} dt
  \le 4 \sqrt{\delta \log \frac{1}{\max(\delta,y)}}$ for all $\delta \le 1/e$ and $y>0$.
\end{observation}
\begin{proof}
The result is obvious when $y>\delta$, as we have $\log \frac{1}{t+y} \le \log \frac{1}{y}$ and $\int_0^\delta \sqrt{1/t}dt = 2 \sqrt{\delta}$. 

We now focus on the case $\delta \le y$, and use two arguments. First, \citet[Eq.~(2.5)]{BorweinCh09} gives
  bounds on the upper Gamma integral that
  $\int_x^\infty e^{-t} t^{\alpha - 1} dt \le B x^{\alpha - 1} e^{-x}$ whenever
  $B > 1$ and $x > \frac{B}{B -1} (\alpha - 1)$. Thus,
  in our initial integral with $y=0$,  noting that the integral is decreasing in $y$,  we make the substitution
  $u = \log\frac{1}{t}$, which gives
  \begin{equation*}
    \int_0^\delta \sqrt{\frac{1}{t} \log \frac{1}{t}} dt
    = \int_{\log\frac{1}{\delta}}^\infty e^{-u/2} \sqrt{u} du
    = 2 \sqrt{2} \int_\frac{\log\frac{1}{\delta}}{2}^\infty
    \sqrt{t} e^{-t} dt
    \le 4
    \sqrt{\log \frac{1}{\delta}}
    \sqrt{\delta},
  \end{equation*}
  where we have used $B = 2$ and $\alpha = \frac{3}{2}$, assuming
  $\log \frac{1}{\delta} > 1$.
\end{proof}

Now, for a distance $d$ on $\mc{R}$, let $N(\mc{R}, d, t)$ be the
$t$-covering number of $\mc{R}$ in distance $d$.  As $\noise(R) -
\noise(R\opt)$ is a Gaussian process with $\E[(\noise(R) - \noise(R\opt))^2]
= 2 \dcor^2(R, R\opt)$, Dudley's entropy
integral~\cite[Thm.~5.22]{Wainwright19} then immediately gives that
\begin{equation}
  \label{eqn:entropy-integral}
  \E\left[\sup_{R \in \Rcorset(r)} |\noise(R) - \noise(R\opt)|\right]
  \lesssim \int_0^r \sqrt{\log N(\Rcorset(r), \dcor, t)} dt.
\end{equation}
We use Lemma~\ref{lem:dh-vs-cor} to relate the covering numbers in
correlation distance and Hamming distance, which allows us to apply
standard VC-covering bounds for discrete sets to compute the
integeral.

By~\citet[Theorem~1]{Haussler95}, for all $t \in [0, n]$ we have
\begin{equation*}
  \log N(\mc{R}, \dham, t)
  \le d \log \frac{2e (n+1)}{t+2d+2} + \log(e(d + 1))
  \lesssim d \log \frac{en}{t+d}
\end{equation*}
as $\mc{R}$ has VC-dimension $d$ (and for $\epsilon > 1$, we have
$\log N(\mc{R}, \dham, n \epsilon) = 0$).
Then by Lemma~\ref{lem:dh-vs-cor}, we have that
\begin{equation*}
  \dcor^2(R_1, R_2) \le \frac{\dham(R_1,R_2)}{\ell}
\end{equation*}
for all $R_1, R_2 \in \mc{R}_\ell$, 
and so for all $r \in [0, 1]$ and $t \in [0, 1]$, we have the covering
number bound
  \begin{equation}
  \label{eqn:set-cover-bounds}
    \log N(\Rcorset(r) \cap \mc{R}_\ell, \dcor, t)
    \le \log N(\mc{R}, \dham,  \ell t^2)
    \lesssim d \log\frac{en}{\ell t^2+d}.
  \end{equation}

We use the bound~\eqref{eqn:set-cover-bounds} to control the entropy
integral~\eqref{eqn:entropy-integral}:
\begin{align*}
  \E\left[\sup_{R \in \Rcorset(r) \cap \mc{R}_\ell} |\noise(R) - \noise(R\opt)|\right]
  & \lesssim \int_0^r \sqrt{\log N\Big(\mc{R}, \dham,  \ell t^2 \Big)}
  dt \\
  & \lesssim \int_0^r \sqrt{d \log \frac{en}{d + \ell t^2}} dt
  = \sqrt{\frac{e n d}{4k}}
  \int_0^\frac{k r^2}{en} \sqrt{\frac{1}{u} \log \frac{1}{u+d/en}} du \\
  & \lesssim
  \sqrt{\frac{n d}{\ell}}
  \sqrt{\frac{\ell r^2}{n} \log \frac{en}{d \vee (\ell r^2)}}
  = r \sqrt{d \log \frac{en}{d \vee (\ell r^2)}}
\end{align*}
where we used the substitution $u = \frac{k t^2}{n}$ in the first equality
and Observation~\ref{observation:gamma-integral} for the final inequality.
two displays gives Lemma~\ref{lemma:correlation-dudley-integral}.

\subsection{Proof of Theorem~\ref{thm:recovery-error-lower-bound}}
\label{sec:proof-of-thm-recovery-error-lower-bound}

The theorem uses a reduction of estimation to testing via either Fano's
inequality or Assouad's method (see, e.g.,~\cite[Ch.~15]{Wainwright19}
or~\cite{Yu97}). We begin by stating the two main lemmas
we use on the error of multiple hypothesis tests.
\begin{lemma}[Fano's inequality]
  \label{lemma:fano}
  Let $\mc{V}$ be an arbitrary set, $V \sim \uniform(\mc{V})$, and
  $\{Y_i\}_{i = 1}^n$ be random variables. Then for any
  function $\what{V}(Y_1^n)$, we have
  \begin{equation*}
    \P(\what{V}(Y_1^n) \neq V) \ge 1 - \frac{I(V; Y_1, \ldots, Y_n) + \log 2}{
      \log|\mc{V}|}.
  \end{equation*}
\end{lemma}

\begin{lemma}[Assouad's lemma]
  \label{lemma:assouad}
  Let distributions $P_v$ on a random variable $Y$ be indexed by vectors $v
  \in \{0, 1\}^d$ and define $\wb{P}_j = \frac{1}{2^{d-1}} \sum_{v : v_j =
    1} P_v$ and $\wb{P}_{-j} = \frac{1}{2^{d-1}} \sum_{v : v_j = 0}
  P_v$. Let $V \sim \uniform\{\pm 1\}^d$, and conditional on $V = v$,
  draw $Y \sim P_v$.
  Then for any estimator $\what{v}$,
  \begin{equation*}
    \E\left[\lone{\what{v}(Y) - V}\right] \ge \half
    \sum_{j = 1}^d \left(1 - \tvnorm{\wb{P}_j - \wb{P}_{-j}}\right),
  \end{equation*}
  where the expectation $\E$ is taken jointly over $V$ and $Y$.
\end{lemma}

To prove each of the results in Theorem~\ref{thm:recovery-error-lower-bound},
we work conditionally on $\{ X_i \}$, and for notational
convenience, we let $R$ designate both the region $R \subset \mc{X}$ and the
subset of indices $\{ i \in [n] \mid X_i \in R\} \subset [n]$, with the
meaning clear from context. In each case, embed the estimation problem into
a testing problem roughly as follows: we first construct a collection of
vectors $\mc{V} \subset \{0, 1\}^n$, where each $v \in \mc{V}$ satisfies
$\lone{v} = k$, where $\mc{V}$ has bounded VC-dimension. (We follow
standard practice~\cite{Haussler95} and say that a
subset $\mc{V} \subset \{0, 1\}^n$ has VC-dimension $d$ under the following
conditions: for index sets $J = (i_1, \ldots, i_k) \subset [n]$, let
$\mc{V}_J = \{(v_{i_1}, \ldots, v_{i_k}) \mid v \in \mc{V}\}$; then
$\vc(\mc{V})$ is the size of the largest subset $J \subset [n]$ such that
$\mc{V}_J = \{0, 1\}^{|J|}$.)
We choose the collection of regions
$\RXs$ so that the vectors
\begin{equation}
  \left\{ \left\{\indic{X_i \in R}\right\}_{i \in [n]} \right\}_{R \in \RXs}
  = \mc{V},
  \label{eqn:R-region-construction}
\end{equation}
indexing the regions $\RXs$ and $\mc{R}$ via $R_v$ for
$v \in \mc{V}$, so that $X_i \in R_v$ if and only if $v_i = 1$.
We may evidently do this while satisfying $\vc(\RXs) \le \vc(\mc{V})$.
For each $R \in \mc{R}$, we let $\P_{R}$ be the probability
distribution for which
\begin{equation}
  \label{eqn:Z-R-distribution}
  Z_i \mid X_i \sim \begin{cases}
    \normal(\mu,\sigma^2) & \mbox{if}~ i \in R, \\
    \normal(0,\sigma^2) & \mbox{otherwise},
  \end{cases}
\end{equation}
independently. We then have an immediate reduction: let
$V \sim \uniform(\mc{V})$, and conditional on $V = v$, set
$R\opt = R_v$ and draw $Z$ from the model~\eqref{eqn:Z-R-distribution}.
Then for a given estimator $\what{R}$, defining
$\what{v} \defeq \{\indics{X_i \in \what{R}}\}_{i = 1}^n$, if
$R\opt$ is chosen uniformly from $\mc{R}$ then
\begin{equation*}
  \P\left(|\what{R} \setdiff R\opt| \ge t\right)
  = \P\left(\lone{\what{v} - V} \ge t \right)
  ~~ \mbox{and} ~~
  \E\left[|\what{R} \setdiff R\opt|\right]
  \ge \E\left[\lones{\what{v} - V}\right],
\end{equation*}
the former inequality holding for all $t$. As such, any lower bound on the
probability or expectation of error in estimating $V$ bounds that in
estimating $R\opt$.

With this setting, we consider two regimes: the ``low
signal-to-noise (SNR)'' regime, when $\frac{\sigma^2}{\mu^2}$ is large, and
the ``high SNR'' regime, when $\frac{\mu^2}{\sigma^2}$ is large.  We begin
with the former.

\subsubsection*{Low SNR Regimes}

We first consider the case that $\frac{\mu^2}{\sigma^2} \le c \log(n - k +
1)$, and we will apply Fano's method. The main challenge is describing a
large and well-separated collection of vectors with a given VC-dimension.
We have the following lemma, which analogizes \citeauthor{Haussler95}'s
development of packing number bounds on the Boolean
$n$-cube~\cite[Thm.~2]{Haussler95} but allows each vector $v \in \mc{V}$ to
have a prescribed cardinality.

\begin{lemma}
  \label{lemma:packing-vc-set}
  Let $n, k, d \in \N$ satisfy
  $d \le k \le \frac{n}{2}$.
  There exists a numerical constant $c > 0$ such that the following holds:
  there is a set $\mc{V} \subset \{ 0, 1\}^n$ with
  $\vc(\mc{V}) = 2d$, $\lone{v} = k$ for each $v \in \mc{V}$,
  and $\ell_1$-packing number
  \begin{equation*}
    M(\mc{V}, \lone{\cdot}, k/2) \ge \exp\left(c \cdot d \log \frac{n}{k}
    \right).
  \end{equation*}
\end{lemma}
\noindent
The proof is technical, so we defer it further to
Appendix~\ref{sec:proof-packing-vc-set}.

Using Lemma~\ref{lemma:packing-vc-set}, we can relatively easily construct
a packing set satisfying the following:
\begin{lemma}
  \label{lemma:packing-properties}
  There exists a numerical constant $c_0 > 0$
  such that for each $1 \le t \le k$,
  there is a set $\mc{V} \subset \{0, 1\}^n$ satisfying
  the following:
  (i) $\vc(\mc{V}) = 2(d \wedge t)$,
  (ii) $\log|\mc{V}| \ge c_0 \cdot(d \wedge t) \log \frac{n - k + t}{t}$
  and $\log|\mc{V}| \ge 2 \log 2$,
  (iii) for each $v \neq w \in \mc{V}$ we have
  $\half t \le \lone{v - w} \le 2t$,
  and (iv) $\lone{v} = k$ for each $v \in \mc{V}$.
\end{lemma}
\begin{proof}
  Let $n_0 = n - (k - t)$. By Lemma~\ref{lemma:packing-vc-set} there is a
  collection $\mc{V}_0 \subset \{0, 1\}^{n_0}$ of $\half t$-separated
  vectors with cardinality $\log |\mc{V}_0| \ge c \cdot d \log
  \frac{n_0}{t}$, where $\vc(\mc{V}_0) = 2(d \wedge t)$ and $\lone{v} = t$
  for each $v \in \mc{V}_0$. Expand $\mc{V}_0$ by concatenating an
  appropriate vector of $1$s, defining $\mc{V} \defeq \{(v, \ones_{k - t})
  \mid v \in \mc{V}_0\} \subset \{0, 1\}^n$. This set satisfies the
  desiderata.
\end{proof}

We now now turn to Fano's method (Lemma~\ref{lemma:fano}) to lower bound the
probability of identifying the region $R$. Fix a $t \in \{1, \ldots, k\}$,
to be chosen later and let $\mc{V} \subset \{0, 1\}^n$ be the set
Lemma~\ref{lemma:packing-properties} specifies.  Identify $\RXs$ and $\mc{R}
= \{R \cap \{X_1, \ldots, X_n\} \mid R \in \RXs\}$ with $\mc{V}$ by the
construction~\eqref{eqn:R-region-construction}, so for $R, R' \in \mc{R}$ we
have
\begin{align*}
  \frac{t}{2}
  \indic{R \neq R'}
  \le |R \setdiff R'| \le
  2t.
\end{align*}
Then by Fano's inequality,
if $R\opt$ is chosen uniformly from $\mc{R}$,
then for any estimator $\Rest$,
\begin{align*}
  \P\left[|\Rest \setdiff R\opt| \ge \frac{t}{2}
    \mid X_1^n \right]
  \ge \P\left[ \Rest \neq R^\star \mid X_1^n \right]
  \ge
  \half - \frac{I(R; Z_1, \ldots, Z_n)}{
    c_0 (d \wedge t) \log \frac{n - k + t}{t}}
\end{align*}
where we used Lemma~\ref{lemma:packing-properties}. Leveraging
the naive bound
$I(R; Z_1^n) \le \max_{R, R'} \dkl{\P_R}{\P_{R'}}$ and
that for any $R, R' \in \mc{R}$ we have
$\dkl{\P_R}{\P_{R'}} = \frac{\mu^2 |R \setdiff R'|}{2 \sigma^2}
\le \frac{\mu^2 t}{\sigma^2}$, we obtain the intermediate
minimax bound
\begin{equation}
  \P
  \left[ |\Rest \setdiff R\opt| \ge \frac{t}{2}
    \mid 
   X_1^n \right]
  \ge \half - \frac{t \mu^2}{c_0 (d \wedge t) \sigma^2 \log \frac{n - k + t}{t}}.
  \label{eqn:intermediate-fano}
\end{equation}
Define the constant $c = \frac{c_0}{4}$. Then by
definition~\eqref{eqn:threshold-value} of the constant $T =
T(n, k, d, \mu, \sigma)$,
it is immediate that whenever $t \le T$ we have
$\frac{t \mu^2}{c_0 (d \wedge t) \sigma^2 \log \frac{n - k + t}{t}}
\le \frac{1}{4}$ and inequality~\eqref{eqn:intermediate-fano}
yields the first claim of the theorem.

For the SNR regime that $\frac{\mu^2}{\sigma^2}
\le c \log(n - k + 1)$, then, it remains to prove the
bounds~\eqref{eqn:threshold-bound} on $T$. We consider the three
regimes inequality~\eqref{eqn:threshold-bound} specifies.
\begin{enumerate}[1)]
\item \textbf{Low SNR}: when
  $\frac{\mu^2}{\sigma^2} \le \frac{c d \log(n/k)}{k}$.
  In this case, it is evident that we may
  take $t = k$ in the definition~\eqref{eqn:threshold-value} of
  $T$.
\item \textbf{Moderate SNR}:
  when $c \frac{d\log(n/k)}{k} < \frac{\mu^2}{\sigma^2} \le  c
  \log \frac{n-k+d}{d}$. Recalling the definition
  $\deffective = \frac{c \sigma^2}{\mu^2} d$,
  we consider
  two internal cases.
  First, if $n - k \le \deffective$, then we have
  $\log \frac{n - k}{\deffective} \le 0$, while
  we claim that $t = d$ satisfies the inequality~\eqref{eqn:threshold-value}
  defining $T$. Indeed,
  we have $d \le c d \frac{\sigma^2}{\mu^2} \log \frac{n - k + d}{d}$
  if and only if $c \log \frac{n - k + d}{d} \ge \frac{\mu^2}{\sigma^2}$,
  which we have assumed, and so
  $T \ge \max\{d, \deffective \log \frac{n - k}{\deffective}\}$.

  In the alternative case that $n - k > \deffective$, we can
  prove a similar equality. By definition~\eqref{eqn:threshold-value},
  we have $T(n,k,d,\mu,\sigma) \ge t$ whenever
  $d \le t \le k$ satisfies
  $\frac{t}{\log (1 + \frac{n-k}{t})} \le \deffective$, which,
  by the change of variables $u \defeq t / \deffective$,  is equivalent to
  \begin{align}
    \label{eqn:u-t-changer}
    \frac{u}{\log(1 + \frac{n-k}{\deffective} \frac{1}{u})} \le 1.
  \end{align}
  Now, for each $\lambda > 1$, the function
  $\varphi_\lambda(x) \defeq \frac{x}{\log(1 + \lambda/x)}$ is strictly
  increasing  on $(0, \infty)$,
  and we claim that $\varphi_\lambda^{-1}(1) \ge \half \log \lambda$:
  a direct computation yields
  \begin{align*}
    \varphi_\lambda\left(\half \log \lambda\right)
    = \frac{\log(\lambda)/2 }{\log \lambda + \log \left( \frac{1}{\lambda} + \frac{2}{ \log \lambda}\right)}
    & \le \frac{1}{2 + 2 \frac{\log\frac{2}{\log \lambda}}{\log \lambda}}
    \stackrel{(\star)}{\le} \frac{1}{2 - e^{-1}} < 1,
  \end{align*}
  where inequality~$(\star)$ follows because
  $\frac{2 \log\frac{2}{t}}{t}$ is minimized at $t = 2e$.
  In particular, the largest $u$ solving inequality~\eqref{eqn:u-t-changer} is
  at least $\half \log \frac{n - k}{\deffective}$, and so
  \begin{align*}
    T(n,k,d,\mu,\sigma) \ge \floor{\half
      \deffective \log \frac{n-k}{\deffective}}.
  \end{align*}
  As previously we likewise have $T \ge d$.
\item \textbf{Slightly High SNR}: when $c \log \frac{n - k +
  d}{d} \le \frac{\mu^2}{\sigma^2} \le c \log(n - k + 1)$.  In this case,
  any $t$ satisfying the inequality~\eqref{eqn:threshold-value} defining
  $T(n,k,d,\mu,\sigma)$ necessarily satisfies
  \begin{equation*}
    \log \left(1 + \frac{n - k}{t}\right)
    \ge \frac{t}{t \wedge d} \frac{1}{c} \frac{\mu^2}{\sigma^2},
  \end{equation*}
  and for $t \le d$, this occurs if and only if
  $1 + \frac{n - k}{t} \ge \exp(\frac{1}{c} \frac{\mu^2}{\sigma^2})$,
  that is,
  $t \le (n - k) (\exp(\frac{\mu^2}{c \sigma^2}) - 1)^{-1}$.
  In particular, it is sufficient that
  $t \le (n - k) \exp(-\frac{\mu^2}{c \sigma^2})$, and
  the condition that $\frac{\mu^2}{c \sigma^2} \ge \log \frac{n - k + d}{d}$
  guarantees that any such $t$ satisfies $t \le d$. This
  yields the
  final bound in inequality~\eqref{eqn:threshold-bound}.
\end{enumerate}

\subsubsection*{High SNR Regime}

When $\frac{\mu^2}{\sigma^2} \ge c \log(n - k + 1)$, which
we term the high SNR regime, we can apply
Assouad's method (Lemma~\ref{lemma:assouad}) to obtain a more direct lower
bound. We describe the construction of $\mc{V}$ first, which has some
parallels to Lemma~\ref{lemma:packing-vc-set}.  Let $\mc{W} = \{(0, 1), (1,
0)\}$ and $\mc{V}_0 = \mc{W}^d$, which has VC-dimension $d$ as in
Lemma~\ref{lemma:packing-vc-set}.  Expand $\mc{V}_0$ into $\mc{V} \subset
\{0, 1\}^n$ by concatenating the two vectors $\ones_{k - d}$ and $\zeros_{n
  - k - d}$ so that $v \in \mc{V}$ satisfies $\lone{v} = k$ and $\vc(\mc{V})
\le 2d$. Then by the construction of the regions $\mc{R}$ (see
Eq.~\eqref{eqn:R-region-construction}), we see that
$|R \setdiff R'| \ge 2$ for any pair $R \neq R' \in \mc{R}$, and
so by an
application of Assouad's method, we have
\begin{equation}
  \label{eqn:apply-assouad}
  \E_{R\opt}\left[\left|\what{R} \setdiff R\opt\right| \mid X_1^n
    \right] \ge \frac{d}{2}\left(1 -
  \max_{|R \setdiff R'| = 2} \tvnorm{\P_R - \P_{R'}}\right).
\end{equation}
A variant Pinsker inequality
for large KL-divergences~\cite[Lemma 2.6]{Tsybakov04} yields that
$\tvnorm{P - Q} \le 1 - \half \exp(-\dkl{P}{Q})$ for any distributions
$P, Q$. As a consequence, in inequality~\eqref{eqn:apply-assouad}
we have
$1 - \tvnorm{\P_R - \P_{R'}} \ge \half \exp(-\frac{\mu^2}{2 \sigma^2})$,
yielding the lower bound
\begin{equation*}
 \E_{R\opt}\left[|\what{R} \setdiff R\opt|
  \mid X_1^n
    \right] \ge \frac{d}{4}
  \exp\left(-\frac{\mu^2}{2 \sigma^2}\right).
\end{equation*}

\subsection{Proof of Theorem \ref{thm:two-step-estimation-error}}
\label{sec:proof-of-thm-two-step-estimation-error}
Even if they target two different practical goals (recovery vs.~refitting), the technical settings of~Theorems~\ref{thm:scan-recovery-error} and~\ref{thm:two-step-estimation-error} are the same, with $Y_i$ in the subpopulation model~\eqref{eq:gaussian-data-generation-refitting} taking the place of $Z_i$ in model~\eqref{eq:gaussian-data-generation}.

Reusing the same notation as in the proof of Theorem~\ref{thm:scan-recovery-error}, i.e., we have for $i \in [n]$,  $Y_i = \mu
\indic{i \in R^\star} + \sigma \noise_i$ where $\noise_i \simiid
\normal(0,1)$,  and recalling that $\mc{R} = \RXs \cap \{X_i\}_{i=1}^n$,
for each $R \in \mc{R}$ we
define the localized noise
$
  \noise(R) \defeq \frac{1}{\sigma \sqrt{|R|}} \sum_{i \in R} \noise_i.
$
Since we have
\begin{align*}
\hat \mu - \muopt =  \frac{\sigma\xi(\Rest)}{\sqrt{\Rest}} \ones_{\Rest} + \mu \left( \ones_{\Rest \setminus \Ropt} -  \ones_{\Ropt \setminus \Rest} \right),
\end{align*}
we can immediately observe that
\begin{align*}
\norm{\hat \mu - \muopt}_2^2 \le 2\left( \sigma^2 \xi(\Rest)^2 + \mu^2 \left|{\Rest}\setdiff{\Ropt}\right| \right).
\end{align*}

From the statement of Theorem 1, there exists a finite universal constant $C > 0$ such that the next three events occur each with probability at least $1-\delta$:
\begin{enumerate}[1)]
\item 
$   \dcor^2(\what{R}, R\opt)
  \le C \frac{\sigma^2}{\mu^2 k} 
  \left[ d \log \frac{n \mu^2}{d \sigma^2}
    + \log \frac{1}{\delta}\right] \defeq r(\mu)^2 \le \half,
$

\item  $
\sup_{R \in \Rcorset(r(\mu))}
  \left|\noise(R) - \noise(R\opt)\right|
  \le C \sqrt{d \min\left\{r(\mu)^2 \log \frac{n}{r(\mu)k}, 1\right\}
    + r(\mu)^2 \log \frac{1}{\delta}},
$

\item $| \noise(\Ropt) | \le \sqrt{2 \log \frac{2}{\delta}}$.
\end{enumerate}
On the intersection of these three events, which occurs with probability at least $1-3\delta$, we then have
$
\left|{\Rest}\setdiff{\Ropt}\right|
  \le 3 k r(\mu)^2,
$
which implies that, for some finite universal constant $C'>0$, the following inequality holds:
\begin{align*}
\norm{\hat \mu - \muopt}_2^2 &\le 6 k \mu^2 r(\mu)^2 + 4C^2 \sigma^2 r(\mu)^2 \log \frac{1}{\delta} + 4 C^2 d \sigma^2 \min\left\{r(\mu)^2 \log \frac{n}{r(\mu)k}, 1\right\} + 8  \sigma^2 \log \frac{2}{\delta} \\
&\le C'\sigma^2  \left( d \log \frac{n \mu^2}{d \sigma^2} + \log \frac{1}{\delta} \right),
\end{align*}
where we used the fact that $r(\mu)^2 \le \half$ and that $\frac{n \mu^2}{d \sigma^2}  \ge \frac{k \mu^2}{d \sigma^2} \ge 1/c > 1$.

\subsection{Proof of Theorem \ref{thm:estimation-error-lower-bound}}
\label{sec:proof-of-thm-estimation-error-lower-bound}
We use here the exact same construction as in Appendix~\ref{sec:proof-of-thm-recovery-error-lower-bound}, except that now we now use a different $\ell_2$-loss $L_2(\hat \mu, \mu) \defeq \norm{\hat \mu - \mu}_2^2$, versus the $\ell_0$-loss $L_0(\what R,  \Ropt) \defeq |\Rest \setdiff R|$ in the proof of Theorem~\ref{thm:recovery-error-lower-bound}.

The collection of regions $\RXs$ that we construct for a fixed $1\le t \le k$ in the proof of Theorem~\ref{thm:recovery-error-lower-bound}---which coincides with the collection $\mc{V}$ from Lemma~\ref{lemma:packing-properties}, by the
construction~\eqref{eqn:R-region-construction}---also satisfies for all $R, R' \in \mc{R}$, 
\begin{align*}
\norm{\mu \ones_R - \mu \ones_{R'} }_2^2 \ge t\mu^2/2,
\end{align*}
i.e\ $\{ \mu \ones_R \}_{R \in \mc{R}}$ is a $\mu\sqrt{\frac{t}{2}}$-packing in the $\ell_2$-norm. 
We can then use the following refinement of Fano's inequality.
\begin{lemma}[Fano's lemma, general loss]
  \label{lemma:fano-bis}
  Let $\mc{V}$ be an arbitrary set, $V \sim \uniform(\mc{V})$, and
  $\{Y_i\}_{i = 1}^n$ be random variables. 
  Let $\rho$ be a semimetric such that $\{ \mu_v \}_{v \in \mc{V}}$ form a $2\delta$-packing in the semimetric $\rho$, and $\Phi$ a convex function. Then for any
  estimator $\hat \mu(Y_1^n)$, we have
  \begin{equation*}
 \E \left[ \Phi\left(\rho(\what \mu(Y_1^n),  \mu_V) \right)\right] \ge \Phi(\delta) \left( 1 - \frac{I(V; Y_1, \ldots, Y_n) + \log 2}{
      \log|\mc{V}|} \right).
  \end{equation*}
\end{lemma}
The end of the proof then follows from the discussion on the value of the threshold $T(n,k,d,\mu,\sigma)$ according to the signal-to-noise ratio $\mu/\sigma$.

The application of Assouad's method in the high SNR regime uses the exact same hard region construction as in Appendix~\ref{sec:proof-of-thm-recovery-error-lower-bound},  but with the following refinement of Assouad's lemma.
\begin{lemma}[Assouad's lemma, general loss]
  \label{lemma:assouad-general-loss}
  Let distributions $P_v$ on a random variable $Y$ be indexed by vectors $v
  \in \{0, 1\}^d$, $\{ \mu_v \}_{v \in \{0,1\}^d} \subset \R^n$ a set of parameters, and define $\wb{P}_j = \frac{1}{2^{d-1}} \sum_{v : v_j =
    1} P_v$ and $\wb{P}_{-j} = \frac{1}{2^{d-1}} \sum_{v : v_j = 0}
  P_v$.
  Let $V \sim \uniform\{0,1\}^d$ and conditional on $V=v$, draw $Y  \sim P_v$.
  Let $\Phi$ be a convex loss function and $\rho$ a semimetric on $\R^n$ such there exist a function $\what v: \R^n \to \{0,1\}^d$ and $\delta > 0$ for which, for all $v \in \{0,1\}^d$ and all $\hat \mu \in \R^n$, 
  \begin{align*}
  \Phi( \rho(\hat \mu, \mu_v)) \ge 2\delta \lone{\what{v}(\hat \mu) - v}.
  \end{align*}
  Then for any estimator $\what{\mu}(Y)$,
  \begin{equation*}
  \E \left[ \Phi( \rho(\what \mu(Y), \mu_V)) \right] \ge \delta
    \sum_{j = 1}^d \left(1 - \tvnorm{\wb{P}_j - \wb{P}_{-j}}\right).
  \end{equation*}
\end{lemma}
In our case, for all $v \in \{ 0,1\}^d$,  we have $\mu_v = \mu (v_1, 1-v_1, \cdots, v_d, 1-v_d, \ones_{k-d}, \zeros_{n-k-d})$.
Define the function $\hat v(\theta) \defeq \left( \indic{\theta_{2i} > \theta_{2i-1}}\right)_{i=1}^d$,  so that for all $\what \mu \in \R^n$ and $v \in \{0,1\}^d$, we have
\begin{align*}
\ltwo{\what \mu - \mu_v}^2 \ge \frac{\mu^2}{2}\sum_{i=1}^d \indic{\hat v(\what \mu) \neq v_i},
\end{align*}
which yields by application of Assouad's lemma~\ref{lemma:assouad-general-loss} for general losses:
\begin{align*}
 \E \left[
 \norm{\hat \mu - \mu \ones_{R\opt}}_2^2
   \mid X_1^n \right] \ge \frac{d\mu^2}{4} \left(1 -
  \max_{|R \setdiff R'| = 2} \tvnorm{\P_R - \P_{R'}}\right).
\end{align*}
From the final discussion in the proof of Theorem~\ref{thm:recovery-error-lower-bound}, we therefore obtain the final lower bound
\begin{align*}
\E \left[
 \norm{\hat \mu - \mu \ones_{R\opt}}_2^2 \mid X_1^n
    \right] \ge \frac{d\mu^2}{8}
  \exp\left(-\frac{\mu^2}{2 \sigma^2}\right),
\end{align*}
valid for any estimator $\what \mu: \R^n \to \R^n$, whenever $R\opt$ is chosen uniformly at random in $\mc{R}$.

\subsection{Proof of Lemma \ref{lem:sure-tuned-mle-estimation-error}}
\label{sec:proof-of-sure-tuned-mle-estimation-error}

Broadly speaking, our strategy here as well as in the proof of Lemma
\ref{lem:hetero-sure-tuned-mle-estimation-error} is to
leverage~\citet[Thm.~1]{CauchoisAlDu21} (see
also~\citet[Sec.~4]{TibshiraniRo19}), which provides a relatively easy-to-use
characterization of the risk of SURE-tuned projection estimators.

We start by introducing a bit of notation, before translating the above theorem
into our notation here.  In what follows, we let $R$ denote either the region $R
\subset \mc{X}$ or its associated index set $\{ i \in [n] \mid X_i \in R \}$,
with the meaning clear from context.  Similarly, let $\mc{R}$ denote either the
collection of regions or the collection of associated index sets.  Now for $R
\in \mc{R}$, write $P_R = \ones_R \ones_R^T / |R|$.  With our notation in place,
and recalling the definitions in~\eqref{eq:sure-tuned-mle}, we may express the
family of projection estimators $\hat \mu_R \defeq \bar Y_R$ indexed by $R \in
\mc{R}$ as $\hat \mu_R = P_R Y$ for $R \in \mc{R}$, noting in particular that
the SURE-tuned estimator in \eqref{eq:sure-tuned-mle} $\hat \mu_\SURE = Y_{\hat
R} = P_{\hat R} Y$.

Below, we restate \citet[Thm.~1]{CauchoisAlDu21}---which we leverage in the
arguments that follow---making a few simplifications and translations into the
notation we use here.

\begin{theorem} \label{thm:generic-sure-bound} Assume the model
  \eqref{eq:gaussian-data-generation-refitting}.  Define the oracle risk
  \[
    \ropt \defeq \min_{R \in \mc{R}} \E \| \hat \mu_R - \muopt \|_2^2,
  \]
  let $\opnorm{P_R} \leq \hop$ for all $R \in \mc{R}$ with $\hop \geq 1$, and
  let $\log_+ z \defeq \max\{0, \log z\}$.  Then the SURE-tuned estimator $\hat
  \mu_\SURE$ in \eqref{eq:sure-tuned-mle} satisfies
    \begin{align*}
      \E \| \hat \mu_\SURE - \muopt \|_2^2 \lesssim \ropt + \hop \sigma^2 \log |\mc{R}| \cdot \Bigg( 1 + \log_+ \Bigg( \frac{\hop^2 \sigma^2 \log |\mc{R}|}{\ropt} \Bigg) \Bigg) + \sqrt{\ropt \sigma^2 \log |\mc{R}|}.
    \end{align*}
\end{theorem}

Now, by Sauer's lemma, we have that $\log |\mc{R}| \lesssim d \log(n/d)$ as
$\mc{R}$ is a VC-class with VC-dimension $d$.  Moreover, the oracle estimator
$\bar Y_{\Ropt}$ with knowledge of $\Ropt$ achieves risk
\[
  \ropt = \min \{ \sigma^2, k \mu^2 \}.
\]
Finally, for $R \in \mc{R}$, we have $\opnorm{P_R} = 1$.  Then under the
assumptions in the statement of the lemma, we have that $\ropt \gtrsim \sigma^2
\log |\mc{R}|$ so that invoking Theorem~\ref{thm:generic-sure-bound} and
simplifying immediately gives the result.

\subsection{Proof of Lemma \ref{lem:hetero-sure-tuned-mle-estimation-error}}
\label{sec:proof-of-hetero-sure-tuned-mle-estimation-error}

The proof follows the same strategy as the proof of Lemma
\ref{lem:sure-tuned-mle-estimation-error}, with just a few minor changes that we
enumerate now.  Here, we let $P_R \in \R^{n \times n}$ denote the projection map
onto $R$, meaning that for any $Z \in \R^n$ we have $(P_R Z)_i = Z_i$ if $i \in
R$ and 0 otherwise, for $i=1,\ldots,n$.  Then we may express the family of
projection estimators $\hat \mu_R = Y_R$ indexed by $R \in \mc{R}$ as $\hat
\mu_R = P_R Y$ for $R \in \mc{R}$, noting in particular that the SURE-tuned
estimator in \eqref{eq:hetero-sure-tuned-mle} $\hat \mu_\SURE = P_{\hat R} Y$.
It follows that $\opnorm{P_R} \leq 1$ and $\ropt = \min \{ k \sigma^2, \| \muopt
\|_2^2 \}$.  Putting together the pieces as before completes the proof.

\section{Technical proofs}

We collect several technical proofs in this appendix.

\subsection{Proof of Lemma~\ref{lemma:packing-vc-set}}
\label{sec:proof-packing-vc-set}

We prove the result in the case that $n$ and $k$ are divisible by $d$;
the general case requires a few tedious bookkeeping tweaks to address
edge effects and discretization errors.

We first consider the case that $k \log \frac{n}{k} \ge 2 \sqrt{2}(d + k)$.
Define $n_0 = n/d$ and $k_0 = k/d$. Consider the subset $\mc{W} \subset \{0,
1\}^{n_0}$ of vectors of $k_0$ consecutive $1$s and other entries $0$,
with ``wrapping'' at the boundaries, i.e.,
\begin{equation*}
  \mc{W} = \left\{
  \left[ \begin{matrix} \ones_{k_0} \\ \zeros_{n_0 - k_0}
    \end{matrix} \right],
  ~
  \left[ \begin{matrix} 0 \\ \ones_{k_0} \\ \zeros_{n_0 - k_0 - 1}
    \end{matrix} \right],
  \ldots,
  \left[ \begin{matrix} \zeros_{n_0 - k_0} \\ \ones_{k_0}
    \end{matrix} \right],
  \left[ \begin{matrix} 1 \\ \zeros_{n_0 - k_0} \\ \ones_{k_0 - 1}
    \end{matrix} \right],
  \ldots,
  \left[ \begin{matrix} \ones_{k_0 - 1} \\ \zeros_{n_0 - k_0} \\ 1
    \end{matrix} \right]
  \right\} \subset \{0, 1\}^{n_0},
\end{equation*}
and let $\mc{V} \defeq \mc{W}^d \subset \{0,1\}^n$ to be all concatenations of
$d$ vectors of $\mc{W}$. Then as $\vc(\mc{W}) = 2$, we see immediately that
$\vc(\mc{V}) = 2d$, and by construction, each $v \in \mc{V}$
immediately satisfies $\norm{v}_1 = k$, so we need only prove that the
packing number $M(\mc{V}, \norm{\cdot}_1, k/2)$ is at least $\exp(
d\log(n/k)/4)$

To prove this, we use the probabilistic method.  Let $W_i \simiid
\uniform(\mc{W})$ and $V = (W_1, \dots,
W_d)$, and fix an arbitrary $v \in \mc{V}$.  Then
setting $u^l = (v_{n_0 (l - 1) + 1}, \ldots, v_{n_0 l})$,
if we take $D_l = \lones{W_l - u^l}$ we have
$\lone{V - v} = k - \sum_{l = 1}^d D_l$, where the $D_l$ are i.i.d.\ with
\begin{align*}
  \P(D_1 = j) = \begin{cases}
    2d/n  &\text{ if } j\in \{1, \dots, k_0-1 \} \\
    d/n  &\text{ if } j=k_0 \\
    1 - 2k/n + d/n &\text{ if } j=0.
  \end{cases}
\end{align*}
For $\lambda \ge 0$, the moment generating function of $D_1$ then satisfies
\begin{align*}
  \E \left[ e^{\lambda D_1} \right] & =
  1 - \frac{2k}{n} + \frac{d}{n}
  + \frac{2d}{n} \sum_{j = 1}^{k_0 - 1} e^{\lambda j}
  + \frac{d}{n} e^{\lambda k_0} \\
  & = 1 - \frac{2k}{n} - \frac{d(e^{\lambda k_0} + 1)}{n}
  + \frac{2d}{n} \left(\frac{e^{\lambda (k_0 + 1)} - 1}{e^\lambda - 1}\right) \\
  & \le
  1 - \frac{2(k + d)}{n}
  + \frac{2d}{n \lambda} \left(e^{\lambda(k_0 + 1)} - 1\right),
\end{align*}
where use that $e^\lambda - 1 \ge \lambda$.
%
Substituting $\lambda = \frac{1}{k_0 + 1}\log\frac{n}{k}$
yields
\begin{align*}
  \E\left[e^{\lambda D_1}\right]
  \le 1 - \frac{2(k + d)}{n}
  + \frac{2d(k_0 + 1)}{n \log \frac{n}{k}} \left(\frac{n}{k} - 1 \right)
  & = 1 + \frac{2(k + d)}{n} \left[\frac{n}{k \log \frac{n}{k}} -
    1 - \frac{1}{\log \frac{n}{k}} \right] \\
  & \le \exp\left(\frac{2(k + d)}{k \log \frac{n}{k}}
  - \frac{2(k + d)}{n} \right).
\end{align*}
By a Chernoff bound and the shorthand $k_0 = k/d$,
we therefore obtain
\begin{align*}
  \P( \norm{V-v}_1 \le k/2)
  & = \P\left( \sum_{l=1}^d D_l \ge k/2 \right)
  \le \E \left[ e^{\lambda D_1} \right]^d e^{-\lambda k/2} \\
  & \le \exp\left( d \left(\frac{2(1+1/k_0)}{\log(n/k)}
  - \frac{2 (k_0 + 1)}{n}
  - \frac{\log(n/k)}{2(1+1/k_0)}  \right) \right) \\
  & <
  \exp\left( \frac{-d\log(n/k)}{4(1+1/k_0)} \right)
  \le  \exp\left( \frac{-d\log(n/k)}{8} \right),
\end{align*}
where the last line follows from the fact that $\frac{t}{2} - \frac{2}{t}
\ge \frac{t}{4}$ for all $t \ge 2\sqrt{2}$, where we have taken
$t = \frac{\log(n/k)}{1 + 1/k_0}$ and used the assumption that
$k \log \frac{n}{k} \ge 2\sqrt{2}(d + k)$.

We now apply the probabilistic method. Fix $M$ to be chosen, and let
$V^i$, $i = 1, \ldots, M$, be i.i.d.\ draws from the above distribution.
Then
\begin{equation*}
  \P(\min_{i \neq j} \lone{V^i - V^j} \le k/2)
  \le \frac{M^2}{2} \exp\left(-\frac{d \log(n/k)}{8}\right)
\end{equation*}
by a union bound,
and taking $M = \exp(\frac{d \log(n/k)}{16})$ gives that
$\lone{V^i - V^j} > \frac{k}{2}$ for all $i \neq j$ with probability
at least $\half$. Thus a packing as claimed in the lemma must exist
when $k \log \frac{n}{k} \ge 2 \sqrt{2}(d + k)$.

In the alternative case that $k \log \frac{n}{k} < 2 \sqrt{2}(d + k)$, we
must have $\log \frac{n}{k} < 4 \sqrt{2}$, or $k > e^{-4 \sqrt{2}}
n$. Then in analogy to the construction above, we consider the sets
$\mc{W} = \{(1, 0), (0, 1)\} \subset \{0, 1\}^2$, and let $\mc{V} =
\mc{W}^d \times \{(\ones_{k - d}, \zeros_{n - (k + d)}\} \subset \{0,
1\}^n$ be the concatenation of $d$ vectors of $\mc{W}$, padded with
appropriate $1$s and zeros. Then $\vc(\mc{V}) = 2 d$ as above, and each $v
\in \mc{V}$ satisfies $\lones{v} = k$. By an application of the
Gilbert-Varshamov bound, there is a collection of vectors $\{v^1, \ldots,
v^M\} \subset \mc{V}$ satisfying $\lone{v^i - v^j} \ge \frac{k}{2}$ with
cardinality $M \ge \exp(c d)$, where $c > 0$ is a numerical constant. As
$\log \frac{n}{k}$ is a numerical constant as well, this completes the
proof of the lemma.

\clearpage
\bibliography{bib}
\bibliographystyle{abbrvnat}

\end{document}